\documentclass[11pt,english]{article}
\usepackage{lmodern}

\usepackage[T1]{fontenc}
\usepackage[latin9]{inputenc}
\usepackage{geometry}
\usepackage{float}
\usepackage{units}
\usepackage{mathtools}
\usepackage{dsfont}
\usepackage{amsmath}
\usepackage{amsthm}
\usepackage{amssymb}

\makeatletter

\newcommand{\lyxmathsym}[1]{\ifmmode\begingroup\def\b@ld{bold}
  \text{\ifx\math@version\b@ld\bfseries\fi#1}\endgroup\else#1\fi}

\floatstyle{ruled}
\newfloat{algorithm}{tbp}{loa}
\providecommand{\algorithmname}{Algorithm}
\floatname{algorithm}{\protect\algorithmname}

\numberwithin{figure}{section}
\numberwithin{equation}{section}
\theoremstyle{plain}
\newtheorem*{question*}{\protect\questionname}
\theoremstyle{plain}
\newtheorem{thm}{\protect\theoremname}[section]
\theoremstyle{definition}
\newtheorem{defn}[thm]{\protect\definitionname}
\theoremstyle{plain}
\newtheorem{cor}[thm]{\protect\corollaryname}
\theoremstyle{plain}
\newtheorem{lem}[thm]{\protect\lemmaname}
\theoremstyle{plain}
\newtheorem{question}[thm]{\protect\questionname}
\theoremstyle{remark}
\newtheorem{rem}[thm]{\protect\remarkname}
\theoremstyle{plain}
\newtheorem{prop}[thm]{\protect\propositionname}
\theoremstyle{remark}
\newtheorem*{acknowledgement*}{\protect\acknowledgementname}

\usepackage{float,psfrag,epsfig,color,url,hyperref}
\hypersetup{
  colorlinks,
  linkcolor={red!50!black},
  citecolor={blue!50!black},
  urlcolor={blue!80!black}
}
\usepackage{comment,url,algorithm,algorithmic,graphicx,subcaption,relsize}
\usepackage{amssymb,amsfonts,amsmath,amsthm,amscd,dsfont,mathrsfs,mathtools,microtype,nicefrac,pifont}
\usepackage{upgreek}
\usepackage[dvipsnames]{xcolor}
\usepackage{epstopdf,bbm,mathtools,enumitem}
\usepackage[toc,page]{appendix}
\usepackage{dsfont,tikz}
\usepackage{aligned-overset}
\usepackage[mathscr]{euscript}
\allowdisplaybreaks
\usepackage{custom}
\makeatletter
\renewcommand{\paragraph}{%
  \@startsection{paragraph}{4}%
  {\z@}{1.25ex \@plus 1ex \@minus .2ex}{-1em}%
  {\normalfont\normalsize\bfseries}%
}
\makeatother
\numberwithin{equation}{subsection}

\makeatother

\usepackage{babel}
\providecommand{\acknowledgementname}{Acknowledgement}
\providecommand{\corollaryname}{Corollary}
\providecommand{\definitionname}{Definition}
\providecommand{\lemmaname}{Lemma}
\providecommand{\propositionname}{Proposition}
\providecommand{\questionname}{Question}
\providecommand{\remarkname}{Remark}
\providecommand{\theoremname}{Theorem}

\begin{document}
\def\balign#1\ealign{\begin{align}#1\end{align}}
\def\baligns#1\ealigns{\begin{align*}#1\end{align*}}
\def\balignat#1\ealign{\begin{alignat}#1\end{alignat}}
\def\balignats#1\ealigns{\begin{alignat*}#1\end{alignat*}}
\def\bitemize#1\eitemize{\begin{itemize}#1\end{itemize}}
\def\benumerate#1\eenumerate{\begin{enumerate}#1\end{enumerate}}

\newenvironment{talign*}
 {\let\displaystyle\textstyle\csname align*\endcsname}
 {\endalign}
\newenvironment{talign}
 {\let\displaystyle\textstyle\csname align\endcsname}
 {\endalign}

\def\balignst#1\ealignst{\begin{talign*}#1\end{talign*}}
\def\balignt#1\ealignt{\begin{talign}#1\end{talign}}

\let\originalleft\left
\let\originalright\right
\renewcommand{\left}{\mathopen{}\mathclose\bgroup\originalleft}
\renewcommand{\right}{\aftergroup\egroup\originalright}

\def\Gronwall{Gr\"onwall\xspace}
\def\Holder{H\"older\xspace}
\def\Ito{It\^o\xspace}
\def\Nystrom{Nystr\"om\xspace}
\def\Schatten{Sch\"atten\xspace}
\def\Matern{Mat\'ern\xspace}

\def\tinycitep*#1{{\tiny\citep*{#1}}}
\def\tinycitealt*#1{{\tiny\citealt*{#1}}}
\def\tinycite*#1{{\tiny\cite*{#1}}}
\def\smallcitep*#1{{\scriptsize\citep*{#1}}}
\def\smallcitealt*#1{{\scriptsize\citealt*{#1}}}
\def\smallcite*#1{{\scriptsize\cite*{#1}}}

\def\blue#1{\textcolor{blue}{{#1}}}
\def\green#1{\textcolor{green}{{#1}}}
\def\orange#1{\textcolor{orange}{{#1}}}
\def\purple#1{\textcolor{purple}{{#1}}}
\def\red#1{\textcolor{red}{{#1}}}
\def\teal#1{\textcolor{teal}{{#1}}}

\def\mbi#1{\boldsymbol{#1}} 
\def\mbf#1{\mathbf{#1}}
\def\mrm#1{\mathrm{#1}}
\def\tbf#1{\textbf{#1}}
\def\tsc#1{\textsc{#1}}

\def\mbiA{\mbi{A}}
\def\mbiB{\mbi{B}}
\def\mbiC{\mbi{C}}
\def\mbiDelta{\mbi{\Delta}}
\def\mbif{\mbi{f}}
\def\mbiF{\mbi{F}}
\def\mbih{\mbi{g}}
\def\mbiG{\mbi{G}}
\def\mbih{\mbi{h}}
\def\mbiH{\mbi{H}}
\def\mbiI{\mbi{I}}
\def\mbim{\mbi{m}}
\def\mbiP{\mbi{P}}
\def\mbiQ{\mbi{Q}}
\def\mbiR{\mbi{R}}
\def\mbiv{\mbi{v}}
\def\mbiV{\mbi{V}}
\def\mbiW{\mbi{W}}
\def\mbiX{\mbi{X}}
\def\mbiY{\mbi{Y}}
\def\mbiZ{\mbi{Z}}

\def\textsum{{\textstyle\sum}} 
\def\textprod{{\textstyle\prod}} 
\def\textbigcap{{\textstyle\bigcap}} 
\def\textbigcup{{\textstyle\bigcup}} 

\def\reals{\mathbb{R}} 
\def\integers{\mathbb{Z}} 
\def\rationals{\mathbb{Q}} 
\def\naturals{\mathbb{N}} 
\def\complex{\mathbb{C}} 

\def\what#1{\widehat{#1}}

\def\twovec#1#2{\left[\begin{array}{c}{#1} \\ {#2}\end{array}\right]}
\def\threevec#1#2#3{\left[\begin{array}{c}{#1} \\ {#2} \\ {#3} \end{array}\right]}
\def\nvec#1#2#3{\left[\begin{array}{c}{#1} \\ {#2} \\ \vdots \\ {#3}\end{array}\right]} 

\def\maxeig#1{\lambda_{\mathrm{max}}\left({#1}\right)}
\def\mineig#1{\lambda_{\mathrm{min}}\left({#1}\right)}

\def\Re{\operatorname{Re}} 
\def\indic#1{\mbb{I}\left[{#1}\right]} 
\def\logarg#1{\log\left({#1}\right)} 
\def\polylog{\operatorname{polylog}}
\def\maxarg#1{\max\left({#1}\right)} 
\def\minarg#1{\min\left({#1}\right)} 
\def\Earg#1{\E\left[{#1}\right]}
\def\Esub#1{\E_{#1}}
\def\Esubarg#1#2{\E_{#1}\left[{#2}\right]}
\def\bigO#1{\mathcal{O}\left(#1\right)} 
\def\littleO#1{o(#1)} 
\def\P{\mbb{P}} 
\def\Parg#1{\P\left({#1}\right)}
\def\Psubarg#1#2{\P_{#1}\left[{#2}\right]}
\def\Trarg#1{\Tr\left[{#1}\right]} 
\def\trarg#1{\tr\left[{#1}\right]} 
\def\Var{\mrm{Var}} 
\def\Vararg#1{\Var\left[{#1}\right]}
\def\Varsubarg#1#2{\Var_{#1}\left[{#2}\right]}
\def\Cov{\mrm{Cov}} 
\def\Covarg#1{\Cov\left[{#1}\right]}
\def\Covsubarg#1#2{\Cov_{#1}\left[{#2}\right]}
\def\Corr{\mrm{Corr}} 
\def\Corrarg#1{\Corr\left[{#1}\right]}
\def\Corrsubarg#1#2{\Corr_{#1}\left[{#2}\right]}
\newcommand{\info}[3][{}]{\mathbb{I}_{#1}\left({#2};{#3}\right)} 
\newcommand{\staticexp}[1]{\operatorname{exp}(#1)} 
\newcommand{\loglihood}[0]{\mathcal{L}} 


\providecommand{\arccos}{\mathop\mathrm{arccos}}
\providecommand{\dom}{\mathop\mathrm{dom}}
\providecommand{\diag}{\mathop\mathrm{diag}}
\providecommand{\tr}{\mathop\mathrm{tr}}
\providecommand{\card}{\mathop\mathrm{card}}
\providecommand{\sign}{\mathop\mathrm{sign}}
\providecommand{\conv}{\mathop\mathrm{conv}} 
\def\rank#1{\mathrm{rank}({#1})}
\def\supp#1{\mathrm{supp}({#1})}

\providecommand{\minimize}{\mathop\mathrm{minimize}}
\providecommand{\maximize}{\mathop\mathrm{maximize}}
\providecommand{\subjectto}{\mathop\mathrm{subject\;to}}

\def\openright#1#2{\left[{#1}, {#2}\right)}

\ifdefined\nonewproofenvironments\else
\ifdefined\ispres\else
 
\fi
\makeatletter
\@addtoreset{equation}{section}
\makeatother
\def\theequation{\thesection.\arabic{equation}}

\newcommand{\cmark}{\ding{51}}

\newcommand{\xmark}{\ding{55}}

\newcommand{\eq}[1]{\begin{align}#1\end{align}}
\newcommand{\eqn}[1]{\begin{align*}#1\end{align*}}
\renewcommand{\Pr}[1]{\mathbb{P}\left( #1 \right)}
\newcommand{\Ex}[1]{\mathbb{E}\left[#1\right]}

\newcommand{\matt}[1]{{\textcolor{Maroon}{[Matt: #1]}}}
\newcommand{\kook}[1]{{\textcolor{blue}{[Kook: #1]}}}
\definecolor{OliveGreen}{rgb}{0,0.6,0}
\newcommand{\sv}[1]{{\textcolor{OliveGreen}{[Santosh: #1]}}}

\global\long\def\on#1{\operatorname{#1}}%

\global\long\def\bw{\mathsf{Ball\ walk}}%
\global\long\def\sw{\mathsf{Speedy\ walk}}%
\global\long\def\gw{\mathsf{Gaussian\ walk}}%
\global\long\def\ps{\mathsf{Proximal\ sampler}}%

\global\long\def\har{\mathsf{Hit\text{-}and\text{-}Run}}%
\global\long\def\gc{\mathsf{Gaussian\ cooling}}%
\global\long\def\ino{\mathsf{\mathsf{In\text{-}and\text{-}Out}}}%
\global\long\def\tgc{\mathsf{Tilted\ Gaussian\ cooling}}%
\global\long\def\PS{\mathsf{PS}}%
\global\long\def\psunif{\mathsf{PS}_{\textup{unif}}}%
\global\long\def\psexp{\mathsf{PS}_{\textup{exp}}}%
\global\long\def\psann{\mathsf{PS}_{\textup{ann}}}%
\global\long\def\psgauss{\mathsf{PS}_{\textup{Gauss}}}%
\global\long\def\eval{\mathsf{Eval}}%
\global\long\def\mem{\mathsf{Mem}}%

\global\long\def\O{\mathcal{O}}%
\global\long\def\Otilde{\widetilde{\mathcal{O}}}%
\global\long\def\Omtilde{\widetilde{\Omega}}%

\global\long\def\E{\mathbb{E}}%
\global\long\def\Z{\mathbb{Z}}%
\global\long\def\P{\mathbb{P}}%
\global\long\def\N{\mathbb{N}}%

\global\long\def\R{\mathbb{R}}%
\global\long\def\Rd{\mathbb{R}^{d}}%
\global\long\def\Rdd{\mathbb{R}^{d\times d}}%
\global\long\def\Rn{\mathbb{R}^{n}}%
\global\long\def\Rnn{\mathbb{R}^{n\times n}}%

\global\long\def\psd{\mathbb{S}_{+}^{d}}%
\global\long\def\pd{\mathbb{S}_{++}^{d}}%

\global\long\def\defeq{\stackrel{\mathrm{{\scriptscriptstyle def}}}{=}}%

\global\long\def\veps{\varepsilon}%
\global\long\def\lda{\lambda}%
\global\long\def\vphi{\varphi}%
\global\long\def\K{\mathcal{K}}%

\global\long\def\half{\frac{1}{2}}%
\global\long\def\nhalf{\nicefrac{1}{2}}%
\global\long\def\texthalf{{\textstyle \frac{1}{2}}}%
\global\long\def\ltwo{L^{2}}%

\global\long\def\ind{\mathds{1}}%
\global\long\def\op{\mathsf{op}}%
\global\long\def\ch{\mathsf{Ch}}%
\global\long\def\kls{\mathsf{KLS}}%
\global\long\def\ts{\mathsf{Ts}}%
\global\long\def\hs{\textup{HS}}%

\global\long\def\cpi{C_{\mathsf{PI}}}%
\global\long\def\clsi{C_{\mathsf{LSI}}}%
\global\long\def\cch{C_{\mathsf{Ch}}}%
\global\long\def\clch{C_{\mathsf{logCh}}}%
\global\long\def\cexp{C_{\mathsf{exp}}}%
\global\long\def\cgauss{C_{\mathsf{Gauss}}}%

\global\long\def\chooses#1#2{_{#1}C_{#2}}%

\global\long\def\vol{\on{vol}}%

\global\long\def\law{\on{law}}%

\global\long\def\tr{\on{tr}}%

\global\long\def\diag{\on{diag}}%

\global\long\def\diam{\on{diam}}%

\global\long\def\poly{\on{poly}}%

\global\long\def\polylog{\on{polylog}}%

\global\long\def\Diag{\on{Diag}}%

\global\long\def\inter{\on{int}}%

\global\long\def\esssup{\on{ess\,sup}}%

\global\long\def\proj{\on{Proj}}%

\global\long\def\e{\mathrm{e}}%

\global\long\def\id{\mathrm{id}}%

\global\long\def\spanning{\on{span}}%

\global\long\def\rows{\on{row}}%

\global\long\def\cols{\on{col}}%

\global\long\def\rank{\on{rank}}%

\global\long\def\T{\mathsf{T}}%

\global\long\def\bs#1{\boldsymbol{#1}}%

\global\long\def\eu#1{\EuScript{#1}}%

\global\long\def\mb#1{\mathbf{#1}}%

\global\long\def\mbb#1{\mathbb{#1}}%

\global\long\def\mc#1{\mathcal{#1}}%

\global\long\def\mf#1{\mathfrak{#1}}%

\global\long\def\ms#1{\mathscr{#1}}%

\global\long\def\mss#1{\mathsf{#1}}%

\global\long\def\msf#1{\mathsf{#1}}%

\global\long\def\textint{{\textstyle \int}}%
\global\long\def\Dd{\mathrm{D}}%
\global\long\def\D{\mathrm{d}}%
\global\long\def\grad{\nabla}%
 
\global\long\def\hess{\nabla^{2}}%
 
\global\long\def\lapl{\triangle}%
 
\global\long\def\deriv#1#2{\frac{\D#1}{\D#2}}%
 
\global\long\def\pderiv#1#2{\frac{\partial#1}{\partial#2}}%
 
\global\long\def\de{\partial}%
\global\long\def\lagrange{\mathcal{L}}%
\global\long\def\Div{\on{div}}%

\global\long\def\Gsn{\mathcal{N}}%
 
\global\long\def\BeP{\textnormal{BeP}}%
 
\global\long\def\Ber{\textnormal{Ber}}%
 
\global\long\def\Bern{\textnormal{Bern}}%
 
\global\long\def\Bet{\textnormal{Beta}}%
 
\global\long\def\Beta{\textnormal{Beta}}%
 
\global\long\def\Bin{\textnormal{Bin}}%
 
\global\long\def\BP{\textnormal{BP}}%
 
\global\long\def\Dir{\textnormal{Dir}}%
 
\global\long\def\DP{\textnormal{DP}}%
 
\global\long\def\Expo{\textnormal{Expo}}%
 
\global\long\def\Gam{\textnormal{Gamma}}%
 
\global\long\def\GEM{\textnormal{GEM}}%
 
\global\long\def\HypGeo{\textnormal{HypGeo}}%
 
\global\long\def\Mult{\textnormal{Mult}}%
 
\global\long\def\NegMult{\textnormal{NegMult}}%
 
\global\long\def\Poi{\textnormal{Poi}}%
 
\global\long\def\Pois{\textnormal{Pois}}%
 
\global\long\def\Unif{\textnormal{Unif}}%

\global\long\def\bpar#1{\bigl(#1\bigr)}%
\global\long\def\Bpar#1{\Bigl(#1\Bigr)}%

\global\long\def\abs#1{|#1|}%
\global\long\def\babs#1{\bigl|#1\bigr|}%
\global\long\def\Babs#1{\Bigl|#1\Bigr|}%

\global\long\def\snorm#1{\|#1\|}%
\global\long\def\bnorm#1{\bigl\Vert#1\bigr\Vert}%
\global\long\def\Bnorm#1{\Bigl\Vert#1\Bigr\Vert}%

\global\long\def\sbrack#1{[#1]}%
\global\long\def\bbrack#1{\bigl[#1\bigr]}%
\global\long\def\Bbrack#1{\Bigl[#1\Bigr]}%

\global\long\def\sbrace#1{\{#1\}}%
\global\long\def\bbrace#1{\bigl\{#1\bigr\}}%
\global\long\def\Bbrace#1{\Bigl\{#1\Bigr\}}%

\global\long\def\Abs#1{\left\lvert #1\right\rvert }%
\global\long\def\Par#1{\left(#1\right)}%
\global\long\def\Brack#1{\left[#1\right]}%
\global\long\def\Brace#1{\left\{  #1\right\}  }%

\global\long\def\inner#1{\langle#1\rangle}%
 
\global\long\def\binner#1#2{\left\langle {#1},{#2}\right\rangle }%

\global\long\def\norm#1{\lVert#1\rVert}%
\global\long\def\onenorm#1{\norm{#1}_{1}}%
\global\long\def\twonorm#1{\norm{#1}_{2}}%
\global\long\def\infnorm#1{\norm{#1}_{\infty}}%
\global\long\def\fronorm#1{\norm{#1}_{\text{F}}}%
\global\long\def\nucnorm#1{\norm{#1}_{*}}%
\global\long\def\staticnorm#1{\|#1\|}%
\global\long\def\statictwonorm#1{\staticnorm{#1}_{2}}%

\global\long\def\mmid{\mathbin{\|}}%

\global\long\def\otilde#1{\widetilde{\mc O}(#1)}%
\global\long\def\wtilde{\widetilde{W}}%
\global\long\def\wt#1{\widetilde{#1}}%

\global\long\def\KL{\msf{KL}}%
\global\long\def\dtv{d_{\textrm{\textup{TV}}}}%
\global\long\def\FI{\msf{FI}}%
\global\long\def\tv{\msf{TV}}%
\global\long\def\ent{\msf{Ent}}%
\global\long\def\TV{\msf{TV}}%

\global\long\def\cov{\on{cov}}%
\global\long\def\var{\on{var}}%

\global\long\def\cred#1{\textcolor{red}{#1}}%
\global\long\def\cblue#1{\textcolor{blue}{#1}}%
\global\long\def\cgreen#1{\textcolor{green}{#1}}%
\global\long\def\ccyan#1{\textcolor{cyan}{#1}}%

\global\long\def\iff{\Leftrightarrow}%
 
\global\long\def\textfrac#1#2{{\textstyle \frac{#1}{#2}}}%

\title{Faster Logconcave Sampling from a Cold Start in High Dimension\date{}\author{Yunbum Kook\\ Georgia Tech\\  \texttt{yb.kook@gatech.edu} \and Santosh S. Vempala\\ Georgia Tech\\ \texttt{vempala@gatech.edu}}}
\maketitle
\begin{abstract}
We present a faster algorithm to generate a warm start for sampling
an arbitrary logconcave density specified by an evaluation oracle,
leading to the first sub-cubic sampling algorithms for inputs in (near-)isotropic
position. A long line of prior work incurred a warm-start penalty
of at least linear in the dimension, hitting a cubic barrier, even
for the special case of uniform sampling from convex bodies.

Our improvement relies on two key ingredients of independent interest.
(1) We show how to sample given a warm start in weaker notions of
distance, in particular $q$-R\'enyi divergence for $q=\widetilde{\mathcal{O}}(1)$,
whereas previous analyses required stringent $\infty$-R\'enyi divergence
(with the exception of Hit-and-Run, whose known mixing time is higher).
This marks the first improvement in the required warmness since Lov\'asz
and Simonovits (1991). (2) We refine and generalize the log-Sobolev
inequality of Lee and Vempala (2018), originally established for isotropic
logconcave distributions in terms of the diameter of the support,
to logconcave distributions in terms of a geometric average of the
support diameter and the largest eigenvalue of the covariance matrix.
\end{abstract}
\tableofcontents{}

\setcounter{page}{0}
\thispagestyle{empty}\newpage{}

\section{Introduction}

Sampling from high-dimensional distributions is a fundamental problem
that arises frequently in differential privacy \cite{MT07mechanism,HT10geometry,Mironov17renyi},
scientific computing \cite{CV16practical,KLSV22sampling}, and system
biology \cite{TSF13community,HCTFV17chrr}. It has broad applications,
ranging from volume estimation of convex bodies and integration of
high-dimensional functions to Bayesian inference and stochastic optimization.

In this paper, we focus on the problem of sampling arbitrary logconcave
functions given access to an evaluation oracle. An important special
case of this problem, which captures many of its challenges and has
provided the gateway to efficient algorithms, is uniform sampling
from a convex body given by a membership oracle (Definition~\ref{def:welldefined-membership}).
We introduce new ideas and analyses first for uniform sampling, and
later extend them to general logconcave sampling. To put our results
in context, we begin with a brief account of the history of sampling
and the crucial role of warm starts.

\paragraph{A brief history of sampling and warmth.}

In high-dimensional regimes, Markov chain Monte Carlo (MCMC) (or \emph{random
walk}) serves as the primary general approach. At its core, MCMC leverages
a transition kernel designed to explore a space in a way that, after
sufficiently many iterations, the resulting sample has distribution
close to a desired one. The efficiency of an MCMC algorithm is quantified
by its query complexity: given target accuracy $\veps>0$, query complexity
measures the number of oracle queries required to generate a sample
whose law is $\veps$-close to the target distribution, measured by
metrics such as total variation ($\tv$) distance or $\chi^{2}$-divergence
(see \S\ref{sec:prelim} for definitions).

Dyer, Frieze, and Kannan \cite{DFK91random} proposed the first algorithm
for uniformly sampling a convex body $\K\subset\Rn$ in polynomial
time, the central ingredient in the first polynomial time (randomized)
algorithm for volume estimation. They used a walk on a sufficiently
fine grid in a suitably ``smoothened'' body and bounded its \emph{conductance}
(i.e., conditional escape probability) from below by an inverse polynomial
in the dimension $n$ and the Euclidean diameter $D$ of the body.

Following this, Lov\'asz and Simonovits \cite{lovasz90compute,LS90mixing}
introduced the $\bw$, where each proposed step is a uniform random
point in a ball of fixed radius around the current point, and is accepted
if the proposal is in the body. They developed a general framework
for sampling in continuous domains, bounded the conductance via \emph{isoperimetry}
of the target distribution, and introduced \emph{$s$-conductance},
which lower-bounds the conductance of the Markov chain for subsets
of measure at least $s>0$. Since the $\bw$ can have bad starts (e.g.,
near the corner of a convex body), this notion led to a mixing rate
of the $\bw$ from a \emph{warm start} (i.e., $M=\esssup_{\pi_{0}}\D\pi_{0}/\D\pi\lesssim1$
for initial $\pi_{0}$ and target $\pi$), where the mixing rate depends
polynomially on $M$. In words, the warm start refers to an initial
distribution already close to a target in this sense. Hereafter, we
use the $q$-R\'enyi divergence ($\eu R_{q}$) for $q\in(1,\infty]$
to quantify the warmness (see \S\ref{sec:prelim}; recall the monotonicity
of $\eu R_{q}$ in $q$, so $\eu R_{\infty}$-warm is stronger than
$\eu R_{2}$-warm for instance).

To achieve such an $M$-warm start, they used a basic annealing procedure---starting
from a unit ball contained in $\K$ and gradually growing the ball
truncated to $\K$; each distribution in the sequence provides an
$\O(1)$-warm start to the next one. In their analysis, the dependence
on the distance $\veps$ to the target distribution is also inverse
polynomial, $\poly\nicefrac{1}{\veps}$. We note that in their paper
and in all subsequent work till very recently \cite{KZ25Renyi,KV25sampling},
even though the initial distribution has bounded pointwise distance
(i.e., in $\eu R_{\infty}$), the final distribution has weaker guarantees
(in $\chi^{2}$ or $\tv$). 

The next major improvement came in \cite{KLS97random} which introduced
the \emph{$\sw$}, an analytical tool to overcome some challenges
of the $\bw$. It directly uses conductance (rather than $s$-conductance),
so its mixing rate achieves $\log\nicefrac{M}{\veps}$, log-dependence
on both the warmness and the target distance. However, its implementation
uses rejection sampling, so its query complexity scales as $\poly M$.
They used the annealing approach in earlier work to generate a warm
start.

Nevertheless, from an $\O(1)$-warm start in $\eu R_{\infty}$, they
achieved an improved $n^{3}$-complexity of sampling an isotropic
convex body. In a companion paper \cite{KLS95isop}, directly motivated
by the mixing rate of the $\bw$ in the same setup, they posited the
Kannan-Lov\'asz-Simonovits (KLS) conjecture, which has since become
a central and unifying problem for asymptotic convex geometry and
functional analysis. Following a long line of progress, the current
best bound is $\O(\log^{1/2}n)$ \cite{Klartag23log}, which in particular
implies a mixing rate of $\Otilde(n^{2})$ for the $\bw$ in an isotropic
convex body from an $\O(1)$-warm start.

So far, the mixing rate was analyzed based on Cheeger or Poincar\'e
inequalities, which have, at best, logarithmic dependence on the warmness
parameter. Kannan, Lov\'asz, and Montenegro \cite{KL99average,KLM06blocking}
introduced \emph{blocking conductance} to improve the analysis of
the $\sw$. They established a \emph{logarithmic Cheeger inequality}
for convex bodies, improving the mixing rate to $n^{3}$ for isotropic
convex bodies \emph{without} a warm start. While this does not yield
an algorithmic improvement (as the $\sw$ has to ultimately be implemented
with the $\bw$ and rejection sampling), this framework has turned
out to be useful for the analysis of Markov chains.

In a parallel line of work, Lov\'asz and Vempala \cite{LV06hit}
analyzed the $\har$ sampler \cite{Smith84HAR}, showing a mixing
rate of $n^{2}D^{2}$ for a convex body with diameter $D$ under a
weaker warm start, namely in $\eu R_{2}$-divergence (i.e., $M_{2}=\norm{\D\pi_{0}/\D\pi}_{L^{2}(\pi)}$)
with a logarithmic dependence on $M_{2}$. This was later extended
by \cite{LV06fast} to logconcave distributions and has remained the
only rapidly mixing Markov chain from any start for general logconcave
distributions. It played a crucial role in an improved volume algorithm
\cite{LV06simulated}. However, unlike the $\bw$, the KLS conjecture
does not imply a faster mixing rate for $\har$ even in isotropic
convex bodies.

In 2015, using the $\bw$, Cousins and Vempala \cite{CV15bypass}
showed how to sample a well-rounded convex body (a condition weaker
than isotropy) with $n^{3}$ queries via a new annealing scheme called
$\gc$, without any warm start assumption. It uses a sequence of Gaussians
of increasing variance restricted to the convex body, each providing
a warm start to the next. This is the current state-of-the-art and
was recently extended to arbitrary logconcave densities for sampling
(and integration) while preserving the same complexity as the special
case of convex bodies \cite{KV25sampling}. Going below $n^{3}$ has
remained an open problem, even after the near-resolution of the KLS
conjecture.

There is, however, one other development that raises the possibility
of going below cubic time. In 2018, Lee and Vempala \cite{LV24eldan}
established logarithmic Sobolev inequality for isotropic logconcave
distributions, improving the log-Sobolev isoperimetric constant from
$\O(D^{2})$ to tight $\O(D)$. This leads to an improved $n^{2.5}$-mixing
of the $\sw$ in an isotropic convex body from any start. Again, since
the $\sw$ is a conceptual process, it does not give an algorithmic
improvement.

The leads us to the main motivating question of this paper. 
\begin{question*}
Is there an algorithm of sub-cubic complexity to sample a convex body
in near-isotropic position? Can the sub-cubic complexity of the $\sw$
be realized algorithmically? 
\end{question*}
Besides uniform sampling (and general logconcave sampling), a special
case of considerable interest is sampling a truncated Gaussian, i.e.,
a standard Gaussian distribution restricted to an arbitrary convex
body containing the origin. 

The above question opens up a set of concrete challenges, which are
also interesting in their own right. We discuss them in more detail
below. To give a quick preview, we mention briefly that our main result
is an affirmative answer to this main motivating question.

\paragraph{Uniform sampling from a warm start.}

 \cite{LS90mixing} connected the convergence rate of the $\bw$
to the \emph{Cheeger constant} $\cch(\pi)$ defined as follows: for
a probability measure $\pi$ over $\Rn$ and Euclidean distance $d$,
\[
C_{\ch,d}(\pi):=\inf_{S:\pi(S)\le\frac{1}{2}}\frac{\pi(\de S)}{\pi(S)}\,,
\]
where $S$ is any open subset of $\Rn$ with smooth boundary, $\pi(\de S):=\liminf_{\veps\downarrow0}\frac{\pi(S_{\veps})-\pi(S)}{\veps}$,
and $S_{\veps}=\{x:d(x,S)\leq\veps\}$. The $\bw$ requires $\Otilde(n^{2}C_{\ch,d}^{-2}(\pi))$
queries in expectation to reach a distribution that is $\veps$-close
to the target uniform distribution $\pi$ over a convex body in $\tv$-distance,
starting from a $\eu R_{\infty}$-warm start. We use the notation
$\bw:\eu R_{\infty}\to\tv$ to indicate the required input warmness
and output metric. Significant breakthroughs have refined the Cheeger
bound \cite{PW60optimal,KLS95isop,eldan13thin,LV24eldan,Chen21almost,KL22Bourgain,Klartag23log},
leading to the current best bound $C_{\ch,d}^{-2}(\pi)\lesssim\norm{\cov\pi}\log n$
by Klartag \cite{Klartag23log}, where $\norm{\cov\pi}$ denotes the
largest eigenvalue of the covariance matrix of $\pi$.

A different sampler $\har:\eu R_{2}\to\eu R_{2}$, studied in \cite{Lovasz99hit,LV06hit},
was shown to mix using $n^{2}D^{2}C_{\ch,d_{\K}}^{-2}(\pi)$ queries,
where $d_{\K}$ is the cross-ratio distance. Unlike the $\bw$, since
$C_{\ch,d_{\K}}\leq1$ is tight, progress on the KLS conjecture does
not lead to an improvement in its complexity. Since $\norm{\cov\pi}<D^{2}$,
the $\bw$ has a better bound from a $\eu R_{\infty}$-warm start.

Recently, Kook, Vempala, and Zhang~\cite{KVZ24INO} introduced $\ino$
(equivalently, the proximal sampler \cite{LST21structured} for uniform
distributions, referred to here as $\psunif$), a refinement of the
$\bw$. Its convergence rate can be directly related to the \emph{Poincar\'e
constant} for a target $\pi$:
\begin{defn}
A probability measure $\pi$ on $\R^{n}$ satisfies a \emph{Poincar\'e
inequality }with constant $C_{\msf{PI}}(\pi)$ if for any locally
Lipschitz function $f:\R^{n}\to\R$,
\begin{equation}
\var_{\pi}f\leq\cpi(\pi)\,\E_{\pi}[\norm{\nabla f}^{2}]\,,\tag{\ensuremath{\msf{PI}}}\label{eq:pi}
\end{equation}
where $\var_{\pi}f=\E_{\pi}[\abs{f-\E_{\pi}f}^{2}]$.
\end{defn}

In general, the Poincar\'e inequality is implied by the log-Sobolev
inequality.
\begin{defn}
A probability measure $\pi$ on $\R^{n}$ satisfies a \emph{logarithmic
Sobolev inequality} with constant $C_{\msf{LSI}}(\pi)$ if for any
locally Lipschitz function $f:\Rn\to\R$,
\begin{equation}
\ent_{\pi}(f^{2})\leq2\clsi(\pi)\,\E_{\pi}[\norm{\nabla f}^{2}]\,,\tag{\ensuremath{\msf{LSI}}}\label{eq:lsi}
\end{equation}
where $\ent_{\pi}(f^{2}):=\E_{\pi}[f^{2}\log f^{2}]-\E_{\pi}[f^{2}]\log\E_{\pi}[f^{2}]$. 
\end{defn}

$\psunif$ achieves a query complexity of $qn^{2}\cpi(\pi)$ to generate
a sample with $\eu R_{q}$-guarantee, starting from a $\eu R_{\infty}$-warmness
(i.e., $\psunif:\eu R_{\infty}\to\eu R_{q}$ under \eqref{eq:pi}).
This guarantee can be viewed as a right generalization of the $\bw$
to $\eu R_{q}$-divergence, since for logconcave probability measures
$\pi$,
\[
\frac{1}{4}\leq\frac{C_{\ch,d}^{-2}(\pi)}{\cpi(\pi)}\leq9\,,
\]
where the first inequality is due to Cheeger \cite{Chee70lower},
and the second follows from the Buser--Ledoux inequality \cite{Buser82iso,Led04spectral}.
Then, Kook and Zhang \cite{KZ25Renyi} provided a query complexity
of $n^{2}\clsi(\pi)$ for $\eu R_{\infty}$-guarantees from a $\eu R_{\infty}$-warmness
(i.e., $\psunif:\eu R_{\infty}\to\eu R_{\infty}$ under \eqref{eq:lsi}),
where $\clsi\lesssim D^{2}$ holds for any logconcave distribution
with finite support of diameter $D$. Even though its output guarantee
is stronger than that of the $\bw$, it still has the same bottleneck,
requiring a rather stringent $\eu R_{\infty}$-warmness. A natural
question arises from these observations. 
\begin{question*}
Can we combine the better isoperimetric bounds of $\psunif$ (e.g.,
$\cpi$) with the relaxed warmness requirements of $\har$ (e.g.,
$\eu R_{2}$-warmness)?
\end{question*}
Our \textbf{Result 1} in \S\ref{subsec:results} demonstrates that
$\psunif$ requires only $\eu R_{c}$-warmness for $c=\Otilde(1)$\footnote{$\Otilde(1)$ is not a rigorous expression to indicate $\polylog$,
but we abuse notation here for the sake of exposition.} to obtain a sample with $\eu R_{2}$-guarantee while retaining a
mixing guarantee characterized by \eqref{eq:pi}. Hence, $\psunif:\eu R_{c}\to\eu R_{2}$
under \eqref{eq:pi}, and this relaxes the requirement of stringent
initial warmness $\eu R_{\infty}$ to $\eu R_{c}$ without compromising
query complexity. This result not only improves theoretical guarantees
but also opens doors to more flexible algorithmic design for downstream
tasks, as seen shortly.

Moreover, we show analogous improvements for the $\ps$ for truncated
Gaussian distributions $\pi\gamma_{\sigma^{2}}$ (referred to as $\psgauss$
in \cite{KZ25Renyi}), where $\pi$ is the uniform distribution over
$\K$ and $\gamma_{\sigma^{2}}$ is a standard Gaussian with covariance
$\sigma^{2}I_{n}$. This family of distributions plays an important
role in annealing (e.g., $\gc$). Specifically, we prove that the
warmness requirement can also be relaxed from $\eu R_{\infty}$ to
$\eu R_{c}$ without deteriorating the previously established query
complexity of $n^{2}\clsi(\pi\gamma_{\sigma^{2}})\leq n^{2}\sigma^{2}$
(i.e., $\psgauss:\eu R_{c}\to\eu R_{2}$ under \eqref{eq:pi}).

\paragraph{Warm-start generation by annealing.}

Our discussion naturally prompts the question of how to quickly generate
a warm start. The main idea is \emph{annealing} \cite{DFK91random,KLS97random,LV06simulated,CV18Gaussian,KZ25Renyi}.
It involves constructing a sequence $\{\mu_{i}\}_{i\in[m]}$ of distributions
gradually approaching the target $\pi$, such that $(i)$ $\mu_{1}$
is easy to sample from (e.g., Gaussians), $(ii)$ $\mu_{m}$ is close
to the target $\pi$, and $(iii)$ a current distribution $\mu_{i}$
provides a warm start to the next one $\mu_{i+1}$.

Lov\'asz and Vempala~\cite{LV06simulated} proposed an annealing
scheme maintaining $\eu R_{2}$-warmness, with $\har$ used to transition
across annealing distributions. Due to the approximate nature of the
distributions obtained during annealing, their analysis uses a coupling
argument to account for this issue, providing guarantees for final
samples in $\tv$-distance. Hence, combined with $\har$, when $R^{2}=\E_{\pi}[\norm{\cdot}^{2}]$,
their final sampling algorithm has query complexity of $n^{3}R^{2}$
for $\tv$-guarantees.

$\gc$ \cite{CV18Gaussian} uses a sequence of truncated Gaussians
$\pi\gamma_{\sigma^{2}}$, increasing the variance $\sigma^{2}$ from
$n^{-1}$ to $R^{2}$ according to a predetermined schedule, $\sigma^{2}\gets\sigma^{2}(1+\nicefrac{\sigma^{2}}{R^{2}})$.
This scheme relays stronger $\eu R_{\infty}$-warmness, thus being
more `conservative' than the Lov\'asz--Vempala scheme, but benefits
from the faster mixing of the $\bw$ (for truncated Gaussian) compared
to $\har$. $\gc$ has complexity of $n^{2}(n\vee R^{2})$ with $\tv$-distance
guarantees. Roughly, doubling $\sigma^{2}$ takes $R^{2}/\sigma^{2}$
phases, and with the complexity per phase of $n^{2}\sigma^{2}$, this
gives the overall claimed complexity. Later, Kook and Zhang~\cite{KZ25Renyi}
refined this scheme by replacing the $\bw$ with $\psgauss$ for Gaussian
sampling, achieving final guarantees in $\eu R_{\infty}$-divergence
with the same complexity. 

Now that we have the sampler for truncated Gaussian and uniform distributions
with $\eu R_{c}$-warmness for $c=\Otilde(1)$, we are naturally led
to accelerate an annealing schedule of $\sigma^{2}$ in $\gc$, with
the Lov\'asz--Vempala scheme in mind. This would make it sufficient
to maintain only $\eu R_{c}$-warmness instead of stringent $\eu R_{\infty}$-warmness,
while still leveraging the faster mixing of $\psgauss$. This relaxation
allows us to update $\sigma^{2}$ more \emph{rapidly} by $\sigma^{2}\gets\sigma^{2}(1+\nicefrac{\sigma}{R})$
(\S\ref{subsubsec:faster-sampling}). Doubling $\sigma^{2}$ now
requires roughly $R/\sigma$ annealing phases, and the query complexity
improves to $n^{2}R\sigma$, compared to $n^{2}R^{2}$. Unfortunately,
since the annealing procedure must continue until $\sigma^{2}\asymp R^{2}$,
this seemingly promising approach yields no overall complexity gain.

\paragraph{Log-Sobolev inequality for strongly logconcave distributions.}

This leads us to the question how to sample $\pi\gamma_{\sigma^{2}}$
faster. One possibility is to develop an improved \eqref{eq:lsi}
if possible.  There are at least  two distinct approaches to bounding
$\clsi$ for a strongly logconcave density with support of diameter
$D$. The first, which yields a bound $\clsi(\pi\gamma_{\sigma^{2}})\leq\sigma^{2}$,
directly leverages strong logconcavity and can be established by multiple
ways, e.g., Caffarelli's contraction theorem \cite{Caffarelli00monotonicity}
or the Bakry--\'Emery condition \cite{BGL14analysis}. The second
is of the form $\clsi(\pi\gamma_{\sigma^{2}})\lesssim D^{2}$ \cite{FK99lsi},
independent of strong logconcavity and arises from finite diameter
of the distribution. The latter was improved to $\mathcal{O}(D)$
for \emph{isotropic} logconcave distributions \cite{LV24eldan}.

In fact, our current understanding still remains incomplete. To illustrate
this, consider the uniform distribution $\pi$ over an isotropic convex
body, which has diameter at most $n+1$ \cite{KLS95isop}. Here, the
strong logconcavity-based bound gives $\clsi(\pi\gamma_{n^{2}})\leq n^{2}$,
while the second method, through the Holley--Stroock perturbation
principle (Lemma~\ref{lem:bdd-perturbation}) and the $\O(D)$-bound
above, implies a better bound of $\clsi(\pi\gamma_{n^{2}})\lesssim\clsi(\pi)\lesssim n$.
This gap highlights the following question.
\begin{question*}
Is there a better bound on $\clsi$ for strongly logconcave distributions
with compact support?
\end{question*}
Our \textbf{Result 2} provides a new bound of the form $\clsi(\pi\gamma_{\sigma^{2}})\lesssim D\,\norm{\cov\pi}^{1/2}$,
valid for all $\sigma^{2}>0$ and logconcave distributions $\pi$
supported on diameter-$D$ domains. In the example above, our bound
yields $\clsi(\pi\gamma_{\sigma^{2}})\lesssim n$ for \emph{any} $\sigma^{2}$.

\paragraph{Back to warm-start generation.}

With this new machinery in hand, we return to the warm-start generation
problem. It is now apparent that our improved LSI bound should have
direct algorithmic implications. In the Gaussian-annealing method,
we may assume that the diameter is of order $R$ by a suitable preprocessing,
and thus once we reach $\sigma^{2}\gtrsim R\,\norm{\cov\pi}^{1/2}$,
sampling from $\pi\gamma_{\sigma^{2}}$ can be \emph{accelerated}
due to the faster mixing rate of $n^{2}\clsi(\pi\gamma_{\sigma^{2}})\leq n^{2}R\,\norm{\cov\pi}^{1/2}$.
The complexity for doubling $\sigma^{2}$ before reaching this accelerated
phase is bounded by $n^{2}R\sigma\lesssim n^{2}R^{3/2}\,\norm{\cov\pi}^{1/4}$,
and similarly, subsequent doublings within the accelerated phase also
incur at most the same complexity. Therefore, we establish a faster
sampling algorithm with total query complexity of $n^{2}R^{3/2}\,\norm{\cov\pi}^{1/4}$
(\textbf{Result 3}) (note $\norm{\cov\pi}\leq R^{2}$). Notably, when
$\pi$ is near-isotropic, this translates to an $n^{2.75}$-complexity
sampling algorithm, breaking the previous cubic bound first established
in \cite{CV15bypass,CV18Gaussian}. For the special case of sampling
a standard Gaussian restricted to an arbitrary convex body, the complexity
is $n^{2.5}$, again improving on the previous best cubic bound \cite{CV14cubic}.

\paragraph{Extension to logconcave distributions.}

Our approach can be extended to general logconcave distributions
under a well-defined function oracle. Although the extension is conceptually
similar, it involves some technical complications related to handling
general convex potentials.

The complexity of logconcave sampling under this zeroth-order oracle
has been studied using similar tools, notably by the $\bw$ \cite{LS93random},
$\har$ \cite{LV06fast,LV07geometry}, and the $\ps$ \cite{KV25sampling}.
In \cite{KV25sampling}, sampling from $\pi$ is reduced to sampling
from an exponential distribution: 
\[
\bar{\pi}(x,t)\propto\exp(-nt)\,\ind_{\K}(x,t)\quad\text{for the covex set }\K:=\{(x,t)\in\Rn\times\R:V(x)\leq nt\}\,.
\]
To facilitate warm-start generation, they proposed $\tgc$ that uses
``tilted'' Gaussians defined as $\mu_{\sigma^{2},t}(x,t)\propto\exp(-\rho t)\,\gamma_{\sigma^{2}}(x)\,\ind_{\K}(x,t)$.
Here, $\exp(-\rho t)$ controls the augmented $t$-direction, while
$\gamma_{\sigma^{2}}$ tames the original $x$-direction similar to
approaches in uniform sampling.

Our \textbf{Result 4} demonstrates that both the exponential and the
tilted Gaussian distributions can also be sampled under relaxed warmness
as in the case of uniform sampling. By suitably adapting $\tgc$ and
leveraging the improved LSI bound, we establish that logconcave sampling
for a $\tv$-guarantee can be accomplished using $n^{2.5}+n^{2}R^{3/2}\,\norm{\cov\pi}^{1/4}$
evaluation queries, which is $n^{2.75}$ for near-isotropic distributions,
matching the complexity of uniform sampling. This improves the previous
best complexity of $n^{2}(n\vee R^{2})$ for logconcave sampling \cite{KV25sampling}. 

\subsection{Results\label{subsec:results}}

Here we state our main results: (1) relaxation of a warmness requirement
of the proximal sampler ($\PS$) for uniform and truncated Gaussian
distributions, (2) a new bound on \eqref{eq:lsi} for strongly logconcave
distributions with compact support, (3) uniform and Gaussian sampling
through faster warm-start generation, and (4) an extension of these
results to general logconcave distributions. We defer a detailed discussion
of techniques to \S\ref{sec:techniques} and preliminaries to \S\ref{sec:prelim}.

\paragraph{Result 1: Uniform and Gaussian sampling under a relaxed warmness
(\S\ref{sec:unif-sampling-warm}).}

The proximal sampler is a general sampling framework consisting of
two steps. For a target distribution $\pi^{X}$ and step size $h$,
it considers an augmented distribution defined as $\pi^{X,Y}(x,y)\propto\pi^{X}(x)\,\gamma_{h}(y-x)$,
and then iterates $(i)$ $y_{i+1}\sim\pi^{Y|X=x_{i}}$ and $(ii)$
$x_{i+1}\sim\pi^{X|Y=y_{i+1}}$. Hence, $\psunif$ can be obtained
by setting $\pi^{X}\propto\ind_{\K}(x)$, under which the second step
corresponds to sampling from $\pi^{X|Y=y_{i+1}}=\gamma_{h}(\cdot-y_{i+1})|_{\K}$,
and it can be implemented by rejection sampling with the Gaussian
proposal $\gamma_{h}(\cdot-y_{i+1})$. 

In \S\ref{subsec:Uniform-sampling}, we relax warmness requirements
of $\psunif$ for uniform distributions from $\eu R_{\infty}$ to
$\eu R_{c}$ with $c=\Otilde(1)$, improving the warmness condition
in \cite{KVZ24INO} (see Theorem~\ref{thm:body-unif-samp} for details).
Below, $M_{q}$ indicates an $L^{q}$-norm defined as $M_{q}:=\norm{\D\pi_{0}/\D\pi}_{L^{q}(\pi)}=\exp(\frac{q-1}{q}\,\eu R_{q}(\pi_{0}\mmid\pi))$.

\begin{algorithm}[t]
\hspace*{\algorithmicindent} \textbf{Input:} initial point $x_{0}\sim\pi_{0}$,
convex body $\K\subset\Rn$, iterations $k\in\mathbb{N}$, threshold
$N$, variance $h$.

\hspace*{\algorithmicindent} \textbf{Output:} $x_{k+1}$.

\begin{algorithmic}[1] \caption{$\protect\ps$ $\protect\psunif$\label{alg:prox-unif}}
\FOR{$i=0,\dotsc,k$}

\STATE Sample $y_{i+1}\sim\gamma_{h}(\cdot-x_{i})=\mc N(x_{i},hI_{n})$.\label{line:forward}

\STATE Sample $x_{i+1}\sim\gamma_{h}(\cdot-y_{i+1})|_{\K}=\mc N(y_{i+1},hI_{n})|_{\K}$.
\label{line:backward}

\STATE$\quad(\uparrow)$ {\small Repeat $x_{i+1}\sim\gamma_{h}(\cdot-y_{i+1})$
until $x_{i+1}\in\mc K$. If $\#$attempts$_{i}\,\geq N$, declare
\textbf{Failure}.}

\ENDFOR \end{algorithmic}
\end{algorithm}

\begin{thm}[Uniform sampling from warm start]
\label{thm:unif-sampling-warmstart-intro} Let $\pi$ be a uniform
distribution over a convex body $\K\subset\Rn$ specified by $\mem_{R}(\K)$
with $\lda=\norm{\cov\pi}$. For any $\eta,\veps\in(0,1)$, initial
distribution $\pi_{0}$ with $c=\polylog\frac{n\lda M_{2}}{\eta\veps}$
and $M_{c}=\norm{\nicefrac{\D\pi_{0}}{\D\pi}}_{L^{c}(\pi)}$, we can
use $\psunif$ (Algorithm~\ref{alg:prox-unif}) with suitable choices
of parameters, so that with probability at least $1-\eta$, we obtain
a sample whose law $\mu$ satisfies $\eu R_{2}(\mu\mmid\pi)\leq\veps$,
using
\[
\Otilde\bpar{M_{c}n^{2}\lda\,\polylog\frac{1}{\eta\veps}}
\]
membership queries in expectation.
\end{thm}

Similarly, we relax the previous $\eu R_{\infty}$-warmness condition
for $\psgauss$ \cite{KZ25Renyi} in \S\ref{subsec:gaussian-sampling}.
See Theorem~\ref{thm:body-gauss-samp} for more details.
\begin{thm}[Restricted Gaussian sampling from warm start]
\label{thm:gauss-sampling-warmstart-intro} Let $\pi$ be the uniform
distribution $\pi$ over a convex body $\K\subset\Rn$ specified by
$\mem_{R}(\K)$ with $\lda=\norm{\cov\pi}$and $x_{0}=0$. Consider
a Gaussian $\pi\gamma_{\sigma^{2}}$ truncated to $\K$. For any $\eta,\veps\in(0,1)$,
initial distribution $\pi_{0}$ with $c=\polylog\frac{n\sigma DM_{2}}{\eta\veps}$
and $M_{c}=\norm{\nicefrac{\D\pi_{0}}{\D\pi}}_{L^{c}(\pi)}$, we can
use $\psgauss$ with suitable choices of parameters, so that with
probability at least $1-\eta$, we obtain a sample whose law $\mu$
satisfies $\eu R_{2}(\mu\mmid\pi\gamma_{\sigma^{2}})\leq\veps$, using
\[
\Otilde\bpar{M_{c}n^{2}(\sigma^{2}\wedge D\lda^{1/2})\,\polylog\frac{1}{\eta\veps}}
\]
membership queries in expectation.
\end{thm}

\paragraph{Result 2: LSI for strongly logconcave distributions with compact
support (\S\ref{app:LSI-interpolation-SL}).}

To attain the improved bound of $\clsi(\pi\gamma_{h})\lesssim D\,\norm{\cov\pi}^{1/2}$,
we first state two new findings which together yield the improved
bound as an immediate corollary.

For a logconcave distribution $\pi$ with support of diameter $D$,
we have $\clsi(\pi)\lesssim D^{2}$, and $\clsi(\pi)\lesssim D$ when
$\pi$ is near-isotropic. In \S\ref{subsec:LSI-interpolation}, we
find an interpolation of these two bounds.
\begin{thm}
\label{thm:intro-LSI-interpolation} Let $\pi$ be a logconcave distribution
 over $\Rn$ with support of diameter $D>0$. Then,
\[
\clsi(\pi)\lesssim\max\{D\,\norm{\cov\pi}^{1/2},D^{2}\wedge\norm{\cov\pi}\log^{2}n\}\,,
\]
and $\clsi(\pi)\lesssim D\cpi^{1/2}(\pi)\lesssim D\,\norm{\cov\pi}^{1/2}\log^{1/2}n$.
\end{thm}

Since $D\geq n^{1/2}$ for isotropic $\pi$, this bound recovers the
$D$-bound for the isotropic case. For a general case, as $\norm{\cov\pi}=\sup_{v\in\mbb S^{n-1}}\E_{\pi}[(v^{\T}(X-\E_{\pi}X))^{2}]\leq D^{2}$,
it also recovers the $D^{2}$-bound. Furthermore, the second bound
achieves $\clsi(\pi)\lesssim D\,\norm{\cov\pi}^{1/2}$ without any
logarithmic factor when the KLS conjecture is true.

Our second result in \S\ref{subsec:cov-str-LC} is purely relevant
to convex geometry, answering a geometric question of how Gaussian
weighting affects the largest eigenvalue of a covariance matrix.
\begin{thm}
\label{thm:intro-cov-str-LC} For a logconcave distribution $\pi$
over $\Rn$ with $R=\E_{\pi}\norm{\cdot}$ and $\lda=\norm{\cov\pi}$,
and any $h\gtrsim R\lda^{1/2}\log^{2}n\log^{2}\nicefrac{R^{2}}{\lda}$,
it holds that
\[
\norm{\cov\pi\gamma_{h}}\lesssim\norm{\cov\pi}\,.
\]
\end{thm}

Combining these two results and using the classical result of $\clsi(\pi\gamma_{h})\leq h$
\cite{BGL14analysis}, we state the improved LSI bound in \S\ref{subsec:FI-together}.
\begin{cor}
\label{cor:intro-new-LSI-str-LC} Let $\pi$ be a logconcave distribution
over $\Rn$ with $\lda=\norm{\cov\pi}$ and support of diameter $D>0$.
Then, for any $h>0$,
\[
\clsi(\pi\gamma_{h})\lesssim D\lda^{1/2}\log^{2}n\log^{2}\tfrac{D^{2}}{\lda}\,.
\]
\end{cor}

\paragraph{Result 3: Faster sampling from uniform and Gaussian distributions
(\S\ref{sec:faster-generation}).}

We accelerate $\gc$ further, where truncated Gaussians $\pi\gamma_{\sigma^{2}}$
are the annealing distributions. Sampling from $\pi\gamma_{\sigma^{2}}$
under $\eu R_{c}$-warmness with $c=\Otilde(1)$ (\textbf{Result 1})
makes it possible to anneal $\sigma^{2}$ more rapidly. Moreover,
once $\sigma^{2}$ reaches $R\,\norm{\cov\pi}^{1/2}$, sampling from
$\pi\gamma_{\sigma^{2}}$ can also be accelerated due to the improved
LSI bound (\textbf{Result 2}). Interweaving these developments together,
we sample from uniform distributions over convex bodies using $\Otilde(n^{2}R^{3/2}\,\norm{\cov\pi}^{1/4})$
queries, provably less than $\Otilde(n^{2}(n\vee R^{2}))$ queries
given in \cite{CV18Gaussian,KZ25Renyi}. See Theorem~\ref{thm:body-unif-gc}
for details.
\begin{thm}[Uniform sampling from cold start]
\label{thm:intro-unif-gc}  For any uniform distribution $\pi$ over
a convex body $\K\subset\Rn$ specified by $\mem_{x_{0},R}(\K)$ with
$\lda=\norm{\cov\pi}$, for any given $\eta,\veps\in(0,1)$, there
exists an algorithm that with probability at least $1-\eta$, returns
a sample whose law $\mu$ satisfies $\norm{\mu-\pi}_{\tv}\leq\veps$,
using
\[
\Otilde\bpar{n^{2}R^{3/2}\lda^{1/4}\polylog\frac{R\lda^{-1/2}}{\eta\veps}}
\]
membership queries in expectation. If $\pi$ is near-isotropic, then
$\Otilde(n^{2.75}\polylog\nicefrac{1}{\eta\veps})$ queries suffice.
\end{thm}

As a consequence, we can sample a truncated standard Gaussian in $n^{2.5}$-complexity
without a warm start, improving the previous best complexity by a
factor of $n^{1/2}$ than \cite{CV18Gaussian}.
\begin{cor}[Restricted Gaussian sampling from cold start]
\label{cor:intro-gauss-gc}  For any convex body $\K\subset\Rn$
given by $\mem_{x_{0}}(\K)$, there exists an algorithm that for any
given $\eta,\veps\in(0,1)$, with probability at least $1-\eta$ returns
a sample whose law $\mu$ satisfies $\norm{\mu-\gamma|_{\K}}_{\tv}\leq\veps$,
using $\Otilde(n^{2.5}\polylog\nicefrac{1}{\eta\veps})$ membership
queries in expectation.
\end{cor}

\paragraph{Result 4: Extension to logconcave distributions (\S\ref{sec:extension}).}

Just as in the uniform sampling, we can relax a previous $\eu R_{\infty}$-warmness
\cite{KV25sampling} to $\eu R_{c}$ with $c=\Otilde(1)$ for samplers
for the reduced exponential distribution $\bar{\pi}(x,t)\propto\exp(-nt)|_{\K}$
and tilted Gaussian distribution $\mu_{\sigma^{2},\rho}(x,t)\propto\exp(-\rho t)\,\gamma_{\sigma^{2}}(x)\,\ind_{\K}(x,t)$.
We refer readers to Theorem~\ref{thm:body-exp-samp} and \ref{thm:body-exp-tilt-Gauss}.
Then in \S\ref{subsec:warm-generation-LC}, we further accelerate
$\tgc$ proposed in \cite{KV25sampling}, culminating in the following
result. 
\begin{thm}[Logconcave sampling from cold start]
\label{thm:intro-LC-sampling} For any logconcave distribution $\pi$
over $\Rn$ presented by $\eval_{x_{0},R}(V)$, there exists an algorithm
that, given $\eta,\veps\in(0,1)$, with probability at least $1-\eta$
returns a sample whose law $\mu$ satisfies $\norm{\mu-\pi}_{\tv}\leq\veps$,
using
\[
\Otilde\bpar{n^{2}\max\{n^{1/2},R^{3/2}(\lda^{1/4}\vee1)\}\polylog\frac{R\lda^{-1/2}}{\eta\veps}}
\]
evaluation queries in expectation. If $\pi$ is near-isotropic, then
$\Otilde(n^{2.75}\polylog\nicefrac{1}{\eta\veps})$ queries suffice.
\end{thm}

\subsection{Technical overview\label{sec:techniques}}

Here we discuss detailed outlines of our approach and the main proofs.

\subsubsection{Sampling under weaker warmness}

\paragraph{$\protect\bw$ under relaxed warmness.}

Analysis of the $\bw$ has hitherto assumed $\eu R_{\infty}$-warmness
~\cite{LS90mixing}. Unlike $\har$, when the walk is near a corner
of a convex body (e.g., the tip of a cone), most steps become wasted
and make no progress, because the probability of finding another point
within the intersection of the step ball and the convex body can be
exponentially small. However, if the initial distribution is already
close to the uniform distribution $\pi$ in the $\eu R_{\infty}$-sense,
then the chance of stepping into such low-probability regions is small,
and thus the total expected wasted queries can still be bounded.

Precisely, in the analysis of the $\bw$ with step size $\delta$
(specifically, the $\sw$) in \cite{KLS97random}, each step samples
uniformly from $B_{\delta}(x)\cap\K$ with proposal uniform in $B_{\delta}(x)$,
so the success probability---called the \emph{local conductance}---is
\[
\ell(x)=\frac{\vol\bpar{B_{\delta}(x)\cap\K}}{\vol\bpar{B_{\delta}(x)}}\,.
\]
Its stationary distribution $\pi^{\msf{SW}}$ is proportional to $\ell$,
so if the initial distribution is already at stationarity, then the
expected number of trials per step is $\E_{\pi^{\msf{SW}}}\ell^{-1}=\O(1)$,
provided that the body contains a unit ball and the step size is $\delta\lesssim n^{-1/2}$.

Relaxing the warmness requirement is challenging. Suppose one starts
from a general initial distribution $\mu$ and attempts a change of
reference measure, say, through a Cauchy--Schwarz inequality:
\[
\bpar{\E_{\mu}\frac{1}{\ell}}^{2}\leq\frac{\int_{\K}\ell^{-1}\,\D x}{\int_{\K}\ell\,\D x}\cdot\Bnorm{\frac{\D\mu^{\ \ \ }\,}{\D\pi^{\msf{SW}}}}_{L^{2}(\pi^{\msf{SW}})}^{2}=\exp\bpar{\eu R_{2}(\mu\mmid\pi^{\msf{SW}})}\,\frac{\int_{\K}\ell^{-1}\,\D x}{\int_{\K}\ell\,\D x}\,.
\]
Even if $\mu$ and $\pi$ are close in $\eu R_{2}$, it is unclear
how to bound $\int_{\K}\ell^{-1}\,\D x$, and thus previous work has
always assumed $\eu R_{\infty}$-warmness to replace $\mu$ by $\pi^{\msf{SW}}$
without introducing extra $\ell^{-\alpha}$ into $\E_{\mu}\ell^{-1}$.

\paragraph{Proximal sampler.}

Before introducing $\psunif$, we briefly discuss the proximal sampler~\cite{LST21structured,CCSW22improved,FYC23improved,wibisono25mixing},
 which has been analyzed in \emph{well-conditioned} settings, where
$-\log\pi^{X}$ is $\alpha$-strongly convex and $\beta$-smooth in
$\Rn$. Such \emph{uniform} regularity allows rejection sampling to
implement the second step $x_{i+1}\sim\pi^{X|Y=y_{i+1}}$ efficiently,
requiring only $\O(1)$ queries for \emph{any} $y_{i+1}$.

However, constrained settings (such as uniform sampling over a convex
body $\K$) pose new challenges. For instance, for the uniform target
$\pi^{X}\propto\ind_{\K}$, one can readily derive that $\pi^{Y}(y)=\ell(y)/\vol\K$
and $\pi^{X|Y=y}\propto\mc N(y,hI_{n})|_{\K}$, where $\ell(y)=\P_{Z\sim\mc N(y,hI_{n})}(Z\in\K)$
serves as a Gaussian version of the local conductance. Even at stationarity
$\pi^{X}$, the expected per-step complexity of rejection sampling
with proposal $\mc N(y,hI_{n})$ becomes infinite, as $\E_{\pi^{Y}}\ell^{-1}$
is unbounded.

To address this, $\psunif$~\cite{KVZ24INO} introduces a threshold
parameter $N$: the algorithm declares failure if the rejection sampling
exceeds $N$ attempts. Under a step size $h\asymp n^{-2}$, they identify
an \emph{essential domain} as $\K_{t/n}=\K+B_{t/n}$, and show that
$\pi^{Y}(\Rn\setminus\K_{t/n})\lesssim e^{-t^{2}}$. For an $M_{\infty}$-warm
$\mu$, the expected number of trials can be bounded by carefully
choosing parameters $t,N$. For $\delta=t/n$,
\begin{align*}
\E_{\mu*\gamma_{h}}\bbrack{\frac{1}{\ell}\wedge N} & \leq M_{\infty}\,\E_{\pi^{Y}}\bbrack{\frac{1}{\ell}\wedge N}\leq M_{\infty}\Bpar{\int_{\K_{\delta}}\frac{\D\pi^{Y}}{\ell}+N\pi^{Y}(\Rn\backslash\K_{\delta})}\\
 & =M_{\infty}\bpar{\frac{\vol\K_{\delta}}{\vol\K}+N\pi^{Y}(\Rn\backslash\K_{\delta})}\lesssim M_{\infty}(e^{t}+Ne^{-t^{2}})\,,
\end{align*}
ensuring exponentially small failure probability as well. Nonetheless,
similar to the $\bw$, applying Cauchy--Schwarz to relax $\eu R_{\infty}$-warmness
fails, as $\int_{\Rn}\ell^{-1}\,\D x$ still appears in the analysis.

\paragraph{Improved analysis under a relaxed warmness.}

Our approach demonstrates the first successful relaxation of $\eu R_{\infty}$-warmness.
We avoid such pointwise closeness by carefully using the $(p,q)$-H\"older
inequality and a refined decomposition of $\E_{\mu*\gamma_{h}}[\ell^{-1}\wedge N]$
into three regions: ($\msf A$) high local conductance in the essential
domain, ($\msf B$) low local conductance in the essential domain,
and ($\msf C$) non-essential domain: for $(\cdot)_{h}:=(\cdot)*\gamma_{h}$
below,
\[
\E_{\mu_{h}}\bbrack{\frac{1}{\ell}\wedge N}=\int_{\K_{\delta}\cap[\ell\geq N^{-p}]}\cdot+\int_{\K_{\delta}\cap[\ell<N^{-p}]}\cdot+\int_{\K_{\delta}^{c}}\cdot=:\msf A+\msf B+\msf C\,.
\]
We then bound the integrand simply by $N$ in $\msf B$ and $\msf C$,
and change the reference measure from $\mu_{h}$ to $\pi_{h}$ in
$\msf A$ and $\msf B$, defining $M_{q}=\norm{\nicefrac{\D\mu_{h}}{\D\pi_{h}}}_{L^{q}(\pi_{h})}$
and using the $(p,q)$-H\"older with parameters $p=1+\alpha^{-1}$,
$q=1+\alpha$ for some $\alpha\geq1$:
\begin{align*}
\msf A & \leq\Bpar{\int_{\K_{\delta}\cap[\ell\geq N^{-p}]}\frac{1}{\ell^{p}}\wedge N^{p}\,\D\pi_{h}}^{1/p}\,M_{q}\leq\Bpar{\int_{\K_{\delta}\cap[\ell\geq N^{-p}]}\frac{1}{\ell^{p}}\,\frac{\ell}{\vol\K}}^{1/p}\,M_{q}\\
 & =\Bpar{\int_{\K_{\delta}\cap[\ell\geq N^{-p}]}\frac{1}{\ell^{1/\alpha}}\,\frac{\D x}{\vol\K}}^{1/p}\,M_{q}\leq M_{q}N^{1/\alpha}\,\bpar{\frac{\vol\K_{\delta}}{\vol\K}}^{1/p}\,,\\
\msf B & \leq N\int_{\K_{\delta}\cap[\ell<N^{-p}]}\frac{\D\mu_{h}}{\D\pi_{h}}\,\D\pi_{h}\leq NM_{q}\,\Bpar{\int_{\K_{\delta}\cap[\ell<N^{-p}]}\frac{\ell}{\vol\K}}^{1/p}\leq M_{q}\bpar{\frac{\vol\K_{\delta}}{\vol\K}}^{1/p}\,,\\
\msf C & \leq N\int_{\K_{\delta}^{c}}\frac{\D\mu_{h}}{\D\pi_{h}}\,\D\pi_{h}\leq NM_{2}\bpar{\pi_{h}(\mc K_{\delta}^{c})}^{1/2}\,.
\end{align*}
Selecting $\alpha=\log N$ ensures $N^{1/\alpha}=1$, so our bounds
depend only on $\pi_{h}(\K_{\delta}^{c})$ and volume ratios. With
proper choices of $h$, $\delta$, and $N$, the per-step complexity
can be bounded under $\eu R_{1+\log N}$-warmness. 

\subsubsection{Improved LSI for strongly logconcave distributions with compact support}

We now provide insights and proofs for two results: (1) a new bound
on the log-Sobolev constant, $\clsi(\pi)\lesssim D\,\norm{\cov\pi}^{1/2}$,
and (2) for a logconcave distribution $\pi$, $\norm{\cov\pi\gamma_{h}}\lesssim\norm{\cov\pi}$
whenever $h\gtrsim\norm{\cov\pi}^{1/2}\mathbb{E}_{\pi}\norm{\cdot}$.
These results, combined with the classical Bakry--\'Emery condition,
yield our key bound: for strongly logconcave distributions supported
within diameter $D$, $\clsi(\pi\gamma_{h})\lesssim_{\log}D\,\norm{\cov\pi}^{1/2}$
for any $h>0$. We believe each of these results is of independent
interest.

\paragraph{(1) Functional inequalities: $\protect\clsi(\pi)\lesssim_{\log}D\,\protect\norm{\protect\cov\pi}^{1/2}$.}

This result interpolates known bounds for \eqref{eq:lsi}: $\clsi(\pi)\lesssim D$
(isotropic case) and $\clsi(\pi)\lesssim D^{2}$ (general case). Unlike
\eqref{eq:pi}, the interpolation via affine transformations does
not work well for \eqref{eq:lsi}. Precisely, let $X\sim\pi$ with
$\Sigma=\cov\pi$, $Y:=\Sigma^{-1/2}X$, and $\nu:=\law Y$. Then,
$\nu$ is isotropic logconcave. For $\cpi(n)$ the largest Poincar\'e
constant of $n$-dimensional isotropic logconcave distributions and
for locally Lipschitz $f$,
\begin{align*}
\var_{\pi}f & =\var_{\nu}(f\circ\Sigma^{1/2})\leq\cpi(n)\,\E_{\nu}[\norm{\nabla(f\circ\Sigma^{1/2})}^{2}]\\
 & \leq\cpi(n)\,\E_{\nu}[\norm{\Sigma}\norm{\nabla f\circ\Sigma^{1/2}}^{2}]\leq\cpi(n)\,\norm{\Sigma}\,\E_{\pi}[\norm{\nabla f}^{2}]\,,
\end{align*}
which implies $\cpi(\pi)\leq\cpi(n)\,\norm{\Sigma}$. Similarly, $\clsi(\pi)\leq\clsi(n)\,\norm{\Sigma}$
holds (see \S\ref{subsec:naive-LSI}). However, $\clsi(n)$ can be
as large as $D\asymp n$ \cite{LV24eldan}, this LSI bound is unsatisfactory
unlike the \eqref{eq:pi} case.

Inspired by \cite{KL24isop}, we present a simple approach leveraging
a deep result established by Milman \cite{milman10isoperimetric}
that \emph{Gaussian concentration} and \eqref{eq:lsi} are equivalent
for logconcave measures. Recall that a concentration function (Definition~\ref{def:concentration})
of a probability measure $\pi$ is defined as
\[
\alpha_{\pi}(r)=\sup_{E:\pi(E)\geq1/2}\pi(\Rn\backslash E_{r})\quad\text{for all }r\geq0\,.
\]
Gaussian concentration with constant $\cgauss(\pi)$ refers to $\alpha_{\pi}(r)\leq2\exp(-r^{2}/\cgauss(\pi))$
while exponential concentration with constant $\cexp(\pi)$ refers
to $\alpha_{\pi}(r)\leq2\exp(-r/\cexp(\pi))$. There are two known
results on exponential concentration of logconcave measures. The first
one, $\cexp^{2}(\pi)\lesssim\cpi(\pi)$, is classical and holds under
\eqref{eq:pi} regardless of logconcavity. The second one, $\alpha_{\pi}(r)\leq2\exp(-c\min\{r/\lda^{1/2},r^{2}/\lda\log^{2}n\})$
for $\lda=\norm{\cov\pi}$ and universal constant $c>0$, is obtained
by Bizeul \cite{bizeul24measures} through Eldan's stochastic localization
(SL)~\cite{eldan13thin}. 

Since the diameter of support is $D$, we clearly have $\alpha_{\pi}(r)=0$
if $r>D$, and $r/D\leq1$ otherwise. Then, we have $\alpha_{\pi}(r)\leq2\exp(-\nicefrac{r}{c\cpi^{1/2}(\pi)})\leq2\exp(-\nicefrac{r^{2}}{cD\cpi^{1/2}(\pi)})$.
Due to Milman's result on equivalence, $\clsi(\pi)\lesssim D\cpi^{1/2}(\pi)\lesssim D\lda^{1/2}\log^{1/2}n$.
Similarly, when using Bizeul's concentration, we can obtain $\clsi(\pi)\lesssim\max\{D\lda^{1/2},\lda\log^{2}n\}$.
Taking minimum with the bound of $\clsi(\pi)\lesssim D^{2}$ and using
$\lda\leq D^{2}$, we can conclude that $\clsi(\pi)\lesssim\max\{D\lda^{1/2},D^{2}\wedge\lda\log^{2}n\}\leq D^{2}\lda^{1/2}\log n$.
See Remark~\ref{rem:LSI-comparison} on a comparison between these
two LSI bounds.

We present another proof through SL, specifically a simplified version
from~\cite{LV24eldan} (see \S\ref{app:LSI-interpolation-SL}),
which was our first approach in an earlier version. 

\paragraph{(2) Convex geometry: $\protect\norm{\protect\cov\pi\gamma_{h}}\lesssim\protect\norm{\protect\cov\pi}$.}

Consider the isotropic uniform distribution $\pi$ over a convex body
for illustration. We have $\norm{\cov\pi\gamma_{h}}\leq\cpi(\pi\gamma_{h})\le h$,
where the first inequality follows from \eqref{eq:pi}, and the second
from the Brascamp--Lieb (or Lichnerowicz) inequality. Thus, $\norm{\cov\pi\gamma_{h}}\leq1$
for $h\leq1$. As we increase $h$, this pushes the mass toward the
boundary $\de\K$, likely boosting covariance. On the other hand,
for large $h\gtrsim n^{2}$, the measure $\pi\gamma_{h}$ becomes
a $\Theta(1)$-perturbation of $\pi$, since the support of $\pi$
has diameter at most $n+1$ (due to isotropy). Thus, $\pi\gamma_{h}$
is $\O(1)$-close to $\pi$ in $\eu R_{\infty}$, and thus $\norm{\cov\pi\gamma_{h}}\lesssim\norm{\cov\pi}=1$.
Therefore, it is plausible to conjecture that $\norm{\cov\pi\gamma_{h}}\lesssim1$
for \emph{all} $h$.

We show this conjecture for $h\gtrsim\mathbb{E}_{\pi}\|\cdot\|\asymp n^{1/2}$,
much smaller than $n^{2}$. To prove this, we bound $\mathbb{E}_{\pi\gamma_{h}}[(X\cdot v)^{2}]$
for any unit vector $v$. Rewriting 
\[
\mathbb{E}_{\pi\gamma_{h}}[(X\cdot v)^{2}]=\frac{\mathbb{E}_{\pi}[(X\cdot v)^{2}\,e^{-\norm X^{2}/2h}]}{\mathbb{E}_{\pi}e^{-\norm X^{2}/2h}}\,,
\]
we aim to upper-bound the numerator and lower-bound the denominator,
both approximately by $e^{-n/2h}$. To this end, we examine which
regions mainly contribute to each expectation, relying on concentration
properties of $\pi$ only.

Recall isotropic logconcave distributions concentrate in a thin shell
of radius $n^{1/2}$ and width $\O(\log\log n)$ \cite{guan2024note}
(conjectured to be $\O(1)$ by the well-known \emph{thin-shell conjecture}
\cite{ABP03central,BK03central}). For the thin-shell ($\mathsf{S}$),
inner ($\mathsf{I}$), and outer ($\mathsf{O}$) regions, we argue
as follows: 
\begin{align*}
\msf{Den.}: & \quad\E_{\pi}e^{-\norm X^{2}/2h}\geq\E_{\pi}\inf_{\msf S}e^{-\norm X^{2}/2h}\gtrsim e^{-n/2h}\,\pi(\msf S)\gtrsim e^{-n/2h}\,,\\
\msf{Num.}: & \quad\E_{\pi}[(X\cdot v)^{2}e^{-\norm X^{2}/2h}]=\E_{\pi}[(\cdot)\times\ind_{\msf I}]+\E_{\pi}[(\cdot)\times\ind_{\msf S}]+\E_{\pi}[(\cdot)\times\ind_{\msf O}]\,.
\end{align*}
As for the numerator, the inner part's contribution is negligible
due to the small ball probabilities \cite{DP10small,bizeul2025slicing},
while the outer region's contribution is exponentially small since
the Gaussian weight is as small as $e^{-n/2h}$. Thus, the essential
contribution comes from the thin-shell, yielding 
\[
\E_{\pi}[(X\cdot v)^{2}e^{-\norm X^{2}/2h}]\lesssim e^{-n/2h}\,\E[(X\cdot v)^{2}]\leq e^{-n/2h}\norm{\cov\pi}=e^{-n/2h}\,.
\]

For general logconcave distributions, we use Lipschitz concentration
under \eqref{eq:pi}, given as 
\[
\pi(\norm X-\E_{\pi}\norm X\geq t)\leq3\exp\bpar{-t/\cpi^{1/2}(\pi)}\,,
\]
and the known bound $\cpi(\pi)\lesssim\norm{\cov\pi}\log n$. We then
apply the co-area formula and integration by parts to control contributions
from the inner and outer parts.

\subsubsection{Faster sampling algorithm\label{subsubsec:faster-sampling}}

We derive a faster sampling algorithm for uniform distributions over
convex bodies, with query complexity $n^{2}R^{3/2}\,\norm{\cov\pi}^{1/4}$,
which is $n^{2.75}$ for nearly isotropic ones, the first sub-cubic
bound.

\paragraph{Algorithm.}

Our algorithm (\S\ref{subsec:GC-algorithm}) follows the annealing
approach from $\gc$; it also utilizes truncated Gaussian $\gamma_{\sigma^{2}}|_{\mathcal{K}}$
for annealing, increasing $\sigma^{2}$ from $n^{-1}$ up to $R^{2}$,
with some key differences. First, we eliminate the initial annealing
by $\sigma^{2}\gets\sigma^{2}(1+n^{-1/2})$. Second, we accelerate
it through a faster multiplicative update by $1+\nicefrac{\sigma}{q^{1/2}R}$
with $q=\Otilde(1)$, instead of the slower update by $1+\nicefrac{\sigma^{2}}{R^{2}}$.
Third, once $\sigma^{2}$ surpasses $R\,\norm{\cov\pi}^{1/2}$ (i.e.,
the geometric mean between $R^{2}(\geq\tr\cov\pi)$ and $\norm{\cov\pi}$),
our algorithm benefits from a faster mixing of  $\psgauss$ due to
our new LSI bound.

\paragraph{Complexity.}

To bound the total complexity, we need a $c$-R\'enyi divergence
bound for $c=\Otilde(1)$ between neighboring annealing distributions.
We establish in Lemma~\ref{lem:variance-annealing} that $\eu R_{q}(\gamma_{\sigma^{2}}|_{\K}\mmid\gamma_{\sigma^{2}(1+\alpha)}|_{\K})\leq qR^{2}\alpha^{2}/\sigma^{2}$
for $q>1$, which generalizes the known result for $q=2$ \cite[Lemma 7.8]{CV18Gaussian}.
This bound justifies our multiplicative update with $\alpha=\nicefrac{\sigma}{q^{1/2}R}$.
Setting $q=c$, we ensure that $\gamma_{\sigma^{2}}|_{\K}$ and $\gamma_{\sigma^{2}(1+\alpha)}|_{\K}$
remain close in $\eu R_{c}$, enabling sampling by $\psgauss$ with
provable guarantees.

Each $\sigma^{2}$-doubling requires at most $q^{1/2}R/\sigma$ phases,
with each phase run by $\psgauss$ incurring $n^{2}\sigma^{2}$ queries.
Thus, each doubling has total complexity of $q^{1/2}n^{2}R\sigma$.
Given $\O(\log nR)$ doubling phases, the total query complexity up
to $\sigma^{2}\approx R\,\norm{\cov\pi}^{1/2}$ is bounded by $q^{1/2}n^{2}R^{3/2}\,\norm{\cov\pi}^{1/4}$.
Beyond $\sigma^{2}\gtrsim R\,\norm{\cov\pi}^{1/2}$, $\psgauss$ requires
$n^{2}R\,\norm{\cov\pi}^{1/2}$ queries thanks to our improved $\clsi$.
Thus, each subsequent doubling also contributes $q^{1/2}n^{2}R^{3/2}\,\norm{\cov\pi}^{1/4}$
as well. The final transition from $\gamma_{R^{2}}|_{\mathcal{K}}$
to $\pi$ via $\psunif$ requires $n^{2}\,\norm{\cov\pi}$ queries.
Therefore, the total complexity is $n^{2}R^{3/2}\,\norm{\cov\pi}^{1/4}$.

One subtle point we have brushed over is the approximate nature of
$\psgauss$. Specifically, when sampling from $\gamma_{\sigma^{2}}|_{\K}$,
it actually outputs a sample $X_{*}$ satisfying $\eu R_{2}(\law X_{*}\mmid\gamma_{\sigma^{2}}|_{\K})\leq\veps$,
thus \emph{slightly off} from $\gamma_{\sigma^{2}}|_{\K}$. Nonetheless,
we have pretended that the output is distributed as $\gamma_{\sigma^{2}}|_{\K}$.
We can readily address this issue by using the triangle inequality
for $\tv$-distance. Let $m$ be the total number of phases during
annealing, $\gamma_{i}$ the target distribution for $i\in[m]$, and
$P_{i}$ the Markov kernel defined by $\psgauss$ with suitable step
size. Then, the actual distribution of a sample at each phase $i$
is $\hat{\gamma}_{i}:=\gamma_{0}P_{1}\cdots P_{i}=\hat{\gamma}_{i-1}P_{i}$.
Using the triangle inequality and data-processing inequality,
\[
\norm{\hat{\gamma}_{i}-\gamma_{i}}_{\tv}\leq\norm{\hat{\gamma}_{i-1}P_{i}-\gamma_{i-1}P_{i}}_{\tv}+\norm{\gamma_{i-1}P_{i}-\gamma_{i}}_{\tv}\leq\norm{\hat{\gamma}_{i-1}-\gamma_{i-1}}_{\tv}+\veps\,,
\]
and induction leads to $\norm{\hat{\gamma}_{i}-\gamma_{i}}_{\tv}\leq m\veps$.
Since the total complexity scales as $\polylog\nicefrac{1}{\veps}$,
by replacing $\varepsilon\gets\varepsilon/m$, we can achieve the
desired accuracy without increasing complexity significantly.

\subsubsection{Extension to general logconcave distributions}

We now extend our developments thus far to general logconcave distributions,
under a well-defined function oracle. The overall strategy parallels
our approach for uniform sampling: (1) sampling under weaker warmness,
and (2) faster annealing while maintaining these weaker warmness conditions.

\paragraph{Logconcave sampling.}

As studied in~\cite{KV25sampling}, sampling from logconcave $\pi\propto e^{-V}$
can be reduced to sampling from the augmented distribution $\bar{\pi}(x,t)\propto e^{-nt}\,\ind_{\K}(x,t)$,
where $\K=\{(x,t)\in\Rn\times\R:V(x)\leq nt\}$ is convex due to the
convexity of $V$, incurring an additive overhead of $n^{2}$.

They used tilted Gaussians of the form $\mu_{\sigma^{2},\rho}(x,t)\propto\exp(-\rho t)\,\gamma_{\sigma^{2}}(x)\,\ind_{\K}$
for annealing. In \S\ref{subsec:lc-sampling-warm}, we study the
query complexity of the $\ps$ for these two distributions (denoted
by $\psexp$ and $\psann$ respectively) under $\eu R_{c}$-warmness
for  $c=\Otilde(1)$. Building upon prior ideas used for $\psunif$
and $\psgauss$, we refine the previous analysis of per-step guarantees
for $\psexp$ and $\psann$ under weaker warmness. Consequently, we
establish that $\psann$ requires $n^{2}(\sigma^{2}\vee1)$ evaluation
queries, while $\psexp$ requires $n^{2}(\norm{\cov\pi}\vee1)$ queries.

\paragraph{Tilted Gaussian cooling.}

We use $\tgc$ from \cite{KV25sampling} with extra care for faster
warm-start generation for the exponential distribution \textbf{$\bar{\pi}$}.

This algorithm (\S\ref{subsec:TCG-algorithm}) increases $\sigma^{2}$
from $n^{-1}$ to $R^{2}$, and $\rho$ from $1$ to $n$, and it
involves three phases---(1) \textbf{$\sigma^{2}$}-warming: Increase
$\sigma^{2}$ from $n^{-1}$ to $1$, (2) \textbf{$\rho$}-annealing:
With fixed $\sigma^{2}=1$, increase $\rho$ from $1$ to $n$, and
(3) \textbf{$\sigma^{2}$}-annealing: With fixed $\rho=n$, increase
$\sigma^{2}$ from $1$ to $R^{2}$. To bound distance between neighboring
distributions, we establish a R\'enyi version of the global update
lemma (Lemma~\ref{lem:global-annealing}), proven for $q=2$ in \cite{LV06simulated,KV06simulated}:
for logconcave $e^{-V}$,
\[
\eu R_{q}(e^{-(1+\alpha)V}\mmid e^{-V})\lesssim qn\alpha^{2}\,.
\]

In Phase I, we update $\sigma^{2}$ multiplicatively by $1+(qn)^{-1/2}$,
justified by the above guarantee. Since $\psann$ requires $n^{2}$
queries in this phase, the total complexity for Phase I is $n^{2.5}$.
Phase II is more involved. Initially, with $\sigma^{2}=1$, we simultaneously
update $\rho\gets\rho\,(1+(q\,(q\vee n))^{-1/2})$ and $\sigma^{2}\gets\sigma^{2}(1+(q\,(q\vee n))^{-1/2})^{-1}$,
where closeness follows from the global update lemma. To make up for
the slight decrease in $\sigma^{2}$, we further update $\sigma^{2}\gets\sigma^{2}(1+\nicefrac{\sigma}{q^{1/2}R})$
until $\sigma^{2}\leq1$, justified by Lemma~\ref{lem:variance-annealing},
repeating at most $\nicefrac{2R}{(q\vee n)^{1/2}}$ times. Note that
each annealing by $\psann$ requires $n^{2}$ queries. Given at most
$(q\,(q\vee n))^{1/2}\log n$ phases in $\rho$-annealing, each followed
by at most $\nicefrac{2R}{(q\vee n)^{1/2}}$ rounds of $\sigma^{2}$-annealing,
the total complexity of Phase II is $(q\,(q\vee n))^{1/2}\times\nicefrac{2R}{(q\vee n)^{1/2}}\times n^{2}\lesssim q^{1/2}n^{2}R.$

In Phase III, we increase $\sigma^{2}$ multiplicatively by $1+\nicefrac{\sigma}{q^{1/2}R}$,
where each annealing takes $n^{2}\sigma^{2}$ queries through $\psann$.
Each $\sigma^{2}$-doubling requires $q^{1/2}R/\sigma$ phases, leading
to $n^{2}R\sigma$ queries per doubling. Exploiting faster mixing
when $\sigma^{2}\gtrsim R\,(\norm{\cov\pi}^{1/2}\vee1)$, the total
complexity of Phase III is by $n^{2}R^{3/2}(\norm{\cov\pi}^{1/4}\vee1)$.
Combining complexities across all phases, the final total complexity
is $n^{2.5}+n^{2}R^{3/2}(\norm{\cov\pi}^{1/4}\vee1)$, which is again
$n^{2.75}$ for near-isotropic logconcave distributions.

\subsection{Preliminaries\label{sec:prelim}}

A function $f:\Rn\to[0,\infty)$ is \emph{logconcave} if $-\log f$
is convex in $\Rn$, and a probability measure (or distribution) is
logconcave if it has a logconcave density function with respect to
the Lebesgue measure\footnote{In this work, we only consider \emph{non-degenerate} logconcave distributions
whose support cannot be embedded into any proper subspace of $\Rn$.
Also, all distributions considered are absolutely continuous with
respect to the Lebesgue measure (i.e., $\pi\ll\mathfrak{m}$).}. This justifies our abuse of notation using the same symbol for a
distribution and density. For $t\geq0$, a distribution $\pi$ is
called \emph{$t$-strongly logconcave} if $-\log\pi$ is $t$-strongly
convex (i.e., $-\log\pi-\frac{t}{2}\,\norm{\cdot}^{2}$ is convex).
Recall that the multiplication and convolution preserve logconcavity.
A distribution is called \emph{isotropic} if it is centered (i.e.,
$\E_{\pi}X=0$) and has the identity covariance matrix (i.e., $\E_{\pi}[X^{\otimes2}]=I_{n}$).
Note that as the logconcave distributions decay exponentially fast
at infinity, they have finite moments of all orders. We reserve $\gamma_{h}$
for the centered Gaussian distribution with covariance matrix $hI_{n}$.
For a set $S$, we use $\ind_{S}(x):=[x\in S]$ to denote its indicator
function, and use $\mu|_{S}$ to denote a distribution $\mu$ truncated
to $S$ (i.e., $\mu|_{S}\propto\mu\cdot\ind_{S}$). For two probability
measures $\mu,\pi$, we use $\mu\pi$ to denote the new distribution
with density proportional to $\mu\pi$.

For $a,b\in\R$, we use $a\vee b$ and $a\wedge b$ to denote their
maximum and minimum, respectively. $B_{r}^{n}(x)$ denotes the $n$-dimensional
ball of radius $r>0$ centered at $x$, dropping the superscript $n$
if there is no confusion. Both $a\lesssim b$ and $a=\O(b)$ mean
$a\le cb$ for a universal constant $c>0$. $a=\Omega(b)$ means $a\gtrsim b$,
and $a\asymp b$ means $a=\O(b)$ and $a=\Omega(b)$. Lastly, $a=\Otilde(b)$
means $a=\mc O(b\polylog b)$. For a vector $v$ and a PSD matrix
$\Sigma$, $\norm v$ and $\norm{\Sigma}$ denote the $\ell_{2}$-norm
of $v$ and the operator norm of $\Sigma$, respectively.

We recall the basic definition of our computational model.
\begin{defn}
[Well-defined membership oracle, \cite{GLS93geometric}] \label{def:welldefined-membership}
It assumes access to a convex body $\mathcal{K}\subset\mathbb{R}^{n}$
through a membership oracle, which answers whether a queried point
belongs to $\mathcal{K}$. We assume that $\mathcal{K}$ contains
$B_{1}(x_{0})$ for some $x_{0}\in\Rn$ and has finite diameter $D>0$
and smooth boundary, and assume that the uniform probability measure
$\pi$ over $\mathcal{K}$, given by $\D\pi(x)\propto\ind_{\K}(x)\,\D x$
satisfies $R^{2}\geq\E_{\pi}[\norm{\cdot-x_{0}}^{2}]$ for some $R>0$.
We denote this oracle by $\mem_{\mc P}(\K)$, where $\mathcal{P}$
indicates access to parameter values in $\mathcal{P}$ (e.g., $\msf{\mem}_{R}(V)$
reveals $R$, while $\mem(\K)$ reveals none).
\end{defn}

This can be generalized to a well-defined function oracle.
\begin{defn}
[Well-defined function oracle] \label{def:welldefined-function}
It assumes access to an integrable logconcave function $\exp(-V)$
on $\mathbb{R}^{n}$, through an evaluation oracle for a convex potential
$V:\mathbb{R}^{n}\to\mathbb{R}\cup\{\infty\}$, which returns the
value of $V(x)$ for the queried point $x\in\Rn$, and also assumes
that the distribution $\pi$ with density $\D\pi\propto\exp(-V)\,\D x$
satisfies (1) $\E_{\pi}[\norm{\cdot-x_{0}}^{2}]\leq R^{2}$ for some
$R>0$ and (2) the \emph{ground set} \emph{$\msf L_{\pi,g}:=\{x\in\Rn:V(x)-\min V\leq10n\}$}
contains $B_{r}(x_{0})$. Without loss of generality, we will assume
that $x_{0}=0$ by translation and $r=1$ by scaling. We denote it
by $\eval_{\mc P}(V)$.
\end{defn}

We recall notions of probability divergences / distances between distributions. 
\begin{defn}
\label{def:p-dist} For two distributions $\mu,\nu$ over $\Rn$,
the \emph{$f$-divergence} of $\mu$ towards $\nu$ with $\mu\ll\nu$
is defined as, for a convex function $f:\R_{+}\to\R$ with $f(1)=0$
and $f'(\infty)=\infty$, 
\[
D_{f}(\mu\mmid\nu):=\int f\bpar{\frac{\D\mu}{\D\nu}}\,\D\nu\,.
\]
For $q\in(1,\infty)$, the \emph{$\KL$-divergence and $\chi^{q}$-divergence}
correspond to $f(x)=x\log x$ and $x^{q}-1$, respectively. The \emph{$q$-R\'enyi
divergence} is defined as 
\[
\eu R_{q}(\mu\mmid\nu):=\frac{1}{q-1}\,\log\bpar{\chi^{q}(\mu\mmid\nu)+1}=\frac{1}{q-1}\,\log\,\bnorm{\frac{\D\mu}{\D\nu}}_{L^{q}(\nu)}^{q}\,.
\]
The \emph{R\'enyi-infinity divergence} is defined as
\[
\eu R_{\infty}(\mu\mmid\nu):=\log\esssup_{\mu}\frac{\D\mu}{\D\nu}\,.
\]
For $c\in[1,\infty]$, a distribution $\mu$ is said to be \emph{$M_{c}$-warm
with respect to a distribution} $\nu$ if $\norm{\nicefrac{\D\mu}{\D\nu}}_{L^{c}(\nu)}\le M$
(i.e., $\eu R_{c}(\mu\mmid\nu)\leq\frac{c}{c-1}\,\log M$)\footnote{When $c=\infty$, we consider the $L^{\infty}$-norm with respect
to $\mu$, not $\nu$, to be consistent with $\eu R_{\infty}$ as
well as a prevalent definition of ($\infty$-)warmness.}. The \emph{total variation} ($\tv$) distance between $\mu$ and
$\nu$ is defined as 
\[
\norm{\mu-\nu}_{\msf{TV}}:=\frac{1}{2}\int|\mu(x)-\nu(x)|\,\D x=\sup_{S\in\mc F}\abs{\mu(S)-\nu(S)}\,,
\]
where $\mc F$ is the collection of all $\mu,\nu$-measurable subsets
of $\Rn$.
\end{defn}

We recall $\KL=\lim_{q\downarrow1}\eu R_{q}\leq\eu R_{q}\leq\eu R_{q'}\leq\eu R_{\infty}$
for $1\leq q\leq q'$ and $2\,\norm{\cdot}_{\tv}^{2}\leq\KL\leq\eu R_{2}=\log(\chi^{2}+1)\leq\chi^{2}$.
 The following classical lemmas will be useful in the proofs. We
refer readers to \cite{vH14renyi} for more properties of the R\'enyi
divergence. 
\begin{lem}
[Data-processing inequality] \label{lem:DPI} For probability measures
$\mu,\nu$, Markov kernel $P$, $f$-divergence $D_{f}$, and $q\in[1,\infty]$,
it holds that 
\[
D_{f}(\mu P\mmid\nu P)\leq D_{f}(\mu\mmid\nu)\,,\quad\text{and}\quad\eu R_{q}(\mu P\mmid\nu P)\leq\eu R_{q}(\mu\mmid\nu)\,.
\]
\end{lem}

\begin{lem}
[Bounded perturbation, \cite{HS87logSobolev}] \label{lem:bdd-perturbation}
Suppose that a probability measure $\pi$ satisfies \eqref{eq:lsi}
with constants $\clsi(\pi)<\infty$. If a probability measure $\mu$
satisfies $c\leq\frac{\D\mu}{\D\pi}\leq C$ for $c,C\in\R_{>0}$,
then 
\[
\clsi(\mu)\leq\frac{C}{c}\,\clsi(\pi)\,.
\]
\end{lem}

\section{Convex body sampling under relaxed warmness\label{sec:unif-sampling-warm}}

In this section, we establish the query complexities of sampling
from uniform and Gaussian distributions over a convex body, relaxing
the warmness requirement from $\eu R_{\infty}$ to $\eu R_{c}$ for
small $c=\Otilde(1)$. We refine previous analyses of the $\ps$ for
these distributions carried out in \cite{KVZ24INO,KZ25Renyi}, and
prove the following theorems.
\begin{thm}
[Restatement of Theorem~\ref{thm:unif-sampling-warmstart-intro}]
\label{thm:body-unif-samp} Consider the uniform distribution $\pi$
over a convex body $\K\subset\Rn$ specified by $\mem_{D}(\K)$, and
an initial distribution $\mu$ with $M_{q}=\norm{\D\mu/\D\pi}_{L^{q}(\pi)}$
for $q\geq2$. For any $\eta,\veps\in(0,1)$ and $k\in\mathbb{N}$
as defined below, $\psunif$ (Algorithm~\ref{alg:prox-unif}) with
$h=(2n^{2}\log\frac{16kM_{2}}{\eta})^{-1}$ and $N=(\frac{16kM_{2}}{\eta})^{2}\log^{4}\frac{16kM_{2}}{\eta}$
achieves $\eu R_{2}(\mu_{k}\mmid\pi)\leq\veps$ after $k=\Otilde(n^{2}\norm{\cov\pi}\log^{3}\frac{M_{2}}{\eta\veps})$
iterations, where $\mu_{k}$ is the law of the $k$-th iterate. With
probability at least $1-\eta$, $\psunif$ iterates $k$ times without
failure, using 
\[
\Otilde\bpar{M_{c}n^{2}\norm{\cov\pi}\log^{7}\frac{1}{\eta\veps}}
\]
 membership queries in expectation for any $c\geq12\log\frac{16kM_{2}}{\eta}$.
\end{thm}

This analysis improves the complexity bound of $\Otilde(M_{\infty}n^{2}\,\norm{\cov\pi}\polylog\nicefrac{1}{\eta\veps})$
under \emph{$\eu R_{\infty}$}-warmness condition established in \cite{KVZ24INO}.
If the inner radius of $\K$ is $r$, then the query complexity includes
an additional multiplicative factor of $r^{-2}$.

Using a similar technique, we also improve the complexity bound for
truncated Gaussian sampling under a $\eu R_{\infty}$-warmness assumption
proven in \cite{KZ25Renyi}.
\begin{thm}
[Restatement of Theorem~\ref{thm:gauss-sampling-warmstart-intro}]
\label{thm:body-gauss-samp} Along with the uniform distribution $\pi$
above with $x_{0}=0$ (i.e., $\mem_{x_{0}=0,D}(\K)$), consider a
Gaussian $\pi\gamma_{\sigma^{2}}$ truncated to $\K$ for $\sigma>0$,
and initial distribution $\mu$ with $M_{q}=\norm{\D\mu/\D(\pi\gamma_{\sigma^{2}})}_{L^{q}(\pi\gamma_{\sigma^{2}})}$
for $q\geq2$. For any $\eta,\veps\in(0,1)$, $k\in\mathbb{N}$ defined
below, $\psgauss$ with $h=(10n^{2}\log\frac{16kM_{2}}{\eta})^{-1}$
and $N=(\frac{16kM_{2}}{\eta})^{2}\log^{3}\frac{16kM_{2}}{\eta}$
achieves $\eu R_{2}(\mu_{k}\mmid\pi\gamma_{\sigma^{2}})\leq\veps$
after $k=\Otilde(n^{2}\sigma^{2}\log^{3}\frac{M_{2}}{\eta\veps})$
iterations, where $\mu_{k}$ is the law of the $k$-th iterate, and
with probability at least $1-\eta$ successfully iterates $k$ times
without failure, using 
\[
\Otilde\bpar{M_{c}n^{2}\sigma^{2}\log^{6}\frac{1}{\eta\veps}}
\]
membership queries in expectation for any $c\geq6\log\frac{16kM_{2}}{\eta}$.
When $\sigma^{2}\gtrsim D\lda^{1/2}\log^{2}n\log^{2}\frac{D^{2}}{\lda}$
for $\lda=\norm{\cov\pi}$, it suffices to run $k=\Otilde(n^{2}D\lda^{1/2}\log^{3}\frac{M_{2}}{\eta\veps})$
times with total query complexity 
\[
\Otilde\bpar{M_{c}n^{2}D\lda^{1/2}\log^{6}\frac{1}{\eta\veps}}\,.
\]
\end{thm}

We remark that the new complexity benefits not only from the relaxed
warmness requirement but also from the tighter bound on $\clsi(\pi\gamma_{\sigma^{2}})$,
which will be proven in \S\ref{app:LSI-interpolation-SL}.

\subsection{Uniform sampling\label{subsec:Uniform-sampling}}

We recall the $\ps$ for the uniform distribution over a convex body
$\K$ (denoted as $\psunif$ here), referred to as $\ino$ in \cite{KVZ24INO}.
Given a step size $h$, a preset threshold $N$, and the uniform target
$\pi^{X}(x)\propto\ind_{\K}(x)$, one iteration consists of
\begin{itemize}
\item {[}Forward{]} $y_{i+1}\sim\pi^{Y|X=x_{i}}=\mc N(x_{i},hI_{n})$.
\item {[}Backward{]} $x_{i+1}\sim\pi^{X|Y=y_{i+1}}=\mc N(y_{i+1},hI_{n})|_{\K}\propto\exp(-\frac{1}{2h}\,\norm{\cdot-y_{i+1}}^{2})\,\ind_{\K}(\cdot)$.
\begin{itemize}
\item This is implemented by repeating $x_{i+1}\sim\gamma_{h}(\cdot-y_{i+1})$
until $x_{i+1}\in\mc K$. If the number of attempts in this iteration
exceeds $N$, then we declare \textbf{Failure}.
\end{itemize}
\end{itemize}
We recall that $\pi^{Y}=\pi^{X}*\gamma_{h}=:\pi_{h}$.

One notable advantage of this sampler is that it separates $(i)$
the convergence-rate analysis and $(ii)$ the query-complexity analysis
of the backward step. The first question is how many iterations ---
each consisting of a successful execution of both the forward and
backward steps --- are required to achieve a target accuracy in a
desired metric. The second is how many queries are needed to successfully
implement the rejection sampling in each backward step. As a result,
the total query complexity is simply the product of these two quantities.

\subsubsection{Mixing analysis}

Previous work on the $\ps$ established its convergence rate under
various assumptions on the target, including \eqref{eq:pi} and \eqref{eq:lsi}.
It was first proven for smooth densities \cite[Theorem 3 and 4]{CCSW22improved}
and extended to distributions truncated to convex constraints \cite[Lemma 22]{KVZ24INO}.
For rigorous treatments of these types of strong data-processing
inequalities, we refer readers to \cite{KO25strong}.
\begin{lem}
[{\cite[Theorem 3 and 4]{CCSW22improved}}] \label{lem:ps-mixing}
Assume that a probability measure $\pi$ is absolutely continuous
with respect to the Lebesgue measure over a convex body  $\K$, and
let $P$ denote the Markov kernel of the $\ps$. If $\pi$ satisfies
\eqref{eq:lsi} with constant $\clsi$, then for any $q\geq1$ and
distribution $\mu$ with $\mu\ll\pi$,
\begin{equation}
\eu R_{q}(\mu P\mmid\pi)\leq\frac{\eu R_{q}(\mu\mmid\pi)}{(1+h/C_{\msf{LSI}})^{2/q}}\,.\label{eq:LSI-contraction}
\end{equation}
If $\pi$ satisfies \eqref{eq:pi} with constant $\cpi$, then 
\begin{equation}
\chi^{2}(\mu P\mmid\pi)\leq\frac{\chi^{2}(\mu\mmid\pi)}{(1+h/C_{\msf{PI}})^{2}}\,.\label{eq:PI-contraction}
\end{equation}
\end{lem}

Thus, the $\ps$ mixes in roughly $qh^{-1}\clsi\log\frac{\eu R_{q}}{\veps}$
and $h^{-1}\cpi\log\frac{\chi^{2}}{\veps}$ iterations under \eqref{eq:lsi}
and \eqref{eq:pi}, respectively. Since any logconcave distribution
$\pi$ satisfies \eqref{eq:pi}, and the current best bound \cite{Klartag23log}
gives 
\[
\norm{\cov\pi}\leq\cpi(\pi)\lesssim\norm{\cov\pi}\log n\,,
\]
$\psunif$ requires no more than $h^{-1}\norm{\cov\pi}\log n\log\frac{\chi^{2}(\pi_{0}\mmid\pi)}{\veps}$
iterations to achieve $\chi^{2}(\pi_{0}P^{n}\mmid\pi)\leq\veps$.

\subsubsection{Complexity of the backward step}

Under the relaxed warmness condition, there is no change in the mixing
analysis, but we must refine previous analysis of the backward step.
Here, we show that both the failure probability and the expected query
complexity (i.e., the expected number of trials until the first success
of rejection sampling) can be moderately bounded under weaker warmness.
In analysis, we can assume that $x_{0}=0$ by translation without
loss of generality.

As briefly sketched in \S\ref{sec:techniques}, the \emph{essential
domain} of $\pi^{Y}=\pi_{h}$ can be identified as $\K_{\delta}=\K+B_{\delta}$:
\begin{lem}
[{\cite[Lemma 26]{KVZ24INO}}] For a convex body $\K\subset\Rn$
containing $B_{1}(0)$,
\[
\pi^{Y}(\K_{\delta}^{c})\leq\exp\bpar{-\frac{\delta^{2}}{2h}+\delta n}\,.
\]
\end{lem}

By taking $\delta=\frac{t}{n},\,h=\frac{c}{n^{2}}$ for some $c,t>0$,
the measure of the non-essential part is bounded as
\[
\pi^{Y}(\K_{\delta}^{c})\leq\exp\bpar{-\frac{t^{2}}{2c}+t}\,.
\]

We now proceed to analyze the failure probability (i.e., the probability
of hitting the threshold at least once before mixing) and the expected
number of trials under wearer warmness, $\eu R_{q}(\mu^{X}\mmid\pi^{X})=\O(1)$
with $q\geq2$. While the former can be bounded straightforwardly
using the Cauchy--Schwarz inequality, the latter requires a more
delicate analysis as seen shortly.

\paragraph{(1) Failure probability.}

Recall that  
\[
\eu R_{q}(\mu\mmid\pi)=\frac{1}{q-1}\,\log\,\Bnorm{\frac{\D\mu}{\D\pi}}_{L^{q}(\pi)}^{q}=\frac{1}{q-1}\,\log\bpar{\chi^{q}(\mu\mmid\pi)+1}\,.
\]
Suppose we run $\psunif$ with initial distribution $\mu^{X}\ll\pi^{X}$.
Denoting $\mu_{h}:=\mu^{X}*\gamma_{h}$, we can bound the failure
probability of one iteration as
\[
\E_{\mu_{h}}[(1-\ell)^{N}]\leq\sqrt{\E_{\pi_{h}}[(1-\ell)^{2N}]}\,\Bnorm{\frac{\D\mu_{h}}{\D\pi_{h}}}_{L^{2}(\pi_{h})}\leq\sqrt{\E_{\pi_{h}}[(1-\ell)^{2N}]}\,\Bnorm{\frac{\D\mu}{\D\pi}}_{L^{2}(\pi)}\,,
\]
where the first inequality follows from the Cauchy--Schwarz inequality,
and the second follows from the data-processing inequality (e.g.,
the DPI for the R\'enyi divergence, Lemma~\ref{lem:DPI}). To bound
$\E_{\pi_{h}}[(1-\ell)^{2N}]$, we adapt the analysis in \cite{KVZ24INO}.
For $M_{2}=\norm{\D\mu/\D\pi}_{L^{2}(\pi)}$, we decompose
\[
\int_{\Rn}(1-\ell)^{2N}\,\D\pi_{h}=\int_{\mc K_{\delta}^{c}}\cdot+\int_{\mc K_{\delta}\cap[\ell\geq N^{-1}\log(3kM_{2}/\eta)]}\cdot+\int_{\mc K_{\delta}\cap[\ell<N^{-1}\log(3kM_{2}/\eta)]}\cdot=:\msf A+\msf B+\msf C\,.
\]
Using $\pi^{Y}=\pi_{h}=\ell/\vol\K$ when bounding $\msf C$ below,
\begin{align*}
\msf A & \leq\pi^{Y}(\mc K_{\delta}^{c})\leq\exp\bpar{-\frac{t^{2}}{2c}+t}\,,\\
\msf B & \leq\int_{[\ell\geq N^{-1}\log(3kM_{2}/\eta)]}\exp(-2\ell N)\,\D\pi^{Y}\leq\bpar{\frac{\eta}{3kM_{2}}}^{2}\,,\\
\msf C & \leq\int_{\mc K_{\delta}\cap[\ell<N^{-1}\log(3kM_{2}/\eta)]}\frac{\ell(y)}{\vol\mc K}\,\D y\leq\frac{\log(3kM_{2}/\eta)}{N}\,\frac{\vol\mc K_{\delta}}{\vol\mc K}\leq\frac{e^{t}}{N}\,\log\frac{3kM_{2}}{\eta}\,.
\end{align*}
Here, we used $\vol\mc K_{\delta}\subset\vol((1+\delta)\,\mc K)=(1+\delta)^{n}\,\vol\mc K\leq e^{t}\vol\mc K$
(due to $B_{1}(0)\subset\K$). 

Let $Z:=\frac{16kM_{2}}{\eta}$, $c=\frac{\log\log Z}{2\log Z}$,
$t=\sqrt{8}\log\log Z$, and $N=Z^{2}\log^{4}Z$. Under these choices,
we have $\frac{t^{2}}{2c}-t\geq\frac{t^{2}}{4c}$, which is equivalent
to $t\geq4c$. Hence, $\E_{\pi_{h}}[(1-\ell)^{2N}]\leq(\frac{\eta}{kM_{2}})^{2}$.
Therefore, 
\[
\E_{\mu_{h}}[(1-\ell)^{N}]\leq M_{2}\,\bpar{\E_{\pi_{h}}[(1-\ell)^{2N}]}^{1/2}\leq\frac{\eta}{k}\,,
\]
which implies that the total failure probability across $k$ iterations
is at most $\eta$.

\paragraph{(2) Complexity of the backward step.}

Let $p=1+\frac{1}{\alpha}$ and $q=1+\alpha$ with $\alpha=\log N$.
Then,
\[
\E_{\mu_{h}}\bbrack{\frac{1}{\ell}\wedge N}=\int_{\K_{\delta}\cap[\ell\geq N^{-p}]}\cdot+\int_{\K_{\delta}\cap[\ell<N^{-p}]}\cdot+\int_{\K_{\delta}^{c}}\cdot=:\msf A+\msf B+\msf C\,.
\]
Using the $(p,q)$-H\"older below,
\begin{align*}
\msf A & \leq\int_{\K_{\delta}\cap[\ell\geq N^{-p}]}\frac{1}{\ell}\,\D\mu_{h}\leq\Bpar{\int_{\K_{\delta}\cap[\ell\geq N^{-p}]}\frac{1}{\ell^{p}}\,\D\pi_{h}}^{1/p}\,\Bnorm{\frac{\D\mu_{h}}{\D\pi_{h}}}_{L^{q}(\pi_{h})}\\
 & =\Bpar{\int_{\K_{\delta}\cap[\ell\geq N^{-p}]}\frac{1}{\ell^{p-1}}\,\frac{\D x}{\vol\K}}^{1/p}\,\Bnorm{\frac{\D\mu_{h}}{\D\pi_{h}}}_{L^{q}(\pi_{h})}\leq N^{1/\alpha}\bpar{\frac{\vol\K_{\delta}}{\vol\K}}^{1/p}\,\Bnorm{\frac{\D\mu}{\D\pi}}_{L^{q}(\pi)}\leq e\,\frac{\vol\K_{\delta}}{\vol\K}\,\Bnorm{\frac{\D\mu}{\D\pi}}_{L^{q}(\pi)}\,,\\
\msf B & \leq N\int_{\K_{\delta}\cap[\ell<N^{-p}]}\frac{\D\mu_{h}}{\D\pi_{h}}\,\D\pi_{h}\leq N\Bpar{\int_{\K_{\delta}\cap[\ell<N^{-p}]}\frac{\ell}{\vol\K}}^{1/p}\,\Bnorm{\frac{\D\mu_{h}}{\D\pi_{h}}}_{L^{q}(\pi_{h})}\leq\frac{\vol\K_{\delta}}{\vol\K}\,\Bnorm{\frac{\D\mu}{\D\pi}}_{L^{q}(\pi)}\,,\\
\msf C & \leq N\int_{\K_{\delta}^{c}}\frac{\D\mu_{h}}{\D\pi_{h}}\,\D\pi_{h}\leq N\sqrt{\pi_{h}(\mc K_{\delta}^{c})}\,\Bnorm{\frac{\D\mu_{h}}{\D\pi_{h}}}_{L^{2}(\pi_{h})}\leq N\sqrt{\pi_{h}(\mc K_{\delta}^{c})}\,\Bnorm{\frac{\D\mu}{\D\pi}}_{L^{2}(\pi)}\,.
\end{align*}
Putting these together, 
\begin{align*}
\E_{\mu_{h}}\bbrack{\frac{1}{\ell}\wedge N} & \leq4\,\frac{\vol\K_{\delta}}{\vol\K}\,\Bnorm{\frac{\D\mu}{\D\pi}}_{L^{q}(\pi)}+N\sqrt{\pi_{h}(\mc K_{\delta}^{c})}\,\Bnorm{\frac{\D\mu}{\D\pi}}_{L^{2}(\pi)}\leq4e^{t}M_{q}+N\exp\bpar{-\frac{t^{2}}{4c}+\frac{t}{2}}M_{2}\\
 & \leq4M_{q}\log^{\sqrt{8}}Z+Z^{2}\log^{4}Z\times\exp\bpar{-\frac{t^{2}}{8c}}\,M_{2}\leq5M_{q}\log^{4}Z\,,
\end{align*}
where we used $M_{2}\leq M_{q}=\norm{\D\mu/\D\pi}_{L^{q}(\pi)}$ (i.e.,
the monotonicity of $L^{q}$-norm). Therefore, we only need $M_{q}$-warmness
with 
\[
q=1+\log N\leq2\log(Z^{2}\log^{4}Z)\leq12\log\frac{16kM_{2}}{\eta}\,.
\]

Combining these mixing and per-step analyses, we prove the main theorem.
\begin{proof}
[Proof of Theorem~\ref{thm:body-unif-samp}] By \eqref{eq:PI-contraction},
$\psunif$ can ensure $\chi^{2}(\mu_{k}\mmid\pi)\leq\veps$ after
iterating 
\[
k\gtrsim\bpar{h^{-1}\cpi(\pi)\vee1}\log\frac{M_{2}}{\veps}
\]
times. Under the choice of $h=\frac{c}{n^{2}}=\frac{\log\log Z}{2n^{2}\log Z}$
with $Z=\frac{16kM_{2}}{\eta}$, the required number $k$ of iterations
must satisfy an inequality of the form 
\[
k\gtrsim A\log^{2}(Bk)\,,
\]
which can be fulfilled if $k\gtrsim A\log^{2}(AB)$. Therefore, the
required number of iterations is of order
\[
\Otilde\bpar{n^{2}\cpi(\pi)\log^{3}\frac{M_{2}}{\eta\veps}}=\Otilde\bpar{n^{2}\norm{\cov\pi}\log^{3}\frac{M_{2}}{\eta\veps}}\,.
\]

Since $\eu R_{q}(\mu_{k}\mmid\pi)\leq\eu R_{q}(\mu\mmid\pi)$ by the
DPI (Lemma~\ref{lem:DPI}), the complexity of the backward step is
$\Otilde(M_{q})$ for any iteration. Multiplying this by the required
number of iterations, the total expected number of queries becomes
\[
\Otilde\bpar{M_{q}n^{2}\norm{\cov\pi}\log^{7}\frac{M_{2}}{\veps\eta}}\,.
\]
This completes the proof of Theorem~\ref{thm:body-unif-samp}.
\end{proof}

\subsection{Truncated Gaussian sampling\label{subsec:gaussian-sampling}}

In this section, we similarly relax a warmness condition of the $\ps$
for truncated Gaussian distributions \cite{KZ25Renyi} (referred to
here as $\psgauss$) from $\eu R_{\infty}$ to $\eu R_{c}$ for small
$c=\Otilde(1)$.

In the same setting as the previous section, let $\pi$ denote the
uniform distribution over $\K$, and define $\mu^{X}:=\pi\gamma_{\sigma^{2}}\propto\mc N(0,\sigma^{2}I_{n})|_{\K}$.
For $\tau:=\frac{\sigma^{2}}{h+\sigma^{2}}<1$, $\psgauss$ for $\mu^{X}$
alternates between
\begin{itemize}
\item {[}Forward{]} $y_{i+1}\sim\mu^{Y|X=x_{i}}=\mc N(x_{i},hI_{n})$.
\item {[}Backward{]} $x_{i+1}\sim\mu^{X|Y=y_{i+1}}=\mc N(\tau y_{i+1},\tau hI_{n})|_{\K}$.
\end{itemize}
This backward step is implemented via rejection sampling using the
proposal $\mc N(\tau y_{i+1},\tau hI_{n})$. If the number of trials
in a given iteration exceeds $N$, then declare \textbf{Failure}.

Regarding the mixing rate, it follows from \cite[Lemma 21]{KZ25Renyi}
or \eqref{eq:LSI-contraction} that $\psgauss$ achieves $\veps$-distance
in $\eu R_{2}$ after iterating 
\[
k\gtrsim h^{-1}\clsi(\mu^{X})\log\frac{\eu R_{2}}{\veps}\asymp h^{-1}\clsi(\mu^{X})\log\frac{\log M_{2}}{\veps}\,.
\]

\paragraph{Preliminaries.}

We now analyze the failure probability and expected query complexity
of the backward step under relaxed warmness. Recall from \cite{KZ25Renyi}
that for $y_{\tau}:=\tau y$ and $\mu_{h}:=\mu^{Y}=\mu^{X}*\gamma_{h}$,
the density is given by
\[
\mu_{h}(y)=\frac{\tau^{n/2}\,\ell(y)\exp(-\frac{1}{2\tau\sigma^{2}}\,\norm{y_{\tau}}^{2})}{\int_{\mc K}\exp(-\frac{1}{2\sigma^{2}}\,\norm x^{2})\,\D x}\,.
\]
To proceed, we will make use of two helper lemmas.
\begin{lem}
[{\cite[Lemma 22]{KZ25Renyi}}] \label{lem:gaussian-effective}
For a convex body $\K\subset\Rn$ containing $B_{1}(0)$, let $R=(1+\frac{h}{\sigma^{2}})\,\K_{\delta}=\tau^{-1}\mc K_{\delta}$.
If $\delta\geq hn$, then
\[
\mu^{Y}(R^{c})\leq\exp\bpar{-\frac{\delta^{2}}{2h}+\delta n+hn^{2}}\,.
\]
\end{lem}

For $h=\frac{c}{n^{2}}$ and $\delta=\frac{c+t}{n}$ for parameter
$c,t>0$, we can satisfy $\delta\geq hn$. Another is the following:
\begin{lem}
[{\cite[Lemma 24]{KZ25Renyi}}] \label{lem:gaussian-helper} Let
$\K\subset\Rn$ be a convex body containing a unit ball $B_{1}(0)$.
For $\tau=\tfrac{\sigma^{2}}{h+\sigma^{2}}$ and $s>0$,
\[
\tau^{-n/2}\int_{\K_{s}}\exp\bpar{-\frac{1}{2\tau\sigma^{2}}\,\norm z^{2}}\,\D z\leq2\exp(hn^{2}+sn)\int_{\K}\exp\bpar{-\frac{1}{2\sigma^{2}}\,\norm z^{2}}\,\D z\,.
\]
\end{lem}

\paragraph{(1) Failure probability.}

Just as in the analysis of $\psunif$, we adapt the analysis of $\psgauss$
from \cite{KZ25Renyi}. For an initial distribution $\nu\ll\mu$,
the failure probability can be bounded using the Cauchy--Schwarz
as follows:
\[
\E_{\nu_{h}}[(1-\ell)^{N}]\leq\sqrt{\E_{\mu_{h}}[(1-\ell)^{2N}]}\,\Bnorm{\frac{\D\nu}{\D\mu}}_{L^{2}(\mu)}=:M_{2}\sqrt{\E_{\mu_{h}}[(1-\ell)^{2N}]}\,.
\]
Then,
\[
\int_{\Rn}(1-\ell)^{2N}\,\D\mu_{h}=\int_{R^{c}}\cdot+\int_{R\cap[\ell\geq N^{-1}\log(3kM_{2}/\eta)]}\cdot+\int_{R\cap[\ell<N^{-1}\log(3kM_{2}/\eta)]}\cdot=:\msf A+\msf B+\msf C\,,
\]
where in $(i)$ below, by using change of variables and Lemma~\ref{lem:gaussian-helper},
\begin{align*}
\msf A & \leq\mu^{Y}(R^{c})\underset{\text{Lemma \ref{lem:gaussian-effective}}}{\leq}\exp\bpar{-\frac{t^{2}}{2c}+\frac{3c}{2}}\,,\\
\msf B & \leq\int_{[\ell\geq N^{-1}\log(3kM_{2}/\eta)]}\exp(-2\ell N)\,\D\mu^{Y}\leq\bpar{\frac{\eta}{3kM_{2}}}^{2}\,,\\
\msf C & \leq\int_{R\cap[\ell<N^{-1}\log(3kM_{2}/\eta)]}\D\mu^{Y}=\int_{R\cap[\ell<N^{-1}\log(3kM_{2}/\eta)]}\frac{\tau^{n/2}\ell(y)\exp(-\frac{1}{2\tau\sigma^{2}}\,\norm{y_{\tau}}^{2})}{\int_{\mc K}\exp(-\frac{1}{2\sigma^{2}}\,\norm z^{2})\,\D z}\,\D y\\
 & \leq\frac{\log(3kM_{2}/\eta)}{N}\,\frac{\int_{R}\tau^{n/2}\exp(-\frac{1}{2\tau\sigma^{2}}\norm{y_{\tau}}^{2})\,\D y}{\int_{\mc K}\exp(-\frac{1}{2\sigma^{2}}\norm z^{2})\,\D z}\underset{(i)}{\leq}\frac{\log(3kM_{2}/\eta)}{N}\,2\exp(hn^{2}+\delta n)\\
 & \leq\frac{2e^{2c+t}}{N}\log\frac{3kM_{2}}{\eta}\,,
\end{align*}
For $Z=\frac{16kM_{2}}{\eta}$, we choose $c=\frac{\log\log Z}{10\log Z}$,
$t=\log\log Z$, and $N=Z^{2}\log^{3}Z$, under which each final bound
is bounded by $(\frac{\eta}{3kM_{2}})^{2}$. Therefore, the failure
probability per iteration is $\eta/k$.

\paragraph{(2) Complexity of the backward step.}

Similar to the uniform-sampling case, let $p=1+\alpha^{-1}$ and $q=1+\alpha$
with $\alpha=\log N$. Then,
\[
\E_{\nu_{h}}\bbrack{\frac{1}{\ell}\wedge N}=\int_{R\cap[\ell\geq N^{-p}]}\cdot+\int_{R\cap[\ell<N^{-p}]}\cdot+\int_{R^{c}}\cdot=:\msf A+\msf B+\msf C\,,
\]
where
\begin{align*}
\msf A & \le\Bpar{\int_{R\cap[\ell\geq N^{-p}]}\frac{1}{\ell^{p}}\,\D\mu_{h}}^{1/p}\,\Bnorm{\frac{\D\nu}{\D\mu}}_{L^{q}(\mu)}=M_{q}\Bpar{\frac{\int_{R\cap[\ell\geq N^{-p}]}\ell^{p-1}(y)\,\tau^{n/2}\exp(-\frac{1}{2\tau\sigma^{2}}\,\norm{y_{\tau}}^{2})\,\D y}{\int_{\mc K}\exp(-\frac{1}{2\sigma^{2}}\,\norm x^{2})\,\D x}}^{1/p}\\
 & \leq N^{1/\alpha}M_{q}\,\Bpar{\frac{\int_{R}\tau^{n/2}\exp(-\frac{1}{2\tau\sigma^{2}}\,\norm{y_{\tau}}^{2})\,\D y}{\int_{\mc K}\exp(-\frac{1}{2\sigma^{2}}\,\norm x^{2})\,\D x}}^{1/p}\leq2eM_{q}\exp(hn^{2}+\delta n)\leq6M_{q}e^{2c+t}\,,\\
\msf B & \leq N\int_{R\cap[\ell<N^{-p}]}\frac{\D\nu_{h}}{\D\mu_{h}}\,\D\mu_{h}\leq N\,\Bpar{\int_{R\cap[\ell<N^{-p}]}\D\mu_{h}}^{1/p}\,\Bnorm{\frac{\D\nu}{\D\mu}}_{L^{q}(\mu)}\\
 & \leq NM_{q}\Bpar{\frac{\int_{R\cap[\ell<N^{-p}]}\tau^{n/2}\,\ell^{p}(y)\exp(-\frac{1}{2\tau\sigma^{2}}\,\norm{y_{\tau}}^{2})\,\D y}{\int_{\mc K}\exp(-\frac{1}{2\sigma^{2}}\,\norm x^{2})\,\D x}}^{1/p}\leq M_{q}\Bpar{\frac{\int_{R}\tau^{n/2}\exp(-\frac{1}{2\tau\sigma^{2}}\,\norm{y_{\tau}}^{2})\,\D y}{\int_{\mc K}\exp(-\frac{1}{2\sigma^{2}}\,\norm x^{2})\,\D x}}^{1/p}\\
 & \leq2M_{q}e^{2c+t}\,,\\
\msf C & \leq N\int_{R^{c}}\frac{\D\nu_{h}}{\D\mu_{h}}\,\D\mu_{h}\leq NM_{2}\sqrt{\mu_{h}(R^{c})}\,.
\end{align*}
Therefore,
\begin{align*}
\E_{\nu_{h}}\bbrack{\frac{1}{\ell}\wedge N} & \leq M_{q}\,\Bpar{8e^{2c+t}+N\exp\bpar{-\frac{t^{2}}{4c}+\frac{3c}{4}}}\leq M_{q}\,\Bpar{8e^{2c+t}+N\exp\bpar{-\frac{t^{2}}{8c}}}\,\\
 & \leq M_{q}\,(2e\log Z+e\log^{3}Z)\leq10M_{q}\log^{3}Z\,.
\end{align*}

We now prove the main theorem for this section.
\begin{proof}
[Proof of Theorem~\ref{thm:body-gauss-samp}] By \eqref{eq:LSI-contraction},
$\psgauss$ achieves $\eu R_{2}(\nu_{k}\mmid\mu)\leq\veps$ after
\[
k\gtrsim\bpar{h^{-1}\clsi(\mu)\vee1}\log\frac{\log M_{2}}{\veps}
\]
iterations. Under the choice of $h=\frac{1}{10n^{2}\log Z}$ with
$Z=\frac{16kM_{2}}{\eta}$, it suffices to run $\psgauss$ for 
\[
k=\Otilde\bpar{n^{2}\clsi(\mu)\log^{3}\frac{M_{2}}{\eta\veps}}
\]
iterations. Therefore, the total expected number of queries over $k$
iterations is 
\[
\Otilde\bpar{M_{q}n^{2}\clsi(\mu)\log^{6}\frac{M_{2}}{\eta\veps}}\,,
\]
where $q=1+\log N\leq6\log\frac{16kM_{2}}{\eta}$.

In general, we can use the bound of $\clsi(\mu)=\clsi(\pi\gamma_{\sigma^{2}})\leq\sigma^{2}$
\cite{BGL14analysis}. In this case, the $\psgauss$ requires $k=\Otilde(n^{2}\sigma^{2}\log^{3}\frac{M_{2}}{\eta\veps})$
iterations, with total query complexity of
\[
\Otilde\bpar{M_{q}n^{2}\sigma^{2}\log^{6}\frac{1}{\eta\veps}}\,.
\]
When $\sigma^{2}\gtrsim D\lda^{1/2}\log^{2}n\log^{2}\nicefrac{D^{2}}{\lda}$
for $\lda=\norm{\cov\pi}$, we can use the improved bound, $\clsi(\mu^{X})\lesssim D\lda^{1/2}\log n$
(which will be proven by Theorem~\ref{thm:lsi-general-bound} and~\ref{thm:cov-Gauss}
in \S\ref{app:LSI-interpolation-SL}). In this case, $\psgauss$
suffices to iterate $k=\Otilde(n^{2}D\lda^{1/2}\log^{3}\frac{M_{2}}{\eta\veps})$
times, with total query complexity of
\[
\Otilde\bpar{M_{q}n^{2}D\lda^{1/2}\log^{6}\frac{1}{\eta\veps}}\,,
\]
which completes the proof.
\end{proof}

\section{Improved logarithmic Sobolev constants\label{sec:improved-FI}}

Throughout this section, $\pi$ denotes a logconcave probability measure
over $\Rn$ with first moment $R:=\E_{\pi}\norm{\cdot}$, covariance
matrix $\Sigma:=\cov\pi$, and operator norm $\lda:=\norm{\Sigma}$.
For $t>0$, we use $\pi\gamma_{t}$ to denote the $t^{-1}$-strongly
logconcave distribution obtained by weighting $\pi$ with a Gaussian
$\gamma_{t}$ (i.e., $\D(\pi\gamma_{t})(x)\propto\pi(x)\,\gamma_{t}(x)\,\D x$).
We also denote $\Sigma_{t}:=\cov\pi\gamma_{t}$ and $\lda_{t}:=\norm{\Sigma_{t}}$.

We begin by proving a bound on $\clsi(\pi)$ for logconcave $\pi$
supported with diameter $D>0$, which interpolates the $D$-bound
for isotropic logconcave distributions \cite{LV24eldan} and the classical
$D^{2}$-bound for general logconcave distributions \cite{FK99lsi}.
\begin{thm}
[Restatement of Theorem~\ref{thm:intro-LSI-interpolation}] \label{thm:lsi-general-bound}
For a logconcave distribution $\pi$ with support of diameter $D>0$
in $\Rn$,
\[
\clsi(\pi)\lesssim\max\{D\lda^{1/2},D^{2}\wedge\lda\log^{2}n\}\,.
\]
It also holds that $\clsi(\pi)\lesssim D\cpi^{1/2}(\pi)\lesssim D\lda^{1/2}\log^{1/2}n$.
\end{thm}

Our second result addresses how a Gaussian factor affects the largest
eigenvalue of the covariance matrix of a logconcave distribution.
\begin{thm}
[Restatement of Theorem~\ref{thm:intro-cov-str-LC}] \label{thm:cov-Gauss}
For a logconcave distribution $\pi$ with $R=\E_{\pi}\norm{\cdot}$
in $\Rn$, if $\sigma^{2}\gtrsim R\lda^{1/2}\log^{2}n\log^{2}\frac{R^{2}}{\lda}$,
then
\[
\lda_{\sigma^{2}}\lesssim\lda\,.
\]
\end{thm}

Combining these two together with the classical result $\clsi(\pi\gamma_{\sigma^{2}})\leq\sigma^{2}$,
we obtain the following bound for strongly logconcave probability
measures with compact support.
\begin{cor}
[Restatement of Corollary~\ref{cor:intro-new-LSI-str-LC}] \label{cor:lsi-global-bound}
Let $\pi$ be a logconcave distribution in $\Rn$ with support of
diameter $D>0$. Then, for any $\sigma^{2}>0$,
\[
\clsi(\pi\gamma_{\sigma^{2}})\lesssim D\lda^{1/2}\log^{2}n\log^{2}\frac{D}{\lda^{1/2}}\,.
\]
\end{cor}

\subsection{Log-Sobolev constant for logconcave distributions with compact support
\label{subsec:LSI-interpolation}}

It follows from the curvature-dimension condition \cite{BGL14analysis}
that $\clsi(\gamma_{h}|_{\K})\leq\clsi(\gamma_{h})=h$ (i.e., truncating
a strongly logconcave measure to a convex set does not worsen the
LSI constant), but how about the other natural possible inequality
$\clsi(\gamma_{h}|_{\K})\leq\clsi(\ind_{\K})$ (i.e., Multiplying
a radially symmetric Gaussian does not increase the LSI constant)? 

Why do we care about this? To motivate this question, suppose $\pi\propto\ind_{\K}$
is an isotropic logconcave distribution with support of diameter $\Theta(n)$,
so $\clsi(\pi)\lesssim n$ \cite[Theorem 49]{LV24eldan}. Now consider
$h\asymp n^{2}$. Then, the distribution $\gamma_{h}|_{\K}$ is a
$\Theta(1)$-perturbation to $\pi$, so by the Holley--Stroock perturbation
principle, we obtain $\clsi(\gamma_{h}|_{\K})\lesssim\clsi(\pi)\lesssim n$.
Note that the straightforward bound $\clsi(\gamma_{h}|_{\K})\leq h$
only yields an $\O(n^{2})$-bound. This suggests that as $h$ increases
from $n$ to $n^{2}$, if the LSI constant $\clsi(\gamma_{h}|_{\K})$
initially increases, it must start to decrease somewhere to satisfy
$\clsi(\gamma_{n^{2}}|_{\K})\lesssim n$. This naturally leads us
to ask if $\clsi(\gamma_{h}|_{\K})\lesssim n$ for \emph{all} $h>0$.

This question can be phrased more generally as follows.
\begin{question}
\label{ques:LSI-question} For any logconcave distribution $\pi$
with compact support, and any $h>0$, do we have $\clsi(\pi\gamma_{h})\lesssim\clsi(\pi)\,?$ 
\end{question}

We may first ask whether $\clsi(\pi\gamma_{h})$ remains close to
the known bound on $\clsi(\pi)$. In \S\ref{subsec:FI-together},
we show that this is indeed the case. Before proceeding, however,
we mention what is currently known about $\clsi(\pi)$.

The classical upper bound for the LSI constant of any logconcave distribution
$\pi$ with compact support of diameter $D>0$ is $\O(D^{2})$, while
for isotropic logconcave distributions, a stronger bound $\clsi(\pi)\lesssim D$
holds. In the context of \eqref{eq:pi}, the known bound $\cpi(\pi)\lesssim\lda\log n$
nearly interpolates the $\O(\log n)$ bound for isotropic cases and
$\O(D^{2})$ for general cases (noting that $\lda\leq D^{2}$). This
motivates the following mathematical question, which we will address
in this section.
\begin{question*}
What is a general bound on \eqref{eq:lsi} that interpolates those
two known bounds on \eqref{eq:lsi} for logconcave distributions with
compact support?
\end{question*}

\subsubsection{A na\"ive approach: Interpolation via a Lipschitz map \label{subsec:naive-LSI}}

As sketched earlier in \S\ref{sec:techniques}, one justification
for focusing on the Poincar\'e constant of \emph{isotropic} logconcave
distributions is via the affine map $T:x\mapsto\Sigma^{1/2}x$, which
has Lipschitzness $\norm{\Sigma}^{1/2}$. 

This approach does not work well for \eqref{eq:lsi}. Indeed, let
$\nu:=(T^{-1})_{\#}\pi$ be the pushforward of $\pi$ under the inverse
affine map, so that $\nu$ is isotropic. For any test function $f:\Rn\to\R$,
we have
\begin{align*}
\ent_{\pi}[f^{2}] & =\ent_{\nu}[(f\circ\Sigma^{1/2})^{2}]\leq C_{\msf{LSI}}(\nu)\,\E_{\nu}[\norm{\nabla(f\circ\Sigma^{1/2})}^{2}]\\
 & \leq C_{\msf{LSI}}(\nu)\,\E_{\nu}[\norm{\Sigma^{1/2}}^{2}\norm{\nabla f\circ\Sigma^{1/2}}^{2}]\leq C_{\msf{LSI}}(\nu)\,\norm{\Sigma}\,\E_{\pi}[\norm{\nabla f}^{2}]\,,
\end{align*}
which implies that
\[
C_{\msf{LSI}}(\pi)\leq C_{\msf{LSI}}(\nu)\,\norm{\Sigma}\,.
\]
This bound has several downsides. First, the best-known bound for
isotropic logconcave distributions is $\clsi(\nu)\lesssim D$, and
this \emph{cannot be} improved in general \cite[Lemma 55]{LV24eldan}.
Moreover, isotropy implies $D\geq\sqrt{n}$, so we cannot hope for
a better general bound on $\clsi(\nu)$ than $\sqrt{n}$. Second,
since $\norm{\Sigma}=\sup_{v\in\mbb S^{n-1}}\E_{\pi}[(v^{\T}(X-\E_{\pi}X))^{2}]\leq D^{2}$,
it is unclear how the RHS of the inequality could even recover the
classical $\O(D^{2})$ bound for $\clsi(\pi)$. Thus, this affine
reduction fails to interpolate the known two bounds and motivates
the need for a more intrinsic analysis.

\subsubsection{A better approach via Gaussian concentration\label{subsec:Gaussian-concentration}}

We prove the following two bounds presented in Theorem~\ref{thm:lsi-general-bound}:
\[
\clsi(\pi)\lesssim\min\bbrace{D\cpi^{1/2}(\pi),\max\{D\lda^{1/2},D^{2}\wedge\lda\log^{2}n\}}\,.
\]
Inspired by \cite[\S8]{KL24isop}, we present a simple proof using
a known equivalence between Gaussian concentration and \eqref{eq:lsi}
for logconcave measures.

We first define a concentration function of a measure.
\begin{defn}
[Concentration] \label{def:concentration} For a probability measure
$\pi$ over $\Rn$, consider the concentration function of $\pi$
defined as 
\[
\alpha_{\pi}(r)=\sup_{S:\pi(S)\geq1/2}\pi(S^{c})\,.
\]
The measure $\pi$ is said to satisfy \emph{exponential concentration}
with constant $\cexp(\pi)$ if
\[
\alpha_{\pi}(r)\leq2\exp\bpar{-\frac{r}{\cexp(\pi)}}\quad\text{for all }r\geq0\,,
\]
and is said to satisfy \emph{Gaussian concentration} with constant
$\cgauss(\pi)$ if
\[
\alpha_{\pi}(r)\leq2\exp\bpar{-\frac{r^{2}}{\cgauss(\pi)}}\quad\text{for all }r\geq0\,.
\]
\end{defn}

Milman \cite{milman10isoperimetric} established that for a logconcave
distribution, exponential concentration is equivalent to \eqref{eq:pi},
and Gaussian concentration is equivalent to \eqref{eq:lsi}.
\begin{thm}
\label{thm:LSI-concentration-Milman} For any logconcave probability
measure $\pi$, it holds that $\clsi(\pi)\asymp\cgauss(\pi)$.
\end{thm}

To use this result, we recall two known results on exponential concentration.
The first is a standard result on exponential concentration under
\eqref{eq:pi}. Precisely, this is a combination of concentration
of Lipschitz functions (around its median or mean) {\cite[\S4.4.3]{BGL14analysis}}
and its equivalence with exponential concentration.
\begin{thm}
\label{thm:PI-concentration-ftn} For a probability measure $\pi$
with $\cpi(\pi)<\infty$, there exists a universal constant $c>0$
such that
\[
\alpha_{\pi}(r)\leq2\exp\bpar{-\frac{r}{c\cpi^{1/2}(\pi)}}\,.
\]
\end{thm}

There is another bound established by \cite[Theorem 1.1]{bizeul24measures}
using Eldan's stochastic localization (see another proof based on
this technique in \S\ref{app:LSI-interpolation-SL}).
\begin{thm}
\label{thm:bizeul-concentration} For any logconcave probability measure
$\pi$ over $\Rn$, there exists a universal constant $c>0$ such
that
\[
\alpha_{\pi}(r)\leq2\exp\Bpar{-c\min\bpar{\frac{r}{\norm{\cov\pi}^{1/2}},\frac{r^{2}}{\norm{\cov\pi}\log^{2}n}}}\,.
\]
\end{thm}

We now prove the desired LSI bound in a simple way by combining these
two bounds.
\begin{proof}
[Proof of Theorem~\ref{thm:lsi-general-bound}] Since the diameter
of the support is $D$, it clearly holds that $\alpha_{\pi}(r)=0$
for $r>D$. For $r\leq D$, using $r/D\leq1$ and letting $\lda=\norm{\cov\pi}$,
by Theorem~\ref{thm:bizeul-concentration},
\[
\alpha_{\pi}(r)\leq2\exp\Bpar{-c\min\bpar{\frac{r^{2}}{D\lda^{1/2}},\frac{r^{2}}{\lda\log^{2}n}}}=2\exp\Bpar{-\frac{cr^{2}}{D\lda^{1/2}\vee\lda\log^{2}n}}\,.
\]
By Theorem~\ref{thm:LSI-concentration-Milman}, since $\pi$ is logconcave,
\[
\clsi(\pi)\asymp\cgauss(\pi)\lesssim D\lda^{1/2}\vee\lda\log^{2}n\,.
\]
Combining with $\clsi(\pi)\lesssim D^{2}$, we can deduce that
\begin{align*}
\clsi(\pi) & \lesssim D^{2}\wedge(D\lda^{1/2}\vee\lda\log^{2}n)=\max\{D\lda^{1/2}\wedge D^{2},D^{2}\wedge\lda\log^{2}n\}\\
 & =\max\{D\lda^{1/2},D^{2}\wedge\lda\log^{2}n\}\,.
\end{align*}

Regarding the second bound, it follows from Theorem~\ref{thm:PI-concentration-ftn}
that for some universal constant $c>0$,
\[
\alpha_{\pi}(r)\leq2\exp\bpar{-\frac{r}{c\cpi^{1/2}(\pi)}}.
\]
Using a similar argument above, we have
\[
\alpha_{\pi}(r)\leq2\exp\bpar{-\frac{r^{2}}{cD\cpi^{1/2}(\pi)}}\,.
\]
Thus, by the equivalence between \eqref{eq:lsi} and Gaussian concentration
(Theorem~\ref{thm:LSI-concentration-Milman}),
\[
\clsi(\pi)\asymp\cgauss(\pi)\lesssim D\cpi^{1/2}(\pi)\leq D\lda^{1/2}\log^{1/2}n\,,
\]
where the last inequality follows from $\cpi(\pi)\lesssim\lda\log n$.
\end{proof}
\begin{rem}
[Comparison of two bounds] \label{rem:LSI-comparison}The first bound
successfully interpolates the $\O(D)$-bound of isotropic logconcave
distributions and $\O(D^{2})$-bound of general logconcave ones. However,
it cannot attain $\O(D\lda^{1/2})$ even if the KLS conjecture is
true. On the other hand, the second bound almost interpolates those
two bounds (up to $\log^{1/2}n$), while it attains $\O(D\lda^{1/2})$
when the KLS conjecture is true.
\end{rem}

\subsection{Covariance of strongly logconcave distributions \label{subsec:cov-str-LC}}

We have just shown that $\clsi(\pi)\lesssim D\lda^{1/2}$ and $\clsi(\pi\gamma_{h})\lesssim D\lda_{h}^{1/2}$
(ignoring logarithmic factors). How does this new bound compare to
the previous two bounds of $\clsi(\pi\gamma_{h})\leq h\wedge\O(D^{2})$,
from strong logconcavity and from the support diameter respectively?

A standard upper bound on the operator norm $\lda_{h}$ can be obtained
via 
\[
\lda_{h}\leq\cpi(\pi\gamma_{h})\leq h\,,
\]
where the second inequality follows from the Brascamp--Lieb inequality.
Plugging this into our bound yields $\clsi(\pi\gamma_{h})\lesssim Dh^{1/2}$,
which is the geometric mean of the previous two bounds ($h$ and $D^{2}$),
and making no essential gain. 

However, for sufficiently large $h$, one can expect that $\pi\gamma_{h}$
is closer to $\pi$, suggesting that $\norm{\Sigma_{h}}$ may be comparable
to $\norm{\Sigma}$. In fact, when $h\geq D^{2}$, the distribution
$\pi\gamma_{h}$ becomes a $\Theta(1)$-perturbation of $\pi$, and
the Holley--Stroock perturbation principle gives $\clsi(\pi\gamma_{h})\lesssim\clsi(\pi)\lesssim D\lda^{1/2}$.
 If this bound is true for \emph{any} $h$, and particularly when
$\pi$ is isotropic, then this would already yield a substantial gain,
since $D\lda^{1/2}=D$ is smaller than the na\"ive $h$-bound, which
only gives $\Omega(D^{2})$.

Our main goal  at this point is to show $\lda_{h}\lesssim\lda$ for
$h\gtrsim\lda^{1/2}R\cpi^{2}(n)\log^{2}\nicefrac{R^{2}}{\lda}$ (Theorem~\ref{thm:cov-Gauss}),
where $\cpi(n)$ is the largest possible Poincar\'e constant of $n$-dimensional
isotropic logconcave distributions. Since $\cpi(n)\lesssim\log n$
\cite{Klartag23log}, a sufficient condition is $h\gtrsim\lda^{1/2}R\log^{2}n\log^{2}\nicefrac{R^{2}}{\lda}$.
We remark that the logarithmic factor in the guarantee can be further
improved by using a thin-shell estimate instead of $\cpi(n)$ and
tightening some computations in the proof.
\begin{proof}
[Proof of Theorem~\ref{thm:cov-Gauss}] We start off by scaling
the coordinate system via $T:x\mapsto\lda^{-1/2}x$, so that $\nu:=T_{\#}\pi$
satisfies $\norm{\cov\nu}=1$. We show that if $\eta\gtrsim\E_{\nu}\norm Y\log^{2}n\log^{2}\E_{\nu}\norm Y$,
then $\norm{\cov\nu_{\eta}}\lesssim1$ for $\nu_{\eta}=\nu\gamma_{\eta}$.
Then, by scaling back through $T^{-1}$, we can deduce the main claim
in the theorem (due to $\E_{\nu}\norm Y=\lda^{-1/2}\,\E_{\pi}\norm X$
and $h=\eta\lda$). Throughout the proof, we use $c$ and $c_{1}$
to denote positive universal constants.

By rotating the coordinate system, it suffices to bound that for $\mu$,
the barycenter of $\nu$,
\[
Q:=\E_{\nu_{\eta}}[(Y-\mu)_{1}^{2}]\lesssim1\,.
\]
To this end, we instead bound
\[
Q=\frac{\E_{\nu}\bbrack{(Y-\mu)_{1}^{2}\,\exp(-\frac{\norm Y^{2}}{2\eta})}}{\E_{\nu}\bbrack{\exp(-\frac{\norm Y^{2}}{2\eta})}}=:\frac{N}{D}\,,
\]
where $N$ and $D$ refer to the numerator and denominator, respectively.

Recall the classical Lipschitz concentration that for $1$-Lipschitz
$f$: for any $t\geq0$,
\begin{equation}
\nu(f-\E_{\nu}f\geq t)\leq3\exp\bpar{-\frac{t}{\cpi^{1/2}(\nu)}}\,.\label{eq:PI-Lips-concent}
\end{equation}
Let $R:=\E_{\nu}\norm{\cdot}_{2}(\geq1)$ and $\cpi:=\cpi(n)$. Noting
$\cpi(\nu)\lesssim\cpi$, and taking $f(\cdot)=\norm{\cdot}_{2}$
and $t=c_{1}\cpi^{1/2}\log R$, we have that 
\[
\nu(\norm Y\in[R\pm t])\geq1-6\exp(-c_{2}\log R)\,,
\]
where $[R\pm t]:=[R-t,R+t]$. Hence, the thin-shell $\mc S:=\{y\in\Rn:\norm y\in[R\pm t]\}$
takes up more than a half of $\nu$-measure by taking $c_{1}$ large
enough. Since $\eta\gtrsim R\cpi^{2}\log^{2}R\gtrsim t\,(R\vee t)=Rt\vee t^{2}$,
\[
D\geq\E_{\nu}\bbrack{\exp\bpar{-\frac{\norm Y^{2}}{2\eta}}\,\ind_{\mc S}}\geq\exp\bpar{-\frac{R^{2}+2Rt+t^{2}}{2\eta}}\,\nu(\mc S)\gtrsim\exp\bpar{-\frac{R^{2}}{2\eta}}\,.
\]

As for the upper bound on the numerator $N$, we decompose $N$ as
follows:
\[
N=\underbrace{\E_{\nu}\bbrack{(Y-\mu)_{1}^{2}\,\exp\bpar{-\frac{\norm Y^{2}}{2\eta}}\,\ind_{\mc S}}}_{=:\msf{(I)}}+\underbrace{\E_{\nu}\bbrack{(Y-\mu)_{1}^{2}\,\exp\bpar{-\frac{\norm Y^{2}}{2\eta}}\,\ind_{\mc S^{c}}}}_{=:\msf{(II)}}\,.
\]
For $(\msf{I)}$, we have
\begin{align*}
\E_{\nu}\bbrack{(Y-\mu)_{1}^{2}\,\exp\bpar{-\frac{\norm Y^{2}}{2\eta}}\,\ind_{\mc S}} & \leq\exp\bpar{-\frac{R^{2}}{2\eta}+\frac{Rt}{\eta}}\,\E_{\nu}\bbrack{(Y-\mu)_{1}^{2}\,\ind_{\mc S}}\lesssim\exp\bpar{-\frac{R^{2}}{2\eta}}\,\E_{\nu}[(Y-\mu)_{1}^{2}]\\
 & \leq\exp\bpar{-\frac{R^{2}}{2\eta}}\,.
\end{align*}
For $\msf{(II)}$, since $\norm{\mu}\leq\E_{\nu}\norm{\cdot}=R$,
we have
\begin{align*}
 & \E_{\nu}\bbrack{(Y-\mu)_{1}^{2}\,\exp\bpar{-\frac{\norm Y^{2}}{2\eta}}\,\ind_{\mc S^{c}}}\\
 & =\E_{\nu}\bbrack{(Y-\mu)_{1}^{2}\,\exp\bpar{-\frac{\norm Y^{2}}{2\eta}}\,(\ind_{\{\norm Y\leq R-t\}}+\ind_{\{\norm Y\geq R+t\}})}\\
 & \lesssim R^{2}\,\E_{\nu}\bbrack{\exp\bpar{-\frac{\norm Y^{2}}{2\eta}}\,\ind_{\{\norm Y\leq R-t\}}}+\E_{\nu}\bbrack{(R^{2}+\norm Y^{2})\exp\bpar{-\frac{\norm Y^{2}}{2\eta}}\,\ind_{\{\norm Y\geq R+t\}}}\\
 & \leq R^{2}\,\E_{\nu}\bbrack{\exp\bpar{-\frac{\norm Y^{2}}{2\eta}}\,\ind_{\{\norm Y\leq R-t\}}}+2\,\E_{\nu}\bbrack{\norm Y^{2}\,\exp\bpar{-\frac{\norm Y^{2}}{2\eta}}\,\ind_{\{\norm Y\geq R+t\}}}\,.
\end{align*}
By the co-area formula, for the $(n-1)$-dimensional Hausdorff measure
$\mc H^{n-1}$,
\begin{align*}
\E_{\nu}\bbrack{\exp\bpar{-\frac{\norm Y^{2}}{2\eta}}\,\ind_{\{\norm Y\leq R-t\}}} & =\int_{0}^{R-t}e^{-\frac{r^{2}}{2\eta}}\underbrace{\int_{\de B_{r}}\nu(x)\,\D\mc H^{n-1}(x)}_{=:\nu_{n-1}(\de B_{r})}\,\D r\,,\\
\E_{\nu}\bbrack{\norm Y^{2}\exp\bpar{-\frac{\norm Y^{2}}{2\eta}}\,\ind_{\{\norm Y\geq R+t\}}} & =\int_{R+t}^{\infty}r^{2}e^{-\frac{r^{2}}{2\eta}}\,\nu_{n-1}(\de B_{r})\,\D r\,.
\end{align*}
Using integration by parts for the first expectation,
\begin{align*}
 & \int_{0}^{R-t}e^{-\frac{r^{2}}{2\eta}}\,\nu_{n-1}(\de B_{r})\,\D r\\
 & =[e^{-\frac{r^{2}}{2\eta}}\,\nu(B_{r})]_{0}^{R-t}+\int_{0}^{R-t}\frac{r}{\eta}\,e^{-\frac{r^{2}}{2\eta}}\,\nu(B_{r})\,\D r\\
 & \underset{\eqref{eq:PI-Lips-concent}}{\leq}3\exp\bpar{-\frac{(R-t)^{2}}{2\eta}}\,\exp\bpar{-\frac{t}{\cpi^{1/2}}}+\eta^{-1}\int_{0}^{R-t}re^{-\frac{r^{2}}{2\eta}}\,\exp\bpar{-\frac{R-r}{\cpi^{1/2}}}\,\D r\\
 & \lesssim\exp\bpar{-\frac{R^{2}}{2\eta}}\,\exp(-c_{1}\log R)+\frac{R}{\eta}\int_{t}^{R}\exp\bpar{-\frac{(R-r)^{2}}{2\eta}}\,\exp\bpar{-\frac{r}{\cpi^{1/2}}}\,\D r\\
 & \leq\exp\bpar{-\frac{R^{2}}{2\eta}}\,R^{-c_{1}}+\frac{R}{\eta}\,\exp\bpar{-\frac{R^{2}}{2\eta}}\int_{t}^{R}\exp\bpar{\frac{Rr}{\eta}-\frac{r}{\cpi^{1/2}}}\,\D r\,.
\end{align*}
Using $\eta\gtrsim Rt$ and taking $c_{1}$ large enough, 
\begin{align*}
\int_{t}^{R}\exp\bpar{\frac{Rr}{\eta}-\frac{r}{\cpi^{1/2}}}\,\D r & \leq\int_{t}^{R}\exp\Bpar{r\,\bpar{\frac{c_{3}}{c_{1}\cpi^{1/2}}-\frac{1}{\cpi^{1/2}}}}\,\D r\leq\int_{t}^{R}\exp\bpar{-\frac{r}{2\cpi^{1/2}}}\,\D r\\
 & \leq2\cpi^{1/2}\exp\bpar{-\frac{t}{2\cpi^{1/2}}}\leq2\cpi^{1/2}R^{-c_{1}/2}\,.
\end{align*}
Substituting these back to above, for large enough $c_{1}>0$, we
obtain that
\[
\int_{0}^{R-t}e^{-\frac{r^{2}}{2\eta}}\,\nu_{n-1}(B_{r})\,\D r\leq\exp\bpar{-\frac{R^{2}}{2\eta}}\,\bpar{R^{-c_{1}}+\frac{2R\cpi^{1/2}R^{-c_{1}/2}}{\eta}}\lesssim R^{-c_{1}/2}\exp\bpar{-\frac{R^{2}}{2\eta}}\,.
\]

We can bound the second expectation in a similar way:
\begin{align*}
 & \int_{R+t}^{\infty}r^{2}e^{-\frac{r^{2}}{2\eta}}\,\nu_{n-1}(\de B_{r})\,\D r\\
 & =\bigl[r^{2}e^{-\frac{r^{2}}{2\eta}}\,\nu(B_{r})\bigr]_{R+t}^{\infty}+\int_{R+t}^{\infty}\bpar{\frac{r^{3}}{\eta}-2r}\,e^{-\frac{r^{2}}{2\eta}}\,\nu(B_{r})\,\D r\\
 & \lesssim\lim_{r\to\infty}r^{2}\exp\bpar{-\frac{r^{2}}{2\eta}-\frac{r-R}{\cpi^{1/2}}}+\eta^{-1}\int_{R+t}^{\infty}r^{3}\exp\bpar{-\frac{r^{2}}{2\eta}-\frac{r-R}{\cpi^{1/2}}}\,\D r\\
 & =\eta^{-1}\int_{t}^{\infty}(r+R)^{3}\exp\bpar{-\frac{(r+R)^{2}}{2\eta}-\frac{r}{\cpi^{1/2}}}\,\D r\\
 & \leq4\eta^{-1}\exp\bpar{-\frac{R^{2}}{2\eta}}\int_{t}^{\infty}(r^{3}+R^{3})\,\exp\bpar{-\frac{r}{\cpi^{1/2}}}\,\D r\,.
\end{align*}
We can bound each integral as follows:
\begin{align*}
\int_{t}^{\infty}R^{3}\exp\bpar{-\frac{r}{\cpi^{1/2}}}\,\D r & =R^{3}\cpi^{1/2}\exp\bpar{-\frac{t}{\cpi^{1/2}}}\leq R^{3-c_{1}}\cpi^{1/2}\,,\\
\int_{t}^{\infty}r^{3}\exp\bpar{-\frac{r}{\cpi^{1/2}}}\,\D r & \underset{(i)}{\lesssim}t^{3}\cpi^{1/2}\exp\bpar{-\frac{t}{\cpi^{1/2}}}\leq t^{3}R^{-c_{1}}\cpi^{1/2}\,,
\end{align*}
where in $(i)$ we used properties of the incomplete Gamma function
to compute $\int_{a}^{\infty}x^{3}e^{-x/b}\,\D x=6b^{4}e^{-a/b}\,(1+\frac{a}{b}+\frac{a^{2}}{2b^{2}}+\frac{a^{3}}{6b^{3}})$.
Thus,
\[
\int_{R+t}^{\infty}r^{2}e^{-\frac{r^{2}}{2\eta}}\,\nu_{n-1}(\de B_{r})\,\D r\lesssim\exp\bpar{-\frac{R^{2}}{2\eta}}\,\frac{(R^{3-c_{1}}+t^{3}R^{-c_{1}})\,\cpi^{1/2}}{\eta}\lesssim\exp\bpar{-\frac{R^{2}}{2\eta}}\,.
\]
Therefore, for large enough $c_{1}>0$,
\[
\msf{II}\lesssim R^{2}\times R^{-c_{1}/2}\exp\bpar{-\frac{R^{2}}{2\eta}}+\exp\bpar{-\frac{R^{2}}{2\eta}}\lesssim\exp\bpar{-\frac{R^{2}}{2\eta}}\,,
\]
and thus
\[
N=\msf I+\msf{II}\lesssim\exp\bpar{-\frac{R^{2}}{2\eta}}\,.
\]

Combining all previous bounds,
\[
Q=\frac{N}{D}\lesssim1\,,
\]
which completes the proof.
\end{proof}

\subsection{Functional inequalities for strongly logconcave distributions with
compact support\label{subsec:FI-together}}

We  examine \eqref{eq:lsi} and \eqref{eq:pi} for strongly logconcave
distributions with compact support.

\paragraph{Log-Sobolev constant.}

We can now answer Question~\ref{ques:LSI-question} in some sense,
showing that $\clsi(\pi\gamma_{h})$ is bounded by $\O(D\lda^{1/2}\log^{2}n\log^{2}D\lda^{-1/2})$,
which also bounds $\clsi(\pi)$. 
\begin{proof}
[Proof of Corollary~\ref{cor:lsi-global-bound}] Using $\clsi(\pi\gamma_{h})\leq h$
when $h\lesssim D\lda^{1/2}\log^{2}n\log^{2}D\lda^{-1/2}$, and $\clsi(\pi\gamma_{h})\lesssim D\lda_{h}^{1/2}\log n$
(Theorem~\ref{thm:lsi-general-bound}) otherwise, along with $\lda_{h}\lesssim\lda$
(Theorem~\ref{thm:cov-Gauss}), we obtain 
\[
\clsi(\pi\gamma_{h})\lesssim\max\bbrace{D\lda^{1/2}\log^{2}n\log^{2}\frac{D}{\lda^{1/2}},D\lda^{1/2}\log n}\,.
\]
\end{proof}

\paragraph{Poincar\'e constant.}

We can ask a similar question for \eqref{eq:pi} whether $\cpi(\pi\gamma_{h})\lesssim\cpi(\pi)$.
Using $\norm{\cov\mu}\leq\cpi(\mu)\lesssim\norm{\cov\mu}\log n$ for
any logconcave probability measures $\mu$, together with Theorem~\ref{thm:cov-Gauss},
we obtain a simple corollary.
\begin{cor}
Let $\pi$ be a logconcave probability measure over $\Rn$. Then,
if $h\gtrsim R\lda^{1/2}\log^{2}n\log^{2}\frac{R}{\lda^{1/2}}$,
\[
\cpi(\pi\gamma_{h})\lesssim\cpi(\pi)\log n\,.
\]
In particular, if $\pi$ is further isotropic, then $\cpi(\pi\gamma_{h})\lesssim\cpi(\pi)\log n$
when $h\gtrsim n^{1/2}\log^{4}n$.
\end{cor}

\begin{proof}
When $h\gtrsim R\lda^{1/2}\log^{2}n\log^{2}R\lda^{-1/2}$, it holds
that $\lda_{h}\lesssim\lda$ by Theorem~\ref{thm:cov-Gauss}. Thus,
\[
\cpi(\pi\gamma_{h})\lesssim\lda_{h}\log n\lesssim\lda\log n\leq\cpi(\pi)\log n\,,
\]
which completes the proof.
\end{proof}
We remark that this is a significant improvement on the bound for
$h$, as compared to $h\asymp n^{2}$ required for applying a $\Theta(1)$-perturbation
argument.

\section{Faster warm-start generation\label{sec:faster-generation}}

Now that we have a refined understanding of functional inequalities
and warmness conditions for the $\ps$, we combine these ingredients
to design a faster algorithm for sampling without a warm start. We
recall the main question. 
\begin{question}
\label{ques:body-unif-samp} Let $\pi$ be the uniform distribution
over a convex body $\K$, presented via $\mem_{x_{0},R}(\K)$, with
covariance matrix $\Sigma$ and second moment $R^{2}:=\E_{\pi}[\norm{\cdot-x_{0}}^{2}]$.
Given $\veps\in(0,1)$, how many membership queries to $\K$ are required
to generate a sample whose law $\mu$ satisfies $\norm{\mu-\pi}_{\tv}\leq\veps$?
\end{question}

To address this question, we first bound the $\eu R_{q}$-divergence
between consecutive annealing distributions in \S\ref{subsec:Renyi-closeneess}.
Then in \S\ref{subsec:Faster-warm-start-generation}, we interweave
new ingredients together to give an improved answer to this question
as follows. 
\begin{thm}
[Restatement of Theorem~\ref{thm:intro-unif-gc}] \label{thm:body-unif-gc}
In the setting of Question~\ref{ques:body-unif-samp}, there exists
an algorithm that for any given $\eta,\veps\in(0,1)$, with probability
at least $1-\eta$ returns a sample whose law $\mu$ satisfies $\norm{\mu-\pi}_{\tv}\leq\veps$,
using
\[
\Otilde\bpar{n^{2}R^{3/2}\,\norm{\cov\pi}^{1/4}\log^{7}\frac{R^{2}}{\eta\veps\norm{\cov\pi}}}
\]
membership queries in expectation. In particular, if $\pi$ is also
near-isotropic (i.e., $\norm{\cov\pi}\lesssim1$), then $\Otilde(n^{2.75}\log^{6}\nicefrac{1}{\eta\veps})$
queries suffice.
\end{thm}

Since $\norm{\cov\pi}\leq\tr(\cov\pi)\leq R^{2}$, this complexity
improves the previous $n^{2}(n\vee R^{2})$-bound for a $\tv$-guarantee
\cite{CV18Gaussian}. 

Alongside, we improve the query complexity of sampling from a standard
Gaussian $\pi\gamma=\gamma|_{\K}$ truncated to an arbitrary convex
body, by a factor of $n^{1/2}$ \cite{CV18Gaussian}.
\begin{cor}
[Restatement of Corollary~\ref{cor:intro-gauss-gc}] \label{cor:body-gauss-gc}
In the setting of Question~\ref{ques:body-unif-samp}, there exists
an algorithm that for any given $\eta,\veps\in(0,1)$, with probability
at least $1-\eta$ returns a sample whose law $\mu$ satisfies $\norm{\mu-\pi\gamma}_{\tv}\leq\veps$,
using
\[
\Otilde\bpar{n^{2.5}\log^{7}\frac{1}{\eta\veps}}
\]
membership queries in expectation.
\end{cor}

\subsection{R\'enyi divergence of annealing distributions\label{subsec:Renyi-closeneess}}

Since we have mixing guarantees under $\eu R_{c}$-warmness for $c=\Otilde(1)$,
we can design and leverage a faster annealing scheme with provable
guarantees. In this section, we extend existing results on the closeness
of consecutive annealing distributions to the setting of $q$-R\'enyi
divergence.

\paragraph{The first type: fixed rate annealing.}

Previous work used that for logconcave $e^{-V}$ and $\alpha=n^{-1/2}$,
the distributions $\exp(-V)$ and $\exp(-(1+\alpha)V)$ are $\O(1)$-close
in $\chi^{2}$-divergence (i.e., $\eu R_{2}$) \cite{LV06simulated,KV06simulated}.
We generalize this result to $q$-R\'enyi divergence (or equivalently,
the relative $L^{q}$-norm).
\begin{lem}
[R\'enyi version of universal annealing] \label{lem:global-annealing}
Let $\D\nu\propto e^{-V}\,\D x$ and $\D\mu\propto e^{-(1+\alpha)\,V}\,\D x$
be logconcave probability measures over $\Rn$. For $q>1$ and $\delta>0$
such that $1-q\delta>0$ ,
\[
\eu R_{q}(\mu\mmid\nu)\leq\begin{cases}
\frac{qn\alpha^{2}}{2} & \text{if }\alpha\geq0\,,\\
\frac{qn\alpha^{2}}{1-q\delta} & \text{if }\alpha\in[-\delta/2,0]\,.
\end{cases}
\]
In particular, $\alpha=(qn)^{-1/2}$ yields an $\O(1)$-bound on $\eu R_{q}(\mu\mmid\nu)$,
and $\alpha=-(16q\,(q\vee n))^{-1/2}$ also yields an $\O(1)$-bound.
\end{lem}

\begin{proof}
Let us define
\[
F(s):=\int e^{-sV(x)}\,\D x\,.
\]
We recall from \cite[Lemma 3.2]{KV06simulated} that $s\mapsto s^{n}F(s)$
is logconcave in $s$.

For $q>1$, we compute the $L^{q}$-warmness of $\mu$ with respect
to $\nu$:
\begin{align*}
\Bnorm{\frac{\D\mu}{\D\nu}}_{L^{q}(\nu)}^{q} & =\int\bpar{\frac{\mu}{\nu}}^{q}\,\D\nu=\int\exp\bpar{-q\,(1+\alpha)\,V(x)+(q-1)\,V(x)}\,\D x\times\frac{F^{q-1}(1)}{F^{q}(1+\alpha)}\\
 & =\frac{F(1+q\alpha)\,F^{q-1}(1)}{F^{q}(1+\alpha)}=\Bpar{\frac{(1+\alpha)^{n}}{(1+q\alpha)^{n/q}}\,\frac{(1+q\alpha)^{n/q}F^{1/q}(1+q\alpha)\,F^{1-q^{-1}}(1)}{(1+\alpha)^{n}F(1+\alpha)}}^{q}\\
 & \leq\bpar{\frac{(1+\alpha)^{q}}{1+q\alpha}}^{n}\,,
\end{align*}
where the last inequality follows from logconcavity of $s^{n}F(s)$
in $s$.

We now claim that for $\alpha\geq0$,
\[
\frac{(1+\alpha)^{q}}{1+q\alpha}\le\exp\bpar{\frac{q\,(q-1)}{2}\,\alpha^{2}}\,.
\]
For $\alpha\geq0$, let us define
\[
g(\alpha):=\frac{q\,(q-1)}{2}\,\alpha^{2}-q\log(1+\alpha)+\log(1+q\alpha)\,.
\]
We can compute that $g(0)=0$ and 
\[
g'(\alpha)=(q-1)\,q\alpha-\frac{q}{1+\alpha}+\frac{q}{1+q\alpha}=\frac{\alpha^{2}q\,(q-1)(\alpha q+q+1)}{(1+\alpha)(\alpha q+1)}\geq0\,,
\]
from which the claim follows.

We also show that for $\alpha\in[-\delta/2,0]$ with $1-q\delta>0$,
\[
\frac{(1+\alpha)^{q}}{1+q\alpha}\le\exp\bpar{\frac{q\,(q-1)}{2\,(1-q\delta)}\,\alpha^{2}}\,.
\]
In a similarly way, we define and compute that
\begin{align*}
g(\alpha) & :=\frac{q\,(q-1)}{2\,(1-q\delta)}\,\alpha^{2}-q\log(1+\alpha)+\log(1+q\alpha)\,,\\
g'(\alpha) & =\frac{q-1}{1-q\delta}\,q\alpha-\frac{q}{1+\alpha}+\frac{q}{1+q\alpha}=\alpha q\,(q-1)\,\bpar{\frac{1}{1-q\delta}-\frac{1}{(1+\alpha)\,(1+q\alpha)}}\leq0\,,
\end{align*}
from which the claim follows.

Using each bound for $\alpha\geq0$ and $\alpha\in[-\delta/2,0]$,
\[
\eu R_{q}(\mu\mmid\nu)=\frac{1}{q-1}\,\log\,\Bnorm{\frac{\D\mu}{\D\nu}}_{L^{q}(\nu)}^{q}\leq\begin{cases}
\frac{qn\alpha^{2}}{2} & \text{if }\alpha\geq0\,,\\
\frac{qn\alpha^{2}}{1-q\delta} & \text{if }\alpha\in[-\delta/2,0]\,.
\end{cases}
\]
This completes the proof.
\end{proof}

\paragraph{The second type: accelerated annealing.}

In this annealing, the change in the annealing parameter depends on
the variance of the current distribution. A bound of this type was
previously established for $\chi^{2}$-divergence in \cite[Lemma 7.8]{CV18Gaussian}. 

We show that this bound also extends to $\eu R_{q}$-divergence, enabling
control over the closeness of annealing steps under R\'enyi-divergence.
In particular, this allows us to multiply the variance by a factor
of $1+\sigma/R$ in each step. 
\begin{lem}
[R\'enyi version of accelerated annealing] \label{lem:variance-annealing}
Let $\mu$ be a logconcave probability density in $\Rn$ with support
of diameter $R>0$. Then, for $q>1$ and $\sigma^{2}>0$,
\[
\eu R_{q}(\mu\gamma_{\sigma^{2}}\mmid\mu\gamma_{\sigma^{2}(1+\alpha)})\leq\frac{qR^{2}\alpha^{2}}{\sigma^{2}}\,.
\]
In particular, $\alpha\lesssim\sigma/(q^{1/2}R)$ yields an $\O(1)$-bound
on the $\eu R_{q}$-divergence.
\end{lem}

\begin{proof}
We define
\[
\mu_{\sigma^{2}}(x):=\mu\gamma_{\sigma^{2}}\,,\qquad F(s):=\int\mu(x)\exp\bpar{-\frac{1+\alpha+s}{2\sigma^{2}\,(1+\alpha)}\,\norm x^{2}}\,,\qquad G(s):=\log F(s)\,.
\]
Then,
\[
\Bnorm{\frac{\D\mu_{\sigma^{2}}}{\D\mu_{\sigma^{2}(1+\alpha)}}}_{L^{q}(\mu_{\sigma^{2}(1+\alpha)})}^{q}=\frac{F^{q-1}(-\alpha)\,F\bpar{(q-1)\,\alpha}}{F^{q}(0)}\,,
\]
and
\begin{align*}
\log\,\Bnorm{\frac{\D\mu_{\sigma^{2}}}{\D\mu_{\sigma^{2}(1+\alpha)}}}_{L^{q}(\mu_{\sigma^{2}(1+\alpha)})}^{q} & =(q-1)\,G(-\alpha)+G\bpar{(q-1)\,\alpha}-q\,G(0)\\
 & =(q-1)\,\bpar{G(-\alpha)-G(0)}+G\bpar{(q-1)\,\alpha}-G(0)\\
 & =-(q-1)\int_{0}^{\alpha}G'(-t)\,\D t+(q-1)\int_{0}^{\alpha}G'\bpar{(q-1)\,t}\,\D t\\
 & =(q-1)\int_{0}^{\alpha}\Bpar{G'\bpar{(q-1)\,t}-G'(-t)}\,\D t\\
 & =(q-1)\int_{0}^{\alpha}\int_{-t}^{(q-1)\,t}G''(u)\,\D u\D t\,.
\end{align*}
For the strongly logconcave distribution $\nu_{u}(x)\propto\mu(x)\exp(-\frac{1+\alpha+u}{2\sigma^{2}(1+\alpha)}\,\norm x^{2})$,
we have
\begin{align*}
G''(u) & =-\frac{1}{2\sigma^{2}\,(1+\alpha)}\,\frac{\D}{\D u}\Bpar{\frac{\int\norm x^{2}\mu(x)\exp(-\frac{1+\alpha+u}{2\sigma^{2}(1+\alpha)}\,\norm x^{2})}{\int\mu(x)\exp(-\frac{1+\alpha+u}{2\sigma^{2}(1+\alpha)}\,\norm x^{2})}}\\
 & =\frac{1}{4\sigma^{4}\,(1+\alpha)^{2}}\,\Bbrack{\frac{\int\norm x^{4}\mu(x)\exp(-\frac{1+\alpha+u}{2\sigma^{2}(1+\alpha)}\,\norm x^{2})}{\int\mu(x)\exp(-\frac{1+\alpha+u}{2\sigma^{2}(1+\alpha)}\,\norm x^{2})}-\Bpar{\frac{\int\norm x^{2}\mu(x)\exp(-\frac{1+\alpha+u}{2\sigma^{2}(1+\alpha)}\,\norm x^{2})}{\int\mu(x)\exp(-\frac{1+\alpha+u}{2\sigma^{2}(1+\alpha)}\,\norm x^{2})}}^{2}}\\
 & =\frac{1}{4\sigma^{4}\,(1+\alpha)^{2}}\var_{\nu_{u}}(\norm X^{2})\underset{\eqref{eq:pi}}{\leq}\frac{1}{4\sigma^{4}\,(1+\alpha)^{2}}\,\cpi(\nu_{u})\,\E_{\nu_{u}}[\norm X^{2}]\\
 & \leq\frac{R^{2}}{\sigma^{2}\,(1+\alpha)(1+\alpha+u)}\,,
\end{align*}
where the last line follows from $\sup_{\nu_{u}}\norm X\le R$ and
$\cpi(\nu_{u})\leq\frac{\sigma^{2}(1+\alpha)}{1+\alpha+u}$. Hence,
\begin{align*}
\eu R_{q}(\mu\gamma_{\sigma^{2}}\mmid\mu\gamma_{\sigma^{2}(1+\alpha)})=\frac{1}{q-1}\,\log\,\norm{\cdot}_{L^{q}}^{q} & \leq\frac{R^{2}}{\sigma^{2}\,(1+\alpha)}\int_{0}^{\alpha}\int_{-t}^{(q-1)\,t}\frac{1}{1+\alpha+u}\,\D u\D t\\
 & =\frac{R^{2}}{\sigma^{2}\,(1+\alpha)}\int_{0}^{\alpha}\log\frac{1+\alpha+(q-1)\,t}{1+\alpha-t}\,\D t\\
 & =\frac{R^{2}}{\sigma^{2}\,(1+\alpha)}\int_{0}^{\alpha}\log\bpar{1+\frac{qt}{1+\alpha-t}}\,\D t\\
 & \leq\frac{qR^{2}}{\sigma^{2}}\int_{0}^{\alpha}\frac{t}{1+\alpha-t}\,\D t\\
 & \leq\frac{qR^{2}\alpha^{2}}{\sigma^{2}}\,,
\end{align*}
which completes the proof.
\end{proof}

\subsection{Faster warm-start  sampling\label{subsec:Faster-warm-start-generation}}

Similar to previous work, we follow an annealing approach based on
$\gc$.

\subsubsection{Algorithm \label{subsec:GC-algorithm}}

We can truncate $\K$ to a ball of radius $D\asymp R$ so that for
$\bar{\K}:=\K\cap B_{D}(0)$, the truncated distribution $\bar{\pi}:=\pi|_{\bar{\K}}\propto\pi\cdot\ind_{\bar{\K}}$
is $2$-warm with respect to $\pi$.
\begin{prop}
[{\cite[Proposition 30]{KZ25Renyi}}] \label{prop:region-truncation}
Given $\veps>0$, there exists a constant $L=\log\frac{e}{\veps}$
such that the volume of $\K\cap B_{LR}(0)$ is at least $(1-\veps)\vol\K$.
\end{prop}

In the following description of our algorithm, we proceed as if the
output distribution from each iteration \emph{matches the intended
target} distribution. This was justified in \S\ref{sec:techniques}
by the argument based on the triangle inequality for $\tv$-distance.

Let $\mu_{i}$ denote a logconcave probability measure with density
\[
\D\mu_{i}\propto\bar{\pi}\gamma_{\sigma_{i}^{2}}\,\D x\,,
\]
and let $m$ denote the number of inner phases in the algorithm.

Below, we set failure probability and target accuracy to $\eta/m$
and $\veps/m$, respectively.
\begin{itemize}
\item Initialization ($\sigma^{2}=n^{-1}$)
\begin{itemize}
\item Rejection sampling to sample from $\mu_{1}\propto\bar{\pi}\gamma_{n^{-1}}$
with proposal $\gamma_{n^{-1}}$.
\end{itemize}
\item Annealing ($n^{-1}\leq\sigma^{2}\leq D^{2}$)
\begin{itemize}
\item Run $\psgauss$ with initial $\mu_{i}$ and target $\mu_{i+1}$, where
\[
\sigma_{i+1}^{2}=\sigma_{i}^{2}\bpar{1+\frac{\sigma}{q^{1/2}D}}\,.
\]
\item Depending on whether $\sigma^{2}$ is below or above $\Theta(D\lda^{1/2}\log^{2}n\log^{2}\nicefrac{D^{2}}{\lda})$,
use suitable parameters ($h$ and $N$) of $\psgauss$ according to
Theorem~\ref{thm:body-gauss-samp}.
\end{itemize}
\item Termination $(\sigma^{2}=D^{2})$
\begin{itemize}
\item Run $\psunif$ with initial $\mu_{\text{last}}=\bar{\pi}\gamma_{D^{2}}$
and target $\pi$.
\end{itemize}
\end{itemize}

\subsubsection{Analysis}

To obtain a provable guarantee on the final query complexity, we should
specify the R\'enyi parameter $q$. While this is somewhat subtle
since parameters are dependent to each other, we walk through how
to set all necessary parameters without incurring any logical gaps.

\paragraph{Choice of parameters.}

For $q\geq2$, our algorithm guarantees $\eu R_{2}$-warmness between
any pair of consecutive distributions (i.e., $M_{2}\lesssim1$). The
number $m$ of inner phases is bounded as
\[
m\leq\frac{q^{1/2}D}{\sigma_{0}}\log nD\leq q^{1/2}n^{1/2}D\log nD=:m_{\max}(q)\,.
\]
Also, by Theorem~\ref{thm:body-gauss-samp}, the number of iterations
required for each inner phase is at most 
\[
k\leq\Otilde\bpar{n^{2}D^{2}\log^{3}\frac{m_{\max}^{2}(q)}{\eta\veps}}\leq\Otilde(n^{2}D^{2}\log^{3}\frac{q}{\eta\veps})=:k_{\max}(q)\,.
\]
To ensure the query complexity bound of $\psgauss$, it suffices to
choose $q$ such that 
\[
q\gtrsim\log\frac{qnD}{\eta\veps}\Bpar{\gtrsim\log\frac{k_{\max}(q)\,m_{\max}(q)}{\eta}\gtrsim6\log\frac{16kmM_{2}}{\eta}}\,,
\]
to obtain a query complexity bound for $\psgauss$ via the monotonicity
of $M_{q}$ in $q$. Since the first inequality is satisfied by choosing
\[
q\gtrsim\log\frac{nD}{\eta\veps}=\Otilde(1)\,,
\]
we can set $q$ to the RHS, which in turn determines the values of
$m_{\max}(q)$ and $k_{\max}(q)$ accordingly. 

\paragraph{Complexity bound.}

We now analyze the expected number of membership queries used in each
phase of the algorithm.
\begin{lem}
[Initialization] The initialization can sample from $\mu_{1}=\bar{\pi}\gamma_{n^{-1}}$,
using $\O(1)$ queries in expectation.
\end{lem}

\begin{proof}
Rejection sampling with proposal $\gamma_{n^{-1}}$ succeeds in the
following number of trials in expectation:
\[
\frac{\int_{\Rn}\exp(-\frac{n}{2}\,\norm x^{2})\,\D x}{\int_{\bar{\K}}\exp(-\frac{n}{2}\,\norm x^{2})\,\D x}=\bpar{\gamma_{n^{-1}}(\bar{\K})}^{-1}\leq\bbrace{\gamma_{n^{-1}}\bpar{B_{1}(0)}}^{-1}\lesssim1\,,
\]
which completes the proof.
\end{proof}
Next, we bound the query complexity of the main annealing phase. Below,
we denote $\lda:=\norm{\Sigma}$.
\begin{lem}
[Annealing] \label{lem:unif-annealing} With probability at least
$1-\eta$, the Gaussian annealing outputs a sample whose law $\nu$
satisfies $\norm{\nu-\bar{\pi}\gamma_{D^{2}}}_{\tv}\leq\veps$, using
\[
\Otilde\bpar{n^{2}D^{3/2}\lda^{1/4}\log^{6}\frac{1}{\eta\veps}\log\frac{D^{2}}{\lda}}
\]
membership queries in expectation.
\end{lem}

\begin{proof}
For any given $\sigma^{2}\in[n^{-1},D^{2}]$, we need at most $q^{1/2}D/\sigma$
phases to double $\sigma^{2}$. While doubling the initial $\sigma^{2}$,
any consecutive distributions are $\O(1)$-close in $\eu R_{q}$ (i.e.,
$M_{q}\lesssim1$) by Lemma~\ref{lem:variance-annealing}.

By Theorem~\ref{thm:body-gauss-samp}, when $\sigma^{2}\lesssim D\lda^{1/2}\log^{2}n\log^{2}\nicefrac{D^{2}}{\lda}$,
with probability at least $1-\eta/m$, $\psgauss$ can sample a next
annealing distribution from a current one with a $\eu R_{2}$-guarantee,
using
\[
\Otilde\bpar{n^{2}\sigma^{2}\log^{6}\frac{m_{\max}^{2}}{\eta\veps}}
\]
queries in expectation. Hence, throughout this doubling, the total
query complexity is 
\[
\Otilde\bpar{q^{1/2}Dn^{2}\sigma\log^{6}\frac{m_{\max}^{2}}{\eta\veps}}=\Otilde\bpar{n^{2}D^{3/2}\lda^{1/4}\log^{6}\frac{1}{\eta\veps}\log\frac{D^{2}}{\lda}}\,.
\]

When $\sigma^{2}$ exceeds this threshold, $\psgauss$ uses $\Otilde(n^{2}D\lda^{1/2}\log^{6}\frac{m_{\max}^{2}}{\eta\veps})$
queries in expectation. Thus, the total query complexity of one doubling
is 
\[
\Otilde\bpar{\frac{q^{1/2}n^{2}D^{2}\lda^{1/2}}{\sigma}\log^{6}\frac{m_{\max}^{2}}{\eta\veps}}=\Otilde\bpar{n^{2}D^{3/2}\lda^{1/4}\log^{6}\frac{1}{\eta\veps}}\,.
\]
Adding up these two bounds, we can bound the total complexity during
annealing as claimed. Using an argument based on the triangle inequality
(explained in \S\ref{sec:techniques}), we can also obtain the final
$\tv$-guarantee. Lastly, the failure probability also follows from
the union bound.
\end{proof}
The termination simply runs $\psunif$.
\begin{lem}
[Termination] With probability at least $1-\eta$, the termination
phase with initial $\bar{\pi}\gamma_{D^{2}}$ and target $\pi$ outputs
a sample whose law $\nu$ satisfies $\eu R_{2}(\nu\mmid\pi)\leq\veps$,
using
\[
\Otilde\bpar{n^{2}\lda\log^{6}\frac{1}{\eta\veps}}
\]
queries in expectation.
\end{lem}

\begin{proof}
We first show that $\bar{\pi}\gamma_{D^{2}}$ is $\O(1)$-close to
$\pi$ in $\eu R_{\infty}$. This is immediate from the following
computation:
\[
\frac{\bar{\pi}\gamma_{D^{2}}}{\pi}=\frac{\bar{\pi}\gamma_{D^{2}}}{\bar{\pi}}\times\frac{\bar{\pi}}{\pi}=\frac{\vol\bar{\K}\,\exp(-\frac{1}{2D^{2}}\,\norm x^{2})}{\int_{\bar{\K}}\exp(-\frac{1}{2D^{2}}\,\norm x^{2})\,\D x}\times\frac{\vol\K}{\vol\bar{\K}}\lesssim1\,,
\]
where the last line follows from $\sup_{x\in\bar{\K}}\norm x\leq D$
and Proposition~\ref{prop:region-truncation}. Using $\psunif$ \cite{KVZ24INO},
with probability at least $1-\eta$, we can sample from $\pi$ with
a $\eu R_{2}$-guarantee, using $\Otilde(n^{2}\lda\log^{6}\nicefrac{1}{\eta\veps})$
queries in expectation.
\end{proof}
Combining the previous three lemmas, we can prove Theorem~\ref{thm:body-unif-gc}.
\begin{proof}
[Proof of Theorem~\ref{thm:body-unif-gc}] Replacing $\veps\gets\veps/2$
and $\eta\gets\eta/2$, and combining the three lemmas together, we
conclude that the total failure probability is at most $\eta$. Also,
the total complexity is dominated by annealing near $\sigma^{2}\approx D\lda^{1/2}$,
leading to the bound of 
\[
\Otilde\bpar{n^{2}D^{3/2}\lda^{1/4}\log^{6}\frac{1}{\eta\veps}\log\frac{D^{2}}{\lda}}=\Otilde\bpar{n^{2}R^{3/2}\lda^{1/4}\log^{6}\frac{1}{\eta\veps}\log\frac{R^{2}}{\lda}}\,
\]
This completes the proof.
\end{proof}

\subsubsection{Sampling from a truncated Gaussian}

We can sample from a standard Gaussian truncated to a convex body
$\K$ in a similar way. Only difference is that we use the global
annealing (Lemma~\ref{lem:global-annealing}), updating $\sigma^{2}$
by $\sigma^{2}(1+(qn)^{-1/2})$.

In this case, for $q\geq2$, we can ensure that $M_{2}\lesssim1$.
Also,
\begin{align*}
m & \leq q^{1/2}n^{1/2}\log n=:m_{\max}(q)\,,\\
k & \leq\Otilde\bpar{n^{2}\log^{3}\frac{m_{\max}^{2}(q)}{\eta\veps}}\leq\Otilde\bpar{n^{2}\log^{3}\frac{q}{\eta\veps}}=:k_{\max}(q)\,.
\end{align*}
By Theorem~\ref{thm:body-gauss-samp}, it suffices to have 
\[
q\gtrsim\log\frac{qn}{\eta\veps}\bpar{\gtrsim6\log\frac{16kmM_{2}}{\eta}}\,,
\]
which is fulfilled if $q\gtrsim\log\frac{n}{\eta\veps}=\Otilde(1)$.
Hence, we set $q$ to the RHS, and then pick $m_{\max}$ and $k_{\max}$
accordingly.

An annealing algorithm for $\pi\gamma$ is almost the same with one
for a uniform distribution.
\begin{itemize}
\item Initialization ($\sigma^{2}=n^{-1}$)
\begin{itemize}
\item Rejection sampling to sample from $\mu_{1}\propto\pi\gamma_{n^{-1}}$
with proposal $\gamma_{n^{-1}}$.
\end{itemize}
\item Annealing ($n^{-1}\leq\sigma^{2}\leq1$)
\begin{itemize}
\item Run $\psgauss$ with initial $\mu_{i}$ and target $\mu_{i+1}$, where
\[
\sigma_{i+1}^{2}=\sigma_{i}^{2}\bpar{1+(qn)^{-1/2}}\,.
\]
\end{itemize}
\end{itemize}
We can then bound query complexity of the annealing phase.
\begin{lem}
[Annealing] With probability at least $1-\eta$, Gaussian annealing
outputs a sample whose law $\nu$ satisfies $\norm{\nu-\pi\gamma}_{\tv}\leq\veps$,
using
\[
\Otilde\bpar{n^{5/2}\log^{7}\frac{1}{\eta\veps}}
\]
queries in expectation.
\end{lem}

\begin{proof}
For any given $\sigma^{2}\in[n^{-1},1]$, we need at most $q^{1/2}n^{1/2}$
inner phases to double $\sigma^{2}$. While doubling the initial $\sigma^{2}$,
any consecutive distributions are $\O(1)$-close in $\eu R_{q}$ (i.e.,
$M_{q}\lesssim1$) by Lemma~\ref{lem:global-annealing}.

By Theorem~\ref{thm:body-gauss-samp}, with probability at least
$1-\eta/m$, $\psgauss$ can sample a next annealing distribution
from a current one with a $\eu R_{2}$-guarantee, using
\[
\Otilde\bpar{n^{2}\sigma^{2}\log^{6}\frac{m_{\max}^{2}}{\eta\veps}}
\]
queries in expectation. Hence, throughout this doubling, the total
query complexity is 
\[
\Otilde\bpar{q^{1/2}n^{5/2}\sigma^{2}\log^{6}\frac{m_{\max}^{2}}{\eta\veps}}=\Otilde\bpar{n^{5/2}\log^{7}\frac{1}{\eta\veps}}\,.
\]
The remaining part can be done similarly as in Lemma~\ref{lem:unif-annealing}.
\end{proof}
Combining this lemma and one on initialization from the previous section
yields the proof of Corollary~\ref{cor:body-gauss-gc}.
\begin{proof}
[Proof of Corollary~\ref{cor:body-gauss-gc}] Since the query complexity
of annealing dominates that for initialization, the claim immediately
follows.
\end{proof}

\section{Extension to general logconcave distributions\label{sec:extension}}

We now extend our previous results to general logconcave distributions
given access to a well-defined function oracle. In particular, we
establish improved query complexity bounds for logconcave sampling,
similar to our results for uniform sampling.
\begin{thm}
[Restatement of Theorem~\ref{thm:intro-LC-sampling}] \label{thm:body-LC-sampling}
For a logconcave distribution $\pi$ over $\Rn$ presented by $\eval_{x_{0},R}(V)$,
there exists an algorithm that for any given $\eta,\veps\in(0,1)$,
with probability at least $1-\eta$ returns a sample whose law $\mu$
satisfies $\norm{\mu-\pi}_{\tv}\leq\veps$, using $\Otilde(n^{2}\max\{n^{1/2},R^{3/2}(\lda^{1/4}\vee1)\}\log^{5}\nicefrac{R\lda^{-1/2}}{\eta\veps})$
evaluation queries in expectation. If $\pi$ is also near-isotropic
, then $\Otilde(n^{2.75}\log^{4}\nicefrac{1}{\eta\veps})$ queries
suffice.
\end{thm}

\subsection{Logconcave sampling under relaxed warmness\label{subsec:lc-sampling-warm}}

\cite{KV25sampling} showed that general logconcave sampling can be
reduced to more structured exponential sampling: for logconcave $\D\pi^{X}\propto\exp(-V)\,\D x$,
consider
\begin{equation}
\D\pi(x,t)\propto\exp(-nt)\,\ind_{\K}(x,t)\,\D x\D t\quad\text{subject to }\K:=\{(x,t)\in\Rn\times\R:V(x)\leq nt\}\,,\tag{\ensuremath{\msf{exp}\text{-}\msf{red}}}\label{eq:exp-red}
\end{equation}
and then focus on sampling from this $\pi$, as its $X$-marginal
is exactly $\pi^{X}$ the target distribution \cite[Proposition 2.3]{KV25sampling}.

In this section, we bound the query complexities of sampling from
this distribution and related annealing distributions under relaxed
warmness conditions on the initial distribution instead of $\eu R_{\infty}$-warmness.
\begin{thm}
\label{thm:body-exp-samp} Given a logconcave density $\pi$ specified
by an evaluation oracle $\eval_{R}(V)$, consider $\pi(x,t)\propto\exp(-nt)|_{\K}$
given in \eqref{eq:exp-red} and initial distribution $\mu$ with
the relative $L^{q}$-norm denoted as $M_{q}=\norm{\D\mu/\D\pi}_{L^{q}(\pi)}$
for $q\geq2$. For any $\eta,\veps\in(0,1)$, $k\in\mathbb{N}$ defined
below, $\psexp$ with $h=(13^{4}n^{2}\log\frac{16kM_{2}}{\eta})^{-1}$
and $N=(\frac{16kM_{2}}{\eta})^{2}\log^{2}\frac{16kM_{2}}{\eta}$,
with probability at least $1-\eta$, achieves $\eu R_{2}(\mu_{k}\mmid\pi)\leq\veps$
after $k=\Otilde(n^{2}(\norm{\cov\pi^{X}}\vee1)\log^{2}\frac{M_{2}}{\eta\veps})$
iterations, where $\mu_{k}$ is the law of the $k$-th iterate, using
\[
\Otilde\bpar{M_{c}n^{2}(\norm{\cov\pi^{X}}\vee1)\log^{3}\frac{1}{\eta\veps}}
\]
evaluation queries in expectation for any $c\geq6\log\frac{16kM_{2}}{\eta}$.
Moreover, an $M_{q}$-warm start for $\pi^{X}$ can be used to generate
an $M_{q}$-warm start for $\pi$.
\end{thm}

\begin{thm}
\label{thm:body-exp-tilt-Gauss} Under $\eval_{x_{0}=0,R}(V)$, consider
$\mu:=\mu_{\sigma^{2},\rho}(x,t)\propto\exp(-\rho t-\frac{1}{2\sigma^{2}}\,\norm x^{2})|_{\bar{\K}}$,
where the diameters of $\bar{\K}$ (obtained by truncation; see \eqref{eq:convex-truncation})
in the $x$ and $t$-direction are $\O(R)$ and $\O(1)$, and initial
distribution $\nu$ with the relative $L^{q}$-norm denoted as $M_{q}=\norm{\D\nu/\D\mu}_{L^{q}(\mu)}$
for $q\geq2$. For any $\eta,\veps\in(0,1)$, $k\in\mathbb{N}$ defined
below, $\psann$ with $h=(1200^{2}n^{2}\log\frac{16kM_{2}}{\eta})^{-1}$
and $N=2(\frac{16kM_{2}}{\eta})^{2}\log^{2}\frac{16kM_{2}}{\eta}$
achieves $\eu R_{2}(\nu_{k}\mmid\mu)\leq\veps$ after $k=\Otilde(n^{2}(\sigma^{2}\vee1)\log^{2}\frac{M_{2}}{\eta\veps})$
iterations, where $\nu_{k}$ is the law of the $k$-th iterate. With
probability at least $1-\eta$, $\psann$ iterates $k$ times without
failure, using 
\[
\Otilde\bpar{M_{c}n^{2}(\sigma^{2}\vee1)\log^{3}\frac{1}{\eta\veps}}
\]
evaluation queries in expectation for any $c\geq6\log\frac{16kM_{2}}{\eta}$.
When $\rho=n$ and $\sigma^{2}\gtrsim R\lda^{1/2}\log^{2}n\log^{2}\frac{D^{2}}{\lda}$
for $\lda=\norm{\cov\pi^{X}}$, it suffices to run $k=\Otilde(n^{2}(R\vee1)(\lda^{1/2}\vee1)\log^{2}\frac{M_{2}}{\eta\veps})$
times with query complexity
\[
\Otilde\bpar{M_{c}n^{2}(R\vee1)(\lda^{1/2}\vee1)\log^{3}\frac{1}{\eta\veps}}\,,
\]
\end{thm}

\subsubsection{Sampling from the reduced exponential distribution}

\cite{KV25sampling} proposed and analyzed the $\ps$ for this reduced
distribution (called $\psexp$ there). For $z=(x,t)\in\Rn\times\R$
and parameter $h>0$, consider the augmented distribution 
\[
\pi^{Z,Y}(z,y)\propto\exp\bpar{-\alpha^{\T}z-\frac{1}{2h}\,\norm{z-y}^{2}}\,\ind_{\K}(z)\qquad\text{for }\alpha:=ne_{n+1}\,,
\]
where the $Z$-marginal $\pi^{Z}$ is the reduced distribution $\pi$
given in \eqref{eq:exp-red}. Then, $\psexp$ alternates the following
two steps:
\begin{itemize}
\item {[}Forward{]} $y\sim\pi^{Y|Z=z}=\mc N(z,hI_{n+1})$.
\item {[}Backward{]} $z\sim\pi^{Z|Y=y}=\mc N(y-h\alpha,hI_{n+1})|_{\K}$
for $\alpha=ne_{n+1}$.
\begin{itemize}
\item It is implemented by rejection sampling with proposal $\mc N(y-h\alpha,hI_{n+1})$.
\end{itemize}
\end{itemize}

\paragraph{Mixing analysis.}

It follows from \cite[Lemma 29]{KV25sampling} or \eqref{eq:PI-contraction}
that $\psexp$ with initial distribution $\pi_{0}$ can achieve $\veps$-distance
in $\eu R_{2}$ after 
\[
k\gtrsim h^{-1}\cpi(\pi^{Z})\log\frac{\chi^{2}(\pi_{0}\mmid\pi^{Z})}{\veps}
\]
iterations. To bound $\cpi(\pi^{Z})$, we recall the following result.
\begin{lem}
[{\cite[Lemma 2.5]{KV25sampling}}]\label{lem:exp-cov} The variance
of $\pi^{Z}$ in the $t$-direction is at most $160$. Moreover, $\tr\cov\pi^{Z}\leq\tr\cov\pi^{X}+160$,
and $\norm{\cov\pi^{Z}}\leq2\,(\norm{\cov\pi^{X}}+160)$.
\end{lem}

Combined with $\cpi(\pi^{Z})\lesssim\norm{\cov\pi^{Z}}\log n$, the
required number of iterations is 
\begin{equation}
k\gtrsim h^{-1}(\norm{\cov\pi^{X}}\vee1)\log\frac{\chi^{2}(\pi_{0}\mmid\pi)}{\veps}\,.\label{eq:psexp-mixing}
\end{equation}

\paragraph{Preliminaries.}

We note that rejection sampling needs 
\[
\frac{1}{\ell(y)}:=\frac{(2\pi h)^{(n+1)/2}}{\int_{\mc K}\exp(-\frac{1}{2h}\,\norm{z-(y-h\alpha)}^{2})\,\D z}
\]
expected trials until the first acceptance.

Before proceeding to the complexity of the backward step, we recall
the density of $\pi^{Y}=\pi^{Z}*\gamma_{h}$ and helper lemmas.
\begin{lem}
[{\cite[Lemma 2.11]{KV25sampling}}] For $\pi^{Z}\propto\exp(-\alpha^{\T}z)|_{\K}$,
the density of $\pi^{Y}=\pi^{Z}*\gamma_{h}$ is
\[
\pi^{Y}(y)=\frac{\ell(y)\exp(-\alpha^{\T}y+\half\,h\norm{\alpha}^{2})}{\int_{\mc K}\exp(-\alpha^{\T}z)\,\D z}\,.
\]
\end{lem}

Its essential domain is given as follows.
\begin{lem}
[{\cite[Lemma 2.12]{KV25sampling}}] \label{lem:exp-eff-domain}
Under $\eval(V)$, for $\widetilde{\K}:=\K_{\delta}+h\alpha$, if
$\delta\geq13hn$, then
\[
\pi^{Y}(\widetilde{\K}^{c})\leq\exp\bpar{-\frac{\delta^{2}}{2h}+13\delta n}\,.
\]
\end{lem}

We will set $h=\tfrac{c}{13^{2}n^{2}}$ and $\delta=\tfrac{t}{13n}$
for some $c,t>0$, under which $\delta\geq13hn$ is equivalent to
$t\geq c$.
\begin{lem}
[{\cite[Lemma 2.13]{KV25sampling}}] \label{lem:exp-expansion}
Under $\eval(V)$,
\[
\frac{\int_{\mc K_{s}}\exp(-\alpha^{\T}v)\,\D v}{\int_{\mc K}\exp(-\alpha^{\T}v)\,\D v}\leq\exp(13sn)\,.
\]
\end{lem}

We will use the following bound:
\begin{align}
\frac{\int_{\widetilde{\K}}\exp(-\alpha^{\T}y+\half\,h\norm{\alpha}^{2})\,\D y}{\int_{\mc K}\exp(-\alpha^{\T}v)\,\D v} & \underset{(i)}{=}\frac{\int_{\K_{\delta}}\exp(-\alpha^{\T}v-\half\,h\norm{\alpha}^{2})\,\D v}{\int_{\mc K}\exp(-\alpha^{\T}v)\,\D v}\leq\frac{\int_{\K_{\delta}}\exp(-\alpha^{\T}v)\,\D v}{\int_{\mc K}\exp(-\alpha^{\T}v)\,\D v}\nonumber \\
 & \underset{(ii)}{\leq}e^{13\delta n}=e^{t}\,,\label{eq:handy-bound}
\end{align}
where $(i)$ follows from change of variables via $y=v+h\alpha$,
and $(ii)$ follows from Lemma~\ref{lem:exp-expansion}.

\paragraph{(1) Failure probability.}

For a distribution $\nu\ll\pi^{Z}$, the failure probability is bounded
as 
\[
\E_{\nu_{h}}[(1-\ell)^{N}]\leq M_{2}\sqrt{\E_{\pi_{h}}[(1-\ell)^{2N}]}\,,
\]
where $\nu_{h}=\nu*\gamma_{h}$, $\pi_{h}=\pi^{Z}*\gamma_{h}=\pi^{Y}$,
and $M_{2}:=\norm{\D\nu/\D\pi^{Z}}_{L^{2}(\pi^{Z})}$. Following the
proof of \cite[Lemma 2.14]{KV25sampling} with $N$ there replaced
by $2N$ and $M$ there replaced by $M_{2}$, we decompose
\[
\E_{\pi_{h}}[(1-\ell)^{2N}]=\int_{\widetilde{\K}^{c}}\cdot+\int_{\widetilde{\K}\cap[\ell\geq N^{-1}\log(3kM_{2}/\eta)]}\cdot+\int_{\widetilde{\K}\cap[\ell\leq N^{-1}\log(3kM_{2}/\eta)]}\cdot=:\msf A+\msf B+\msf C\,,
\]
where we have,
\begin{align*}
\msf A & \leq\pi^{Y}(\widetilde{\K}^{c})\leq\exp\bpar{-\frac{t^{2}}{2c}+t}\,,\\
\msf B & \leq\int_{\widetilde{\K}\cap[\ell\geq N^{-1}\log(3kM_{2}/\eta)]}\exp(-2\ell N)\,\D\pi^{Y}\leq\bpar{\frac{\eta}{3kM_{2}}}^{2}\,,\\
\msf C & \leq\int_{\widetilde{\K}\cap[\ell\leq N^{-1}\log(3kM_{2}/\eta)]}\frac{\ell(y)\exp(-\alpha^{\T}y+\half\,h\norm{\alpha}^{2})}{\int_{\mc K}\exp(-\alpha^{\T}v)\,\D v}\,\D y\\
 & \underset{\text{Lemma }\ref{lem:exp-eff-domain}}{\leq}\frac{\log(3kM_{2}/\eta)}{N}\,\frac{\int_{\widetilde{\K}}\exp(-\alpha^{\T}y+\half\,h\norm{\alpha}^{2})\,\D y}{\int_{\mc K}\exp(-\alpha^{\T}v)\,\D v}\underset{\eqref{eq:handy-bound}}{\leq}\frac{\log(3kM_{2}/\eta)}{N}\,e^{t}\,.
\end{align*}
For $S=\frac{16kM_{2}}{\eta}(\geq16)$, by choosing $c=\frac{(\log\log S)^{2}}{13^{2}\log S}$,
$t=\log\log S$, and $N=S^{2}\log^{2}S$, we can bound each term by
$(\frac{\eta}{3kM_{2}})^{2}$. Therefore, the total failure probability
over $k$ iterations is at most $\eta$ by a union bound.

\paragraph{(2) Complexity of the backward step.}

Let $p=1+\frac{1}{\alpha}$ and $q=1+\alpha$ with $\alpha=\log N\geq1$.
Then,
\[
\E_{\nu_{h}}\bbrack{\frac{1}{\ell}\wedge N}=\int_{\widetilde{\K}\cap[\ell\geq N^{-p}]}\cdot+\int_{\widetilde{\K}\cap[\ell<N^{-p}]}\cdot+\int_{\widetilde{\K}^{c}}\cdot=:\msf A+\msf B+\msf C\,,
\]
where 
\begin{align*}
\msf A & \leq M_{q}\,\Bpar{\int_{\widetilde{\K}\cap[\ell\geq N^{-p}]}\frac{1}{\ell^{p}}\wedge N^{p}\,\D\pi_{h}}^{1/p}\leq M_{q}\,\Bpar{\int_{\widetilde{\K}\cap[\ell\geq N^{-p}]}\frac{1}{\ell^{p}}\,\frac{\ell(y)\exp(-\alpha^{\T}y+\half\,h\norm{\alpha}^{2})}{\int_{\mc K}\exp(-\alpha^{\T}z)\,\D z}\,\D y}^{1/p}\\
 & \leq M_{q}N^{1/\alpha}\Bpar{\frac{\int_{\widetilde{\K}}\exp(-\alpha^{\T}y+\frac{1}{2}\,h\norm{\alpha}^{2})\,\D y}{\int_{\K}\exp(-\alpha^{\T}z)\,\D z}}^{1/p}\underset{\eqref{eq:handy-bound}}{\leq}M_{q}e^{t+1}\,,\\
\msf B & \leq N\int_{\widetilde{\K}\cap[\ell<N^{-p}]}\frac{\D\nu_{h}}{\D\pi_{h}}\,\D\pi_{h}\leq NM_{q}\,\Bpar{\int_{\widetilde{\K}\cap[\ell<N^{-p}]}\frac{\ell(y)\exp(-\alpha^{\T}y+\half\,h\norm{\alpha}^{2})}{\int_{\mc K}\exp(-\alpha^{\T}z)\,\D z}\,\D y}^{1/p}\underset{\eqref{eq:handy-bound}}{\leq}M_{q}e^{t}\,,\\
\msf C & \leq N\int_{\widetilde{\K}^{c}}\frac{\D\nu_{h}}{\D\pi_{h}}\,\D\pi_{h}\leq NM_{2}\bpar{\pi_{h}(\widetilde{\K}^{c})}^{1/2}\,.
\end{align*}
Using Lemma~\ref{lem:exp-eff-domain} to $\msf C$,
\[
\E_{\nu_{h}}\bbrack{\frac{1}{\ell}\wedge N}\leq M_{q}\,\Bpar{4e^{t}+N\exp\bpar{-\frac{t^{2}}{4c}+\frac{t}{2}}}\leq5M_{q}\log S\,.
\]

We combine all ingredients to bound the query complexity of $\psexp$.
\begin{proof}
[Proof of Theorem~\ref{thm:body-exp-samp}] By \eqref{eq:psexp-mixing},
$\psexp$ achieves $\eu R_{2}(\mu_{k}\mmid\pi)\leq\veps$ if
\[
k\gtrsim n^{2}(\norm{\cov\pi^{X}}\vee1)\log\frac{kM_{2}}{\eta}\log\frac{M_{2}}{\veps}\bpar{\gtrsim h^{-1}(\norm{\cov\pi^{X}}\vee1)\log\frac{M_{2}}{\veps}}\,,
\]
which is fulfilled if $k\gtrsim n^{2}(\norm{\cov\pi^{X}}\vee1)\log^{2}\frac{M_{2}}{\eta\veps}$
as claimed. Under the choices of $h$ and $N$, each iteration succeeds
with probability at least $1-\eta/k$, so the total failure probability
is at most $\eta$. Also, the total query complexity during $k$ iterations
is
\[
\O(kM_{c}\log S)=\Otilde\bpar{M_{c}n^{2}(\norm{\cov\pi^{X}}\vee1)\log^{3}\frac{1}{\eta\veps}}\,.
\]

Consider the following procedure: generate $x\sim\mu^{X}$ and draw
$t\sim\pi^{T|X=x}\propto e^{-nt}\cdot\ind[t\geq\nicefrac{V(x)}{n}]$,
obtaining $(x,t)$, where $\pi^{T|X=x}$ can be sampled by drawing
$u\sim\Unif\,([0,1])$ and taking $F^{-1}(u)$, where $F$ is the
CDF of $\pi^{T|X=x}$. Then, the density of the law of $(x,t)$ is
$\mu^{X}\cdot\pi^{T|X}$, and it follows from $\pi(x,t)=\pi^{X}\cdot\pi^{T|X}$
that
\[
\frac{\mu^{X}\pi^{T|X}}{\pi}=\frac{\mu^{X}}{\pi^{X}}\,.
\]
Hence, 
\[
\Bnorm{\frac{\mu^{X}\pi^{T|X}}{\pi}}_{L^{q}(\pi)}=\Bnorm{\frac{\mu^{X}}{\pi^{X}}}_{L^{q}(\pi)}=\Bnorm{\frac{\mu^{X}}{\pi^{X}}}_{L^{q}(\pi^{X})}\,,
\]
which completes the proof.
\end{proof}

\subsubsection{Sampling from a tilted Gaussian distribution}

Just as we used Gaussians as annealing distributions for uniform sampling,
we will need to sample from a distribution of the form 
\[
\mu_{\sigma^{2},\rho}(x,t)\propto\exp\bpar{-\frac{1}{2\sigma^{2}}\,\norm x^{2}-\rho t}\cdot\ind_{\K}(x,t)\,.
\]
As seen later, in the annealing scheme, we increase $\rho$ and $\sigma^{2}$
in suitable speed so that we can arrive at $\pi^{Z}(x,t)\propto\exp(-nt)|_{\K}$.
Roughly speaking, $\gamma_{\sigma^{2}}$ tames the $x$-direction
while $e^{-\rho t}$ handles the $t$-direction.

\cite{KV25sampling} proposed  $\psann$ to sample from these annealing
distributions. For $v:=(x,t)\in\Rn\times\R$ and $w:=(y,s)\in\Rn\times\R$,
$\psann$ considers the augmented target given as 
\[
\mu^{V,W}(v,w)\propto\exp\bpar{-\frac{1}{2\sigma^{2}}\,\norm x^{2}-\rho t-\frac{1}{2h}\,\norm{w-v}^{2}}\big|_{\K}\,,
\]
then $\psann$ with variance $h$ alternates the following: for $\tau:=\frac{\sigma^{2}}{h+\sigma^{2}}<1$,
$y_{\tau}:=\tau y$, and $h_{\tau}:=\tau h$,
\begin{itemize}
\item {[}Forward{]} $w\sim\mu^{W|V=v}=\mc N(v,hI_{n+1})$.
\item {[}Backward{]} Sample
\[
v\sim\mu^{V|W=w}\propto\exp\bpar{-\frac{1}{2\sigma^{2}}\,\norm x^{2}-\rho t-\frac{1}{2h}\,\norm{w-v}^{2}}\big|_{\K}\propto\bbrack{\mc N(y_{\tau},h_{\tau}I_{n})\otimes\mc N(s-\rho h,h)}\big|_{\K}\,,
\]
\end{itemize}
where we use rejection sampling with the proposal $\mc N(y_{\tau},h_{\tau}I_{n})\otimes\mc N(s-\rho h,h)$
for the backward step. The success probability of each trial at $w=(y,s)$
is 
\[
\ell(w):=\frac{1}{(2\pi h_{\tau})^{n/2}(2\pi h)^{1/2}}\,\int_{\mc{\bar{K}}}\exp\bpar{-\frac{1}{2h_{\tau}}\,\bnorm{x-y_{\tau}}^{2}-\frac{1}{2h}\,\abs{t-(s-\rho h)}^{2}}\,\D x\D t\,.
\]

\paragraph{Mixing analysis. }

As done for warm-start generation for uniform sampling, we also truncate
$\K=\{(x,t):V(x)\leq nt\}$ to $B_{R}(0)\times[\pm\O(1)]$, denoted
by $\bar{\K}$, so that $\hat{\pi}:=\pi|_{\bar{\K}}$ is $\O(1)$-close
to $\pi$ in $\eu R_{\infty}$. We specify this preprocessing in \S\ref{subsec:warm-generation-LC}.

As for the mixing rate, it follows from \eqref{eq:LSI-contraction}
that $\psann$ with initial distribution $\mu_{0}\ll\mu_{\sigma^{2},\rho}$
achieves $\veps$-distance in $\eu R_{2}$ after at most
\begin{equation}
k\gtrsim h^{-1}\clsi(\mu_{\sigma^{2},\rho})\log\frac{\eu R_{2}(\mu_{0}\mmid\mu_{\sigma^{2},\rho})}{\veps}\label{eq:psann-mixing}
\end{equation}
iterations. One can bound $\clsi(\mu_{\sigma^{2},\rho})$ via the
Bakry--\'Emery criterion and bounded perturbation.
\begin{lem}
[{\cite[Lemma 3.3]{KV25sampling}}]\label{lem:lsi-annealing} $\clsi(\mu)\lesssim\sigma^{2}\vee1$
for $\mu(x,t)\propto\exp(-\frac{1}{2\sigma^{2}}\,\norm x^{2}-\rho t)|_{\bar{\K}}$.
\end{lem}

We provide another bound on the LSI constant. 
\begin{lem}
\label{lem:global-lsi-ann} Let $D$ be the diameter of $\bar{\K}$
in the $x$-direction. If $\sigma^{2}\gtrsim D\,\norm{\cov\pi^{X}}^{1/2}\log^{2}n\log^{2}\frac{D^{2}}{\norm{\cov\pi^{X}}}$,
\[
\clsi(\mu_{\sigma^{2},n})\lesssim(D\vee1)\,(\norm{\cov\pi^{X}}^{1/2}\vee1)\log n\,.
\]
\end{lem}

\begin{proof}
By Theorem~\ref{thm:lsi-general-bound}, $\clsi(\mu_{\sigma^{2},n})\lesssim(D\vee1)\,\norm{\cov\mu_{\sigma^{2},n}}^{1/2}\log n$.
Note that $\mu_{\sigma^{2},n}(x,t)=\hat{\pi}(x,t)\,\gamma_{\sigma^{2}}(x)$.
As $\hat{\pi}$ is $\O(1)$-close to $\pi$ in $\eu R_{\infty}$,
it is clear that $\hat{\pi}\gamma_{\sigma^{2}}$ is also $\O(1)$-close
to $\pi\gamma_{\sigma^{2}}$ in $\eu R_{\infty}$. Thus,
\[
\norm{\cov\mu_{\sigma^{2},n}}=\norm{\cov\hat{\pi}\gamma_{\sigma^{2}}}\lesssim\norm{\cov\pi\gamma_{\sigma^{2}}}\leq2\,(\norm{\cov\pi^{X}\gamma_{\sigma^{2}}}+\norm{\cov(\pi\gamma_{\sigma^{2}})^{T}})\lesssim\norm{\cov\pi^{X}}\vee1\,,
\]
where the last inequality follows from Theorem~\ref{thm:cov-Gauss}.
\end{proof}

\paragraph{Preliminaries.}

We recall helper lemmas for the analysis of per-step guarantees. We
first deduce the density of $\mu^{W}=\mu^{V}*\gamma_{h}$.
\begin{lem}
[{\cite[Lemma 3.7]{KV25sampling}}] \label{lem:anneal-muW-density}
For $\mu^{V}\propto\exp(-\frac{\norm x^{2}}{2\sigma^{2}}-\rho t)|_{\bar{\K}}$,
\[
\mu^{W}(y,s)=\frac{\tau^{n/2}\ell(y,s)}{\int_{\bar{\K}}\exp(-\frac{1}{2\sigma^{2}}\norm x^{2}-\rho t)\,\D x\D t}\,\exp\bpar{-\frac{1}{2\tau\sigma^{2}}\,\norm{y_{\tau}}^{2}}\exp\bpar{-\rho s+\half\,\rho^{2}h}\,.
\]
\end{lem}

The following is the essential domain of $\mu^{W}$:
\begin{equation}
\widetilde{\K}=\left[\begin{array}{cc}
\tau^{-1}I_{n}\\
 & 1
\end{array}\right]\bar{\mc K}_{\delta}+\left[\begin{array}{c}
0\\
\rho h
\end{array}\right]\,.\label{eq:annealing-eff-domain}
\end{equation}

\begin{lem}
[{\cite[Lemma 3.8]{KV25sampling}}]\label{lem:annealing-eff-domain}
Under $\eval(V)$, for $\widetilde{\K}$ in \eqref{eq:annealing-eff-domain},
if $\delta\geq24hn$, then
\[
\mu^{W}(\widetilde{\K}^{c})\leq\exp\bpar{-\frac{\delta^{2}}{2h}+24\delta n+hn^{2}}\,.
\]
\end{lem}

Choosing $h=\tfrac{c}{24^{2}n^{2}}$ and $\delta=\tfrac{c/24+t}{n}$
for some $c,t>0$, we can ensure that $\delta\geq24hn$ and
\[
\mu^{W}(\widetilde{\K}^{c})\leq\exp\bpar{-\frac{t^{2}}{2c}+c}\,.
\]

\begin{lem}
[{\cite[Lemma 3.9]{KV25sampling}}] \label{lem:annealing-expansion}
In the setting of Lemma~\ref{lem:annealing-eff-domain}, for $\tau=\frac{\sigma^{2}}{h+\sigma^{2}}<1$,
$s>0$, and $\rho\in(0,n]$,
\[
\tau^{-n/2}\int_{\bar{\K}_{s}}e^{-\frac{1}{2\tau\sigma^{2}}\,\norm z^{2}-\rho l}\,\D z\D l\leq2\exp(hn^{2}+24sn)\int_{\bar{\K}}e^{-\frac{1}{2\sigma^{2}}\,\norm z^{2}-\rho l}\,\D z\D l\,.
\]
\end{lem}

We will use the following inequality:
\begin{align}
\frac{\int_{\widetilde{\K}}\tau^{n/2}\exp(-\frac{1}{2\tau\sigma^{2}}\,\norm{y_{\tau}}^{2}-\rho s+\half\,\rho^{2}h)}{\int_{\bar{\K}}\exp(-\frac{1}{2\sigma^{2}}\,\norm x^{2}-\rho t)} & \underset{(i)}{=}\frac{\tau^{-n/2}\int_{\bar{\K}_{\delta}}\exp(-\frac{1}{2\tau\sigma^{2}}\,\norm y^{2}-\rho s-\half\,\rho^{2}h)}{\int_{\bar{\K}}\exp(-\frac{1}{2\sigma^{2}}\,\norm x^{2}-\rho t)}\nonumber \\
 & \underset{(ii)}{\leq}2\exp(hn^{2}+24\delta n)\leq2e^{1.1c+24t}\,,\label{eq:tilt-Gaussian-ineq}
\end{align}
where $(i)$ follows from change of variables, and $(ii)$ follows
from Lemma~\ref{lem:annealing-expansion}.

\paragraph{(1) Failure probability.}

For a distribution $\nu\ll\mu^{V}$, the failure probability is bounded
as 
\[
\E_{\nu_{h}}[(1-\ell)^{N}]\leq M_{2}\sqrt{\E_{\mu_{h}}[(1-\ell)^{2N}]}\,,
\]
where $\nu_{h}=\nu*\gamma_{h}$, $\mu_{h}=\mu^{V}*\gamma_{h}$, and
$M_{2}:=\norm{\D\nu/\D\mu^{V}}_{L^{2}(\mu^{V})}$. Emulating the proof
of \cite[Lemma 3.10]{KV25sampling} with $N$ there replaced by $2N$
and $M$ there replaced by $M_{2}$, we decompose
\[
\E_{\mu_{h}}[(1-\ell)^{2N}]=\int_{\widetilde{\K}^{c}}\cdot+\int_{\widetilde{\K}\cap[\ell\geq N^{-1}\log(3kM_{2}/\eta)]}\cdot+\int_{\widetilde{\K}\cap[\ell\leq N^{-1}\log(3kM_{2}/\eta)]}=:\msf A+\msf B+\msf C\,.
\]
Then,
\begin{align*}
\msf A & \le\mu^{W}(\widetilde{\K}^{c})\leq\exp\bpar{-\frac{t^{2}}{2c}+c}\,,\\
\msf B & \leq\int_{\widetilde{\K}\cap[\ell\geq N^{-1}\log(3kM_{2}/\eta)]}\exp(-2\ell N)\,\D\mu^{W}\leq\bpar{\frac{\eta}{3kM_{2}}}^{2}\,,\\
\msf C & \leq\int_{\widetilde{\K}\cap[\ell\leq N^{-1}\log(3kM_{2}/\eta)]}\frac{\tau^{n/2}\ell(y,s)\exp(-\frac{1}{2\tau\sigma^{2}}\,\norm{y_{\tau}}^{2})\exp(-\rho s+\half\,\rho^{2}h)}{\int_{\bar{\K}}\exp(-\frac{1}{2\sigma^{2}}\,\norm x^{2}-\rho t)}\\
 & \leq\frac{\log(3kM_{2}/\eta)}{N}\,\frac{\int_{\widetilde{\K}}\tau^{n/2}\exp(-\frac{1}{2\tau\sigma^{2}}\,\norm{y_{\tau}}^{2}-\rho s+\half\,\rho^{2}h)}{\int_{\bar{\K}}\exp(-\frac{1}{2\sigma^{2}}\,\norm x^{2}-\rho t)}\underset{\eqref{eq:tilt-Gaussian-ineq}}{\leq}\frac{\log(3kM_{2}/\eta)}{N}\,2e^{25t}\,,
\end{align*}
where the last inequality follows from the choice of $c=\frac{(\log\log S)^{2}}{4\cdot24^{2}\log S}$
and $t=\tfrac{1}{25}\log\log S$ for $S=\frac{16kM_{2}}{\eta}$. Taking
$N=2S^{2}\log^{2}S$, we obtain that
\[
\E_{\mu_{h}}[(1-\ell)^{2N}]\leq\bpar{\frac{\eta}{kM_{2}}}^{2}\,.
\]
Therefore, the total failure probability during $k$ iterations is
at most $\eta$ by the union bound.

\paragraph{(2) Complexity of the backward step.}

Let $p=1+\frac{1}{\alpha}$ and $q=1+\alpha$ with $\alpha=\log N\geq1$.
Then,
\[
\E_{\nu_{h}}\bbrack{\frac{1}{\ell}\wedge N}=\int_{\widetilde{\K}\cap[\ell\geq N^{-p}]}\cdot+\int_{\widetilde{\K}\cap[\ell<N^{-p}]}\cdot+\int_{\widetilde{\K}^{c}}\cdot=:\msf A+\msf B+\msf C\,,
\]
where
\begin{align*}
\msf A & \leq M_{q}\,\Bpar{\int_{\widetilde{\K}\cap[\ell\geq N^{-p}]}\frac{1}{\ell^{p}}\wedge N^{p}\,\D\mu_{h}}^{1/p}\\
 & \leq M_{q}\,\Bpar{\int_{\widetilde{\K}\cap[\ell\geq N^{-p}]}\frac{1}{\ell^{p-1}}\,\frac{\tau^{n/2}\exp(-\frac{1}{2\tau\sigma^{2}}\,\norm{y_{\tau}}^{2}-\rho s+\half\,\rho^{2}h)}{\int_{\bar{\K}}\exp(-\frac{1}{2\sigma^{2}}\,\norm x^{2}-\rho t)}\,\D y}^{1/p}\\
 & \leq M_{q}N^{1/\alpha}\Bpar{\frac{\int_{\widetilde{\K}}\tau^{n/2}\exp(-\frac{1}{2\tau\sigma^{2}}\,\norm{y_{\tau}}^{2}-\rho s+\half\,\rho^{2}h)}{\int_{\bar{\K}}\exp(-\frac{1}{2\sigma^{2}}\,\norm x^{2}-\rho t)}}^{1/p}\underset{\eqref{eq:tilt-Gaussian-ineq}}{\leq}M_{q}e^{1+25t}\,,\\
\msf B & \leq N\int_{\widetilde{\K}\cap[\ell<N^{-p}]}\frac{\D\nu_{h}}{\D\mu_{h}}\,\D\mu_{h}\leq NM_{q}\,\Bpar{\int_{\widetilde{\K}\cap[\ell<N^{-p}]}\frac{\tau^{n/2}\ell(y)\exp(-\frac{1}{2\tau\sigma^{2}}\,\norm{y_{\tau}}^{2}-\rho s+\half\,\rho^{2}h)}{\int_{\bar{\K}}\exp(-\frac{1}{2\sigma^{2}}\,\norm x^{2}-\rho t)}\,\D y}^{1/p}\\
 & \underset{\eqref{eq:tilt-Gaussian-ineq}}{\leq}M_{q}e^{25t}\,,\\
\msf C & \leq N\int_{\widetilde{\K}^{c}}\frac{\D\nu_{h}}{\D\mu_{h}}\,\D\mu_{h}\leq NM_{2}\bpar{\mu_{h}(\widetilde{\K}^{c})}^{1/2}\leq NM_{2}\exp\bpar{-\frac{t^{2}}{4c}+\frac{c}{2}}\,.
\end{align*}
Adding these up,
\[
\E_{\nu_{h}}\bbrack{\frac{1}{\ell}\wedge N}\leq M_{q}\,\Bpar{4e^{1+25t}+N\exp\bpar{-\frac{t^{2}}{4c}+\frac{c}{2}}}\lesssim M_{q}\log S\,.
\]

Combining the per-step guarantees, we can bound query complexity of
sampling from $\mu_{\sigma^{2},\rho}$.
\begin{proof}
[Proof of Theorem~\ref{thm:body-exp-tilt-Gauss}] By \eqref{eq:psann-mixing},
$\psann$ can achieve $\eu R_{2}(\nu_{k}\mmid\mu)\leq\veps$ if 
\[
k\gtrsim n^{2}(\sigma^{2}\vee1)\log\frac{kM_{2}}{\eta}\log\frac{\log M_{2}}{\veps}\bpar{\gtrsim h^{-1}\clsi(\mu_{\sigma^{2},\rho})\log\frac{\eu R_{2}(\mu_{0}\mmid\mu_{\sigma^{2},\rho})}{\veps}}\,.
\]
This is fulfilled if $k\gtrsim n^{2}(\sigma^{2}\vee1)\log^{2}\frac{M_{2}}{\eta\veps}$
as claimed. From the earlier analysis of failure probability and expected
number of queries, $\psann$ succeeds with probability at least $1-\eta$,
using 
\[
\Otilde\bpar{M_{c}n^{2}(\sigma^{2}\vee1)\log^{3}\frac{1}{\eta\veps}}
\]
queries in total, where $c=1+\log N\leq6\log\frac{16kM_{2}}{\eta}$.

When $\rho=n$ and $\sigma^{2}\gtrsim R\lda^{1/2}\log^{2}n\log^{2}R\lda^{-1/2}$
for $\lda=\norm{\cov\pi}$, by Lemma~\ref{lem:global-lsi-ann} we
could use
\[
\clsi(\mu_{\sigma^{2},n})\lesssim(R\vee1)(\lda^{1/2}\vee1)\log n.
\]
Using a similar argument, it suffices to use
\[
k=\Otilde\bpar{n^{2}(R\vee1)(\lda^{1/2}\vee1)\log^{2}\frac{M_{2}}{\eta\veps}}
\]
iterations, and the total query complexity will be
\[
\Otilde\bpar{M_{c}n^{2}(R\vee1)(\lda^{1/2}\vee1)\log^{3}\frac{1}{\eta\veps}}\,.
\]
\end{proof}

\subsection{Faster warm-start generation for logconcave distributions\label{subsec:warm-generation-LC}}

\subsubsection{Algorithm\label{subsec:TCG-algorithm}}

For a target distribution $\pi^{X}\propto\exp(-V)$ over $\Rn$, we
first reduce it to the exponential distribution $\pi(x,t)\propto e^{-nt}|_{\K}$
as in \eqref{eq:exp-red}, and then truncate $\K$ to a smaller convex
domain. As demonstrated in \cite[Convex truncation and \S3.1]{KV25sampling},
for any given $\veps>0$ and $l=\log\frac{2e}{\veps}$, one may assume
that $x_{0}=0$ and then consider
\begin{equation}
\bar{\K}:=\K\cap\{B_{Rl}(0)\times[-21,13l-6]\}\,,\label{eq:convex-truncation}
\end{equation}
which satisfies $\pi(\R^{n+1}\backslash\bar{\K})\leq\veps$. Taking
$\veps=1/2$, we can ensure that $\bar{\pi}=\pi|_{\bar{\K}}$ is $\O(1)$-close
to $\pi$ in $\eu R_{\infty}$. Let $D:=Rl$. We note that $\bar{\K}$
has the diameter $D$ and $30$ in the $x$ and $t$-direction, respectively.

We now adapt and accelerate $\tgc$, a warm-start generating algorithm
proposed in \cite{KV25sampling}. We denote $\mu_{i}:=\mu_{\sigma_{i}^{2},\rho_{i}}$
and set failure probability and target accuracy to $\eta/m$ and $\veps/m$,
where $m$ is the number of total phases throughout the algorithm.
\begin{itemize}
\item \textbf{{[}Phase I{]}} $\sigma^{2}$-warming ($n^{-1}\leq\sigma^{2}\leq1$)
\begin{itemize}
\item Initialization: Sample from $\mu_{0}\propto\exp(-\frac{n}{2}\,\norm x^{2})|_{\bar{\K}}$
(run rejection sampling with proposal $\mc N(0,n^{-1}I_{n})\otimes\text{Unif}\,([-21,13l-6])$).
\item Run $\psann$ with initial $\mu_{i}$ and target $\mu_{i+1}$, where
\[
\sigma_{i+1}^{2}=\sigma_{i}^{2}\bpar{1+\frac{1}{(qn)^{1/2}}}\,.
\]
\end{itemize}
\item \textbf{{[}Phase II{]}} $\rho$-annealing ($\sigma^{2}\approx1$,
$1\leq\rho\leq n$)
\begin{itemize}
\item Initialization: Run $\psann$ with initial $\mu_{i}\propto\exp(-\frac{1}{2}\,\norm x^{2})|_{\bar{\K}}$
and target $\mu_{i+1}\propto\exp(-\frac{1}{2}\,\norm x^{2}-t)|_{\bar{\K}}$.
\item Run $\psann$ with initial $\mu_{i}$ and target $\mu_{i+1}$, where
\[
\sigma_{i+1}^{2}=\sigma_{i}^{2}\bpar{1+\frac{1}{\bpar{16q\,(q\vee n)}^{1/2}}}^{-1}\quad\&\quad\rho_{i+1}=\rho_{i}\bpar{1+\frac{1}{\bpar{16q\,(q\vee n)}^{1/2}}}\,.
\]
\item Run the following inner annealing: until $\sigma_{i+1}^{2}\leq1$,
run $\psann$ with initial $\mu_{i+1}$ and \emph{new} $\mu_{i+1}$
defined by 
\[
\sigma_{i+1}^{2}\gets\sigma_{i+1}^{2}\bpar{1+\frac{\sigma_{i+1}}{q^{1/2}D}}\,,
\]
\end{itemize}
\item \textbf{{[}Phase III{]}} $\sigma^{2}$-annealing ($1\leq\sigma^{2}\leq D^{2}$,
$\rho=n$)
\begin{itemize}
\item Run $\psann$ with initial $\mu_{i}=\bar{\pi}\gamma_{\sigma_{i}^{2}}$
and target $\mu_{i+1}=\bar{\pi}\gamma_{\sigma_{i+1}^{2}}$, where
\[
\sigma_{i+1}^{2}=\sigma_{i}^{2}\bpar{1+\frac{\sigma_{i}}{q^{1/2}D}}\,.
\]
\item Termination $(\sigma^{2}=D^{2})$: run $\psexp$ with initial $\mu_{D^{2},n}=\bar{\pi}\gamma_{D^{2}}$
and target $\pi\propto\exp(-nt)|_{\K}$.
\end{itemize}
\end{itemize}

\subsubsection{Analysis}

Just as in the previous section, we first carefully pick all parameters.

\paragraph{Choice of parameters.}

For $q\geq2$, it is clear that consecutive distributions in Phase
I and Phase II are $\O(1)$ in $\eu R_{q}$ by Lemma~\ref{lem:global-annealing},
while closeness within the inner annealing of Phase II follows from
Lemma~\ref{lem:variance-annealing}. To see this, we note that an
annealing distribution in the inner annealing can be written as 
\[
\mu_{j}=\exp\bpar{-\rho t-\frac{1}{2\sigma_{j}^{2}}\,\norm x^{2}}\big|_{\bar{\K}}\,,
\]
so its $X$-marginal can be written as $\mu_{j}^{X}\propto\nu\gamma_{\sigma_{j}^{2}}$
for some logconcave distribution $\nu$ over $\Rn$. Then, one can
readily check that
\[
\Bnorm{\frac{\D\mu_{j}}{\D\mu_{j+1}}}_{L^{q}(\mu_{j+1})}=\Bnorm{\frac{\D\mu_{j}^{X}}{\D\mu_{j+1}^{X}}}_{L^{q}(\mu_{j+1}^{X})}\,.
\]
Closeness in Phase III also follows in the same way via Lemma~\ref{lem:variance-annealing}.

The number of inner phases is bounded by 
\[
m\leq(qn)^{1/2}\log n+\bpar{16q\,(q\vee n)}^{1/2}\log n\times\frac{D}{(q\vee n)^{1/2}}+q^{1/2}D\log D^{2}\lesssim q^{1/2}n^{1/2}D\log nD=:m_{\max}(q)\,.
\]
By Theorem~\ref{thm:body-exp-tilt-Gauss}, the number of iterations
for each inner phase is 
\[
k=\Otilde\Bpar{n^{2}\bpar{(\sigma^{2}\vee1)+(D\vee1)(\lda^{1/2}\vee1)}\log^{2}\frac{m_{\max}^{2}(q)}{\eta\veps}}\lesssim\Otilde\bpar{n^{2}D^{2}\log^{2}\frac{q}{\eta\veps}}=:k_{\max}(q)\,,
\]
and it suffices to take
\[
q\gtrsim\log\frac{qnD}{\eta\veps}\bpar{\gtrsim\log\frac{k_{\max}(q)m_{\max}(q)}{\eta}\gtrsim6\log\frac{16kmM_{2}}{\eta}}
\]
in order to have a provable bound on a query complexity of $\psann$
under $\eu R_{q}$-warmness. Since the condition on $q$ is fulfilled
if
\[
q\gtrsim\log\frac{nD}{\eta\veps}=\Otilde(1)\,,
\]
we set $q$ to the RHS above, from which $k_{\max}(q)$ and $m_{\max}(q)$
are determined. 

\paragraph{Complexity bound.}

We now bound the query complexity of each phase.
\begin{lem}
[Phase I,  $\sigma^2$-warming] With probability at least $1-\eta$,
Phase I outputs a sample whose law $\nu$ satisfies 
\[
\norm{\nu-\gamma|_{\bar{\K}}}_{\tv}\leq\veps\,,
\]
using $\Otilde(n^{5/2}\log^{4}\nicefrac{D}{\eta\veps})$ evaluation
queries in expectation.
\end{lem}

\begin{proof}
At initialization, one can readily check that $\mu_{0}$ can be sampled
by rejection sampling with $\O(1)$  queries in expectation (see \cite[Lemma 5.6]{KV25sampling}).

For any given $\sigma^{2}\in[n^{-1},1]$, we need at most $(qn)^{1/2}$
inner phases to double the initial $\sigma^{2}$. Any pair of consecutive
distributions is $\O(1)$-close in $\eu R_{q}$ by Lemma~\ref{lem:global-annealing}.
By Theorem~\ref{thm:body-exp-tilt-Gauss}, each inner phase then
requires 
\[
\Otilde\bpar{n^{2}\log^{3}\frac{m_{\max}^{2}}{\eta\veps}}=\Otilde\bpar{n^{2}\log^{3}\frac{D}{\eta\veps}}
\]
evaluation queries in expectation. Therefore, Phase I uses 
\[
\Otilde\bpar{n^{2}\log^{3}\frac{D}{\eta\veps}}\times(qn)^{1/2}\times\log n=\Otilde\bpar{n^{5/2}\log^{4}\frac{D}{\eta\veps}}
\]
evaluation queries in expectation. The final failure probability is
immediate from the union bound, while the final $\tv$-guarantee follows
from the triangle inequality as in \S\ref{sec:techniques}.
\end{proof}
\begin{lem}
[Phase II,  $\rho$-annealing] With probability at least $1-\eta$,
Phase II started at the output of Phase I returns a sample whose law
$\nu$ satisfies
\[
\norm{\nu-\bar{\pi}\gamma}_{\tv}\leq\veps\,,
\]
using $\Otilde(n^{2}(n^{1/2}\vee D)\log^{4}\nicefrac{1}{\eta\veps})$
evaluation queries in total.
\end{lem}

\begin{proof}
At initialization, the two distributions are $\O(1)$-close in $\eu R_{\infty}$
since $\bar{\K}$ has diameter of $\O(1)$ in the $t$-direction.
Hence, the query complexity of $\psann$ is simply $\Otilde(n^{2}\log^{3}\nicefrac{D}{\eta\veps})$.

The outer annealing, which updates $\sigma_{i}^{2}$ and $\rho_{i}$
simultaneously, happens at most $(qn)^{1/2}\log n$ times. For each
outer annealing, the second form of Lemma~\ref{lem:global-annealing}
with $\alpha=-\gamma\,(1+\gamma)^{-1}$ for $\gamma=(16q\,(q\vee n))^{-1/2}$
implies that
\[
\eu R_{q}(\mu_{i}\mmid\mu_{i+1})=\O(1)\,.
\]
Then, the outer annealing via $\psann$ requires $\Otilde(n^{2}\log^{3}\nicefrac{D}{\eta\veps})$
queries in expectation. 

This outer annealing is immediately followed by the inner annealing,
where it increases $\sigma^{2}$ from $(1+(16q\,(q\vee n))^{-1/2})^{-1}$
to $1$. Under the update of $\sigma^{2}\gets\sigma^{2}(1+\frac{\sigma}{q^{1/2}D})$,
there would be at most 
\[
\frac{q^{1/2}D}{\sigma}\times\frac{1}{\bpar{16q\,(q\vee n)}^{1/2}}\leq\frac{D}{(q\vee n)^{1/2}}
\]
inner-annealing phases, and each annealing via $\psann$ also takes
$\Otilde(n^{2}\log^{3}\nicefrac{D}{\eta\veps})$ queries.

Putting these together, the total query complexity throughout the
$\rho$-annealing is
\[
\bpar{16q\,(q\vee n)}^{1/2}\log n\times\Bpar{\Otilde\bpar{n^{2}\log^{3}\frac{D}{\eta\veps}}+\Otilde\bpar{n^{2}\log^{3}\frac{D}{\eta\veps}}\times\frac{D}{(q\vee n)^{1/2}}}\leq\Otilde\bpar{n^{2}(n^{1/2}\vee D)\log^{4}\frac{1}{\eta\veps}}\,,
\]
which completes the proof.
\end{proof}
\begin{lem}
[Phase III,  $\sigma^2$-annealing] For $\lda=\norm{\cov\pi^{X}}$,
with probability at least $1-\eta$, Phase III started at the output
of Phase II returns a sample whose law $\nu$ satisfies
\[
\norm{\nu-\pi}_{\tv}\leq\veps\,,
\]
using $\Otilde(n^{2}D^{3/2}(\lda^{1/4}\vee1)\log^{4}\nicefrac{1}{\eta\veps}\log\nicefrac{D^{2}}{\lda})$
evaluation queries in expectation.
\end{lem}

\begin{proof}
For any given $\sigma^{2}\in[1,D^{2}]$, its doubling requires at
most $q^{1/2}D/\sigma$ phases. Any consecutive distributions are
$\O(1)$-close in $\eu R_{q}$ (Lemma~\ref{lem:variance-annealing}),
so $\psann$ requires $\Otilde(n^{2}\sigma^{2}\log^{3}\nicefrac{D}{\eta\veps})$
queries. Thus, one doubling takes 
\[
\frac{q^{1/2}D}{\sigma}\times\Otilde\bpar{n^{2}\sigma^{2}\log^{3}\frac{D}{\eta\veps}}\leq\Otilde\bpar{n^{2}D^{3/2}(\lda^{1/4}\vee1)\log^{4}\frac{1}{\eta\veps}\log\frac{D^{2}}{\lda}}
\]
until $\sigma^{2}\lesssim D\,(\lda^{1/2}\vee1)\log^{2}n\log^{2}\nicefrac{D^{2}}{\lda}$.

When $\sigma^{2}$ exceeds this threshold, $\psann$ requires $\Otilde(n^{2}D\,(\lda^{1/2}\vee1)\log^{3}\nicefrac{1}{\eta\veps})$
queries in expectation. Hence, one doubling in this regime takes 
\[
\frac{q^{1/2}D}{\sigma}\times\Otilde\bpar{n^{2}D\,(\lda^{1/2}\vee1)\log^{3}\frac{1}{\eta\veps}}\leq\Otilde\bpar{n^{2}D^{3/2}(\lda^{1/4}\vee1)\log^{4}\frac{1}{\eta\veps}}\,.
\]

At termination, we note that
\[
\frac{\bar{\pi}\gamma_{D^{2}}}{\pi}=\frac{\bar{\pi}\gamma_{D^{2}}}{\bar{\pi}}\,\frac{\bar{\pi}}{\pi}\leq\bpar{\pi(\bar{\K})}^{-1}\leq2\,.
\]
By Theorem~\ref{thm:body-exp-samp}, $\psexp$ requires $\Otilde(n^{2}(\lda\vee1)\log^{3}\frac{D}{\eta\veps})$
queries in expectation. Adding these up, the total query complexity
is
\[
\Otilde\bpar{n^{2}D^{3/2}(\lda^{1/4}\vee1)\log^{4}\frac{1}{\eta\veps}\log\frac{D^{2}}{\lda}}
\]
as claimed.
\end{proof}
Combining the previous three lemmas together, we can establish the
query complexity of sampling from a logconcave distribution as claimed
in Theorem~\ref{thm:body-LC-sampling}.
\begin{proof}
[Proof of Theorem~\ref{thm:body-LC-sampling}] We simply add up
the query complexities in the previous three lemmas, with $D\asymp R$.
That is,
\begin{align*}
 & n^{5/2}\log^{4}\frac{R}{\eta\veps}+n^{2}(n^{1/2}\vee R)\log^{4}\frac{1}{\eta\veps}+n^{2}R^{3/2}(\lda^{1/4}\vee1)\log^{4}\frac{1}{\eta\veps}\log\frac{D^{2}}{\lda}\\
 & =\Otilde\Bpar{n^{2}\max\bbrace{n^{1/2},R^{3/2}(\lda^{1/4}\vee1)}\log^{5}\frac{R\lda^{-1/2}}{\eta\veps}}\,.
\end{align*}
When $\pi^{X}$ is further near-isotropic, the claim follows from
$R\lesssim n^{1/2}$ and $\lda\asymp1$.
\end{proof}

\begin{acknowledgement*}
Yunbum Kook thanks Sinho Chewi for a discussion on log-Sobolev inequalities
for logconcave distributions with compact support when YK was visiting
the IAS in 2024. This work was supported in part by NSF Award CCF-2106444
and a Simons Investigator grant.
\end{acknowledgement*}
\bibliographystyle{alpha}
\bibliography{main}

\appendix

\section{Another proof via stochastic localization \label{app:LSI-interpolation-SL}}

In \S\ref{subsec:LSI-interpolation}, we proved $\clsi(\pi)\lesssim\max\{D\lda^{1/2},D^{2}\wedge\lda\log^{2}n\}$
using Bizeul's result on exponential concentration (Theorem~\ref{thm:bizeul-concentration})
and equivalence between Gaussian concentration and \eqref{eq:lsi}.
In this section, we present another proof for Theorem~\ref{thm:lsi-general-bound}
using Eldan's stochastic localization, which was our first proof in
an earlier version. We essentially follow the approach in~\cite[\S5]{LV24eldan},
with adaptations to accommodate the general (non-isotropic) setting.

\paragraph{Outline.}

Stochastic localization $(\pi_{t})_{t\geq0}$ is a density-valued
process driven by a stochastic linear tilting, characterized by the
stochastic differential equation, $\D\pi_{t}(x)=\pi_{t}(x)\inner{x-b_{t},\D B_{t}}$
(so $\pi_{t}(\cdot)$ is a martingale), where $b_{t}=\int x\,\pi_{t}(\D x)$
is the barycenter of $\pi_{t}$, and $B_{t}$ is a standard Brownian
motion. The explicit solution to SL is given by
\[
\pi_{t,\theta_{t}}(x)\propto\pi(x)\exp(-\inner{\theta_{t},x}-\frac{t}{2}\,\norm x^{2})\,,
\]
where the tilt process $(\theta_{t})_{t\geq0}$ satisfies $\D\theta_{t}=b_{t,\theta_{t}}\,\D t+\D B_{t}$,
and $b_{t,\theta_{t}}$ is the barycenter of $\pi_{t,\theta_{t}}$.
The distribution $\pi_{t}$ is $t$-strongly logconcave.

At a high level, we generalize a bound on $\norm{\cov\pi_{t}}$ from~\cite{KL22Bourgain}
and the approach in~\cite{LV24eldan} used to prove $\clsi(\pi)\lesssim D$
for isotropic logconcave distributions. The SL process is not affine-invariant
due to its dependence on standard Brownian motion, making generalization
to non-isotropic cases non-trivial. Moreover, the isotropic result
uses $D\geq n^{1/2}$, which may not hold for general cases.

We now sketch the proof. Recall the \emph{log-Cheeger constant} of
a probability measure $\pi$ defined as
\[
\clch(\pi)=\inf_{E:\,\pi(E)\leq\nicefrac{1}{2}}\frac{\pi(\partial E)}{\pi(E)\sqrt{\log\tfrac{1}{\pi(E)}}}\,.
\]
Since $\clsi(\pi)\asymp\clch^{-2}(\pi)$ for any logconcave probability
measure $\pi$ \cite{Ledoux94simple}, we focus on lower-bounding
$\clch(\pi)$ instead. Using the Bakry--\'Emery criterion combined
with this equivalence, we have $\clch(\pi)\gtrsim t^{1/2}$ for any
$t$-strongly logconcave probability measure $\pi$, which implies
that for the SL process $(\pi_{t})_{t\geq0}$ with $\pi_{0}=\pi$,
\[
\pi_{t}(\de E)\gtrsim t^{1/2}\,\pi_{t}(E)\sqrt{\log\tfrac{1}{\pi_{t}(E)}}\,.
\]
Since $\pi_{t}(x)$ is a martingale for $x\in\Rn$ 
\[
\pi(\de E)=\E_{\msf{SL}}[\pi_{t}(\de E)]\gtrsim t^{1/2}\,\E_{\msf{SL}}\Bbrack{\pi_{t}(E)\sqrt{\log\tfrac{1}{\pi_{t}(E)}}}\,.
\]
The main challenge is to determine how long we can run SL while $\pi_{t}(E)\log^{1/2}\nicefrac{1}{\pi_{t}(E)}$
remains close to its initial value $\pi(E)\log^{1/2}\nicefrac{1}{\pi(E)}$
with high probability.

To quantify the deviation of the martingale $g_{t}:=\pi_{t}(E)$ from
$\E_{\msf{SL}}g_{t}=\pi(E)$, we use its quadratic variation bound
$\D[g]_{t}\leq30\,\norm{\cov\pi_{t}}\,g_{t}^{2}\log^{2}\nicefrac{e}{g_{t}}\,\D t$
\cite{LV24eldan}. We can also control the rate at which $\log g_{t}^{-1}$
changes. As implied from this quadratic variation bound, we need a
high-probability control on $\norm{\cov\pi_{t}}$. In Lemma~\ref{lem:SL-operator},
we generalize the existing result for isotropic logconcave $\pi_{0}$
\cite{KL22Bourgain}, following the streamlined proof in \cite{KL24isop}:
for some universal constant $c>0$,
\[
\P_{\msf{SL}}(\exists\,t\in[0,T]:\norm{\cov\pi_{t}}\geq2\,\norm{\cov\pi})\leq\exp\bpar{-(cT\,\norm{\cov\pi})^{-1}}\,.
\]
Finally, combining these bounds, we demonstrate in Lemma~\ref{lem:random-measure-control}
that SL can be run until $(\max\{D\,\norm{\cov\pi}^{1/2},\,D^{2}\wedge\norm{\cov\pi}\log^{2}n\})^{-1}$,
thereby achieving the desired interpolation.

\paragraph{Stochastic localization.}

Our main tool is \emph{stochastic localization} (SL), which is a density-valued
process defined by
\begin{align}
\pi_{0} & =\pi\nonumber \\
\D\pi_{t}(x) & =\pi_{t}(x)\,\inner{x-b_{t},\D B_{t}}\ \text{for all }x\in\Rn\,,\tag{\ensuremath{\msf{SL}}-\ensuremath{\msf{SDE}}}\label{eq:SL-SDE}
\end{align}
where $b_{t}=\int x\,\D\pi_{t}(x)$ is the barycenter of $\pi_{t}$.
One can think of it as infinitely many SDEs that are coupled through
the barycenter $b_{t}$. Its four important properties are
\begin{enumerate}
\item (Almost surely) $\pi_{t}$ is a probability measure over $\Rn$ (i.e.,
$\int\D\pi_{t}=1$):
\[
\D\int\pi_{t}=\int\D\pi_{t}=\Bigl<\int(x-b_{t})\,\D\pi_{t}(x),\D B_{t}\Bigr>=\inner{0,\D B_{t}}=0\,.
\]
\item $\pi_{t}$ is a martingale with respect to the filtration induced
by the Brownian motion (i.e., $\E_{\msf{SL}}[\pi_{t}(x)]=\pi(x)$
for all $x\in\Rn$), where the expectation is taken over the randomness
given by $\D B_{t}$ (not $\pi$).
\item The solution to SL can be explicitly stated as follows: consider a
tilt process $(\theta_{t})_{t\geq0}$ defined by
\[
\D\theta_{t}=b_{t,\theta_{t}}\,\D t+\D B_{t}\,,
\]
where $b_{t,\theta_{t}}$ is the barycenter of the probability distribution
defined by 
\begin{equation}
\pi_{t,\theta_{t}}(x)\propto\pi(x)\,\exp\bpar{\inner{\theta_{t},x}-\frac{t}{2}\,\norm x^{2}}\propto\pi(x)\,\exp\bpar{-\frac{t}{2}\,\bnorm{x-\frac{\theta_{t}}{t}}^{2}}\,.\tag{\ensuremath{\msf{SL}}-\ensuremath{\msf{pdf}}}\label{eq:SL-pdf}
\end{equation}
The existence and uniqueness of a strong solution to this are standard
(see \cite{Chen21almost}). We abbreviate $\pi_{t}:=\pi_{t,\theta_{t}}$,
which is indeed a solution to \eqref{eq:SL-SDE}. Note that $\pi_{t}$
is $t$-strongly logconcave, since $\pi$ is logconcave.
\item For a test function $F$, the It\^o derivative of the martingale
$M_{t}=\int F(x)\,\pi_{t}(\D x)$ has a moment-generating feature:
\begin{equation}
\D M_{t}=\int_{\Rn}F(x)\,\inner{x-b_{t},\D B_{t}}\,\pi_{t}(\D x)\,.\tag{{\ensuremath{\msf{MG}}}}\label{eq:moment}
\end{equation}
\end{enumerate}

\paragraph{Proof.}

Here, we use $\Sigma_{t}$ to denote the covariance matrix of $\pi_{t}$
obtained by running the SL process with initial logconcave distribution
$\pi_{0}=\pi$, where $\Sigma_{0}:=\Sigma$ is the covariance matrix
of $\pi$.

\paragraph{(1) Operator norm control.}

One of central questions in SL is how long the process can run without
the covariance matrix deviating much in operator norm. This question
has been extensively studied and improved over time, but under the
assumption that the initial distribution $\pi_{0}=\pi$ is isotropic
(for an example of adaptation to the anisotropic setting, see \cite[Lemma 2.8, Theorem B.12]{JLLV24reducing}).
The current best result is the following:
\begin{prop}
[{\cite[Lemma 5.2]{KL22Bourgain}}] Let $\pi$ be an isotropic logconcave
distribution over $\Rn$. For $T=(C\log^{2}n)^{-1}$, we have
\[
\P(\exists\,t\in[0,T]:\norm{\Sigma_{t}}\geq2)\leq\exp\bpar{-\frac{1}{CT}}\,,
\]
where $C>0$ is a universal constant.
\end{prop}

It is not immediately clear how to extend this result to the general
case. The Brownian motion driving an SL process in the isotropized
space (i.e. $x':=\Sigma^{-1/2}x$) becomes \emph{non}-isotropic when
mapped back to the original space. However, we need an SL process
driven by a \emph{standard} Brownian motion in the original space,
so we open up the original proof and adapt it accordingly, while following
the streamlined argument in~\cite[\S7]{KL24isop}.
\begin{lem}
[Operator-norm control] \label{lem:SL-operator} Let $\pi$ be a
logconcave distribution in $\Rn$ with covariance $\Sigma$. For $T=(C\,\norm{\Sigma}\log^{2}n)^{-1}$,
we have
\[
\P(\exists\,t\in[0,T]:\norm{\Sigma_{t}}\geq2\,\norm{\Sigma})\leq\exp\bpar{-\frac{1}{CT\,\norm{\Sigma}}}\,,
\]
where $C>0$ is a universal constant.
\end{lem}

To control $\norm{\Sigma_{t}}$, we work with a proxy function defined
as follows: for a symmetric matrix $M\in\Rnn$ and constant $\beta>0$,
\[
h(M):=\frac{1}{\beta}\,\log(\tr e^{\beta M})\,,
\]
which clearly satisfies that
\[
\norm{\Sigma_{t}}\leq h(\Sigma_{t})\leq\norm{\Sigma_{t}}+\frac{\log n}{\beta}\,.
\]
To control $h(\Sigma_{t})$, we will evaluate the It\^o derivative
$\D h(\Sigma_{t})$ to see the magnitude of its drift. To this end,
we first note that by \eqref{eq:moment},
\begin{align*}
\D b_{t} & =\Sigma_{t}\,\D B_{t}\,,\\
\D\Sigma_{t} & =\underbrace{\int(x-b_{t})^{\otimes2}\inner{x-b_{t},\D B_{t}}\,\D\pi_{t}}_{=:\D H_{t}}-\Sigma_{t}^{2}\,\D t\,,
\end{align*}
where $v^{\otimes2}:=vv^{\T}$. To analyze, we define for each $i\in[n]$,
\[
H_{t,i}=H_{i}:=\int(x-b_{t})^{\otimes2}(x-b_{t})_{i}\,\D\pi_{t}\in\Rnn\,,
\]
so $\D H_{t}=\sum_{i=1}^{n}H_{i}\,\D B_{t,i}$. By It\^o's formula,
\begin{align*}
\D h(\Sigma_{t}) & =\nabla h(\Sigma_{t})[\D\Sigma_{t}]+\frac{1}{2}\,\hess h(\Sigma_{t})[\D\Sigma_{t},\D\Sigma_{t}]\\
 & =\underbrace{\nabla h(\Sigma_{t})[\D H_{t}]}_{=\tr(\nabla h(\Sigma_{t})\,\D H_{t})}-\underbrace{\nabla h(\Sigma_{t})[\Sigma_{t}^{2}]}_{=\tr(\nabla h(\Sigma_{t})\,\Sigma_{t}^{2})}\D t+\frac{1}{2}\,\hess h(\Sigma_{t})[\D H_{t},\D H_{t}]\,.
\end{align*}
From direct computation,
\[
G_{t}:=\nabla h(\Sigma_{t})=\frac{e^{\beta\Sigma_{t}}}{\tr e^{\beta\Sigma_{t}}}\succeq0\,,
\]
and note that $\tr G_{t}=1$. Then,
\begin{align*}
\D h(\Sigma_{t}) & \leq\sum_{i}\tr(G_{t}H_{i})\,\D B_{t,i}+\half\,\hess h(\Sigma_{t})\Bbrack{\sum_{i}H_{i}\,\D B_{t,i},\sum_{i}H_{i}\,\D B_{t,i}}\\
 & =\sum_{i}\tr(G_{t}H_{i})\,\D B_{t,i}+\half\,\sum_{i}\hess h(\Sigma_{t})[H_{i},H_{i}]\,\D t\\
 & \underset{(i)}{\leq}\sum_{i}\tr(G_{t}H_{i})\,\D B_{t,i}+\frac{\beta}{2}\,\sum_{i}\tr(G_{t}H_{i}^{2})\,\D t\\
 & \leq\underbrace{\sum_{i}\tr(G_{t}H_{i})\,\D B_{t,i}}_{\eqqcolon\D Z_{t}}+\frac{\beta}{2}\,\Bnorm{\sum_{i}H_{i}^{2}}\,\D t\,,
\end{align*}
where in $(i)$ the bound on $\hess h(\Sigma_{t})[H_{i},H_{i}]$ follows
from below, and the last line follows from the $(1,\infty)$-H\"older
inequality with $\tr G_{t}=\norm{G_{t}}_{1}\leq1$.
\begin{prop}
[{\cite[Corollary 56]{KL24isop}}] For symmetric matrices $M,H\in\Rnn$,
\[
\hess h(M)[H,H]\leq\beta\,\tr\bpar{\frac{e^{\beta M}}{\tr e^{\beta M}}\,H^{2}}\,.
\]
\end{prop}

Integrating both sides, we have that for some constant $C>0$,
\[
h(\Sigma_{t})\leq h(\Sigma_{0})+Z_{t}+\frac{\beta}{2}\int_{0}^{t}\Bnorm{\sum_{i}H_{s,i}^{2}}\,\D s\leq\norm{\Sigma}+\frac{\log n}{\beta}+Z_{t}+\frac{\beta C}{2}\int_{0}^{t}\frac{\norm{\Sigma_{s}}^{5/2}}{s^{1/2}}\,\D s\,,
\]
where the last line follows from below and Klartag's improved Lichnerowicz
inequality \cite{Klartag23log}, namely, $\cpi(\pi)\leq(\norm{\Sigma}/t)^{1/2}$
for $t$-strongly logconcave distributions. 
\begin{prop}
[{\cite[Lemma 58]{KL24isop}}] Let $X\sim\pi$ be a centered logconcave
random vector. Then, 
\[
\Bnorm{\sum_{i}(\E[X_{i}X^{\otimes2}])^{2}}\leq4\cpi(\pi)\,\norm{\Sigma}^{2}\,.
\]
\end{prop}

Let us take the smallest $\tau$ such that $\tau\leq t$ and $\norm{\Sigma_{\tau}}\geq2\,\norm{\Sigma}$.
Then,
\[
2\,\norm{\Sigma}=\norm{\Sigma_{\tau}}\leq h(\Sigma_{\tau})\leq\norm{\Sigma}+\frac{\log n}{\beta}+Z_{\tau}+\frac{\beta C\tau^{1/2}\,\norm{\Sigma}^{5/2}}{2}\,.
\]
For $\beta=2\,\norm{\Sigma}^{-1}\log n$ and $t\lesssim\norm{\Sigma}^{-1}\log^{-2}n$,
we can ensure $Z_{\tau}\geq\norm{\Sigma}/4$. 

To apply a deviation inequality to $Z_{t}$, let us compute the quadratic
variation of $Z_{t}$:
\[
\D[Z]_{t}=\sum_{i}\tr^{2}(G_{t}H_{i})\,\D t=\norm v^{2}\,\D t\,,
\]
where the vector $v\in\Rn$ satisfies $v_{i}=\tr(G_{t}H_{i})$. For
any unit vector $\theta\in\mbb S^{n-1}$, using the $(1,\infty)$-H\"older
inequality and $\tr G_{t}\leq1$,
\[
v\cdot\theta=\tr\Bpar{G_{t}\sum_{i}H_{i}\theta_{i}}\leq\Bnorm{\sum_{i}H_{i}\theta_{i}}=\Bnorm{\int(x-b_{t})^{\otimes2}\inner{x-b_{t},\theta}\,\D\pi_{t}}\lesssim\norm{\Sigma_{t}}^{3/2}\,,
\]
where the last inequality follows from the next proposition. 
\begin{prop}
[{\cite[Lemma 57]{KL24isop}}] Let $X\sim\pi$ be a centered logconcave
random vector. Then,
\[
\sup_{\theta\in\mbb S^{n-1}}\norm{\E[(X\cdot\theta)\,X^{\otimes2}]}\lesssim\norm{\Sigma}^{3/2}\,.
\]
\end{prop}

Thus, $\norm v^{2}\leq\norm{\Sigma_{t}}^{3}$ and $[Z]_{\tau}\lesssim\int_{0}^{\tau}\norm{\Sigma_{s}}^{3}\,\D s\leq8\,\norm{\Sigma}^{3}t$.
Putting all these together, we conclude that for some universal constant
$C>0$,
\[
\P(\exists\,\tau\leq t:\norm{\Sigma_{\tau}}\geq2\,\norm{\Sigma})\leq\P\bpar{\exists\,\tau>0:Z_{\tau}\geq\frac{\norm{\Sigma}}{4}\ \text{and}\ [Z]_{\tau}\lesssim\norm{\Sigma}^{3}\,t}\leq\exp\bpar{-\frac{1}{C\norm{\Sigma}\,t}}\,,
\]
where the last inequality follows from the classical deviation inequality
below for a local martingale, and this completes the proof of Lemma~\ref{lem:SL-operator}.
\begin{prop}
[Freedman's inequality] \label{prop:freedman} Let $(M_{t})_{t\geq0}$
be a continuous local martingale with $M_{0}=0$. Then for $u,\sigma^{2}>0$,
we have
\[
\P(\exists\,t>0:M_{t}>u\ \text{and}\ [M]_{t}\leq\sigma^{2})\leq\exp\bpar{-\frac{u^{2}}{2\sigma^{2}}}\,.
\]
\end{prop}

\paragraph{(2) Relating $\protect\clch(\pi)$ and $\protect\clch(\pi_{t})$.}

 Recall that $\clch^{-2}(\pi_{t})\asymp\clsi(\pi_{t})$ for logconcave
measures \cite{Ledoux94simple}. Since $\pi_{t}$ is $t$-strongly
logconcave, $\clch(\pi_{t})\gtrsim t^{1/2}$. Thus, for any measurable
subset $E$ of measure $\pi_{t}(E)\leq1/2$,
\[
\pi_{t}(\partial E)\gtrsim\sqrt{t}\,\pi_{t}(E)\sqrt{\log\tfrac{1}{\pi_{t}(E)}}\,.
\]
Since $\pi_{t}(\cdot)$ is almost surely a martingale in $\Rn$, we
have $\E[\pi_{t}(\de E)]=\pi(\de E)$. Taking expectation on both
sides, we obtain that
\[
\pi(\de E)\gtrsim\sqrt{t}\,\E\Bbrack{\pi_{t}(E)\sqrt{\log\tfrac{1}{\pi_{t}(E)}}\,\ind_{[\pi_{t}(E)\leq1/2]}}\,.
\]
Thus, it suffices to show that there exists large enough $T>0$ such
that if $t\leq T$, then for any measurable subset $E$ of measure
$\pi(E)\leq1/2$,
\[
\E\Bbrack{\pi_{t}(E)\sqrt{\log\tfrac{1}{\pi_{t}(E)}}\,\ind_{[\pi_{t}(E)\leq1/2]}}\gtrsim\pi(E)\sqrt{\log\tfrac{1}{\pi(E)}}\,.
\]
To this end, we analyze how fast $\pi_{t}(E)$ and $\log\nicefrac{1}{\pi_{t}(E)}$
deviate from $\pi(E)$ and $\log\nicefrac{1}{\pi(E)}$, respectively. 

For any measurable subset $E$, we denote its measure at time $t$
as $g_{t}:=\pi_{t}(E)$.
\begin{prop}
[{\cite[Lemma 44 and 45]{LV24eldan}}] \label{prop:random-measure-control}
It holds that $\D[g]_{t}\leq30\,\norm{\Sigma_{t}}\,g_{t}^{2}\log^{2}\frac{e}{g_{t}}\,\D t$.
Also, for any $T,\gamma\geq0$, it holds that
\[
\P\bpar{-\gamma\leq\log\frac{1}{g_{t}}-\log\frac{1}{g_{0}}\leq\half\,D^{2}t+\gamma\,,\forall t\in[0,T]}\geq1-4\exp\bpar{-\frac{\gamma^{2}}{2D^{2}T}}\,.
\]
\end{prop}

Using these results and Lemma~\ref{lem:SL-operator}, we can prove
a refined version of \cite[Lemma 46]{LV24eldan}.
\begin{lem}
[Randommeasure control] \label{lem:random-measure-control} Consider
SL $(\pi_{t})_{t\geq0}$ with logconcave $\pi_{0}=\pi$ supported
with diameter $D$. Then, there exists a universal constant $c>0$
such that for any measurable subset $E$ of $\pi(E)\leq1/16$, if
\[
0\leq t\leq T:=c\,\max\Bbrace{\frac{1}{D^{2}}\,\log\frac{1}{\pi(E)},\Bpar{\norm{\Sigma}\,\bpar{\log\frac{1}{\pi(E)}\vee\log^{2}n}}^{-1}}\,,
\]
then 
\[
\E\Bbrack{\pi_{t}(E)\sqrt{\log\tfrac{1}{\pi_{t}(E)}}\,\ind_{[\pi_{t}(E)\leq1/2]}}\geq\frac{1}{4}\,\pi(E)\sqrt{\log\tfrac{1}{\pi(E)}}\,.
\]
\end{lem}

\begin{proof}
We fix any such $E$ and denote $g_{t}:=\pi_{t}(E)$.  We first show
that 
\[
\P\bpar{\log\frac{1}{g_{t}}\geq\frac{1}{4}\log\frac{1}{g_{0}}\,,\forall t\in[0,T]}\geq1-4g_{0}^{2}\,.
\]

When $T\lesssim\tfrac{1}{D^{2}}\,\log\frac{1}{g_{0}}$, the second
part of Proposition~\ref{prop:random-measure-control} with $\gamma=10^{-2}\log\frac{1}{g_{0}}$
ensures that that with probability at least $1-4g_{0}^{1/(c\cdot10^{5})}$,
\[
0.99\log\frac{1}{g_{0}}\leq\log\frac{1}{g_{t}}\leq\bpar{\frac{c}{2}+1.01}\log\frac{1}{g_{0}}\,.
\]
Taking small enough $c$, the claim follows.

When $T\lesssim\norm{\Sigma}^{-1}\,(\log\frac{1}{g_{0}}\vee\log^{2}n)^{-1}$,
It\^o's formula leads to
\[
\D\log\log\frac{e}{g_{t}}=-\frac{1}{g_{t}\log\frac{e}{g_{t}}}\,\D g_{t}+\frac{\log\frac{1}{g_{t}}}{2g_{t}^{2}\log^{2}\frac{e}{g_{t}}}\,\D[g]_{t}\,.
\]
Recall that $g_{t}=\pi_{t}(E)=\int_{E}\D\pi_{t}(x)$, so we can write
for some $\alpha_{t}$ as follows:
\begin{align*}
\D g_{t} & =\int_{E}\inner{x-b_{t},\D B_{t}}\,\D\pi_{t}(x)=\pi_{t}(E)\int_{E}\inner{x-b_{t},\D B_{t}}\,\frac{\D\pi_{t}(x)}{\pi_{t}(E)}=:g_{t}\,\inner{\alpha_{t},\D B_{t}}\,,\\
\D[g]_{t} & =g_{t}^{2}\,\norm{\alpha_{t}}^{2}\,\D t\,.
\end{align*}
Substituting these back,
\begin{equation}
\D\log\log\frac{e}{g_{t}}=-\frac{\inner{\alpha_{t},\D B_{t}}}{\log\frac{e}{g_{t}}}+\frac{\log\frac{1}{g_{t}}}{2\log^{2}\frac{e}{g_{t}}}\,\norm{\alpha_{t}}^{2}\,\D t\geq-\frac{\inner{\alpha_{t},\D B_{t}}}{\log\frac{e}{g_{t}}}=:\D M_{t}\,.\label{eq:ito-loglog}
\end{equation}
It readily follows that $\D[M]_{t}=\norm{\alpha_{t}}^{2}\log^{-2}\frac{e}{g_{t}}\,\D t=g_{t}^{-2}\log^{-2}\frac{e}{g_{t}}\,\D[g]_{t}$.
By the first part of Proposition~\ref{prop:random-measure-control},
\[
\D[M]_{t}\leq30\,\norm{\Sigma_{t}}\,\D t\,.
\]
Let $B$ be a good event defined as
\[
\underbrace{\max_{t\in[0,T]}\norm{\Sigma_{t}}\leq2\,\norm{\Sigma}}_{=:B_{1}}\,,\quad\&\quad\underbrace{\inf_{t\in[0,T]}M_{t}\geq-12\,\sqrt{\norm{\Sigma}\,T\log\tfrac{1}{g_{0}}}}_{=:B_{2}}\,.
\]
Under $B$, integrating \eqref{eq:ito-loglog} and using $T\lesssim\norm{\Sigma}^{-1}\log^{-1}\frac{1}{g_{0}}$
lead to
\[
\log\log\frac{e}{g_{t}}\geq\log\log\frac{e}{g_{0}}-12\sqrt{\norm{\Sigma}\,T\log\tfrac{1}{g_{0}}}\geq\log\log\frac{e}{g_{0}}-\frac{1}{100}\,,
\]
 which implies that $\log\frac{1}{g_{t}}\geq\frac{1}{4}\,\log\frac{1}{g_{0}}$
for all $t\in[0,T]$ due to $g_{0}\leq1/16$. Hence, it suffices to
bound $\P(B^{c})$ by $4g_{0}^{2}$.

We note that
\[
\P(B^{c})=\P\bpar{B_{1}^{c}\cup(B_{1}\cap B_{2}^{c})}\leq\P(B_{1}^{c})+\P(B_{1}\cap B_{2}^{c})\,.
\]
As for the first term, since $T\lesssim\norm{\Sigma}^{-1}(\log\frac{1}{g_{0}}\vee\log^{2}n)^{-1}$,
Lemma~\ref{lem:SL-operator} ensures that for some universal constant
$C>0$ and small enough $c>0$,
\[
\P(B_{1}^{c})\leq\exp\bpar{-\frac{1}{C\norm{\Sigma}\,T}}\leq2g_{0}^{2}\,.
\]
Regarding the second term, since $[M]_{t}\leq60\,\norm{\Sigma}\,T$
under $B_{1}$, Freedman's inequality (Proposition~\ref{prop:freedman})
ensures that
\[
\P(B_{1}\cap B_{2}^{c})\leq\exp\bpar{-\frac{144\,\norm{\Sigma}\,T\log\frac{1}{g_{0}}}{60\,\norm{\Sigma}\,T}}\leq2g_{0}^{2}\,.
\]
Combining those two bounds above, we conclude that 
\[
\P\Bpar{\log\frac{1}{g_{t}}\geq\frac{1}{4}\log\frac{1}{g_{0}}\,,\ \forall t\in[0,T]}\geq\P(B)\geq1-4g_{0}^{2}\,.
\]

We now prove the main claim. Due to $g_{0}\leq\frac{1}{16}$,
\begin{align*}
\E\Bbrack{g_{t}\sqrt{\log\tfrac{1}{g_{t}}}\,\ind_{[g_{t}\leq1/2]}} & \geq\E\Bbrack{g_{t}\sqrt{\log\tfrac{1}{g_{t}}}\,\ind_{[\log\frac{1}{g_{t}}\geq\frac{1}{4}\,\log\frac{1}{g_{0}}]}}\geq\frac{1}{2}\sqrt{\log\tfrac{1}{g_{0}}}\,\E\bbrack{g_{t}\,\ind_{[\log\frac{1}{g_{t}}\geq\frac{1}{4}\,\log\frac{1}{g_{0}}]}}\\
 & \underset{(i)}{\geq}\frac{1}{2}\sqrt{\log\tfrac{1}{g_{0}}}\,\Bpar{g_{0}-\P\bpar{\log\frac{1}{g_{t}}<\frac{1}{4}\,\log\frac{1}{g_{0}}}}\underset{(ii)}{\geq}\frac{1}{4}\,g_{0}\sqrt{\log\tfrac{1}{g_{0}}}\,,
\end{align*}
where we used $\E g_{t}=g_{0}$ and $g_{t}\leq1$ in $(i)$ , and
used the first claim in $(ii)$.
\end{proof}
 We now prove the main theorem of this section.
\begin{proof}
[Proof of Theorem~\ref{thm:lsi-general-bound}] Let $E$ be any
measurable subset of measure $\pi(E)\leq1/2$. If $\pi(E)\geq1/16$,
then $\cpi(\pi)\lesssim D^{2}\wedge\norm{\Sigma}\log n$ and equivalence
of the Cheeger constants and \eqref{eq:pi} (due to the Buser--Ledoux
inequality) implies that
\[
\pi(\de E)\gtrsim\frac{\pi(E)}{\sqrt{D^{2}\wedge\norm{\Sigma}\log n}}\,.
\]
Moreover,
\[
\pi(E)\sqrt{\log\tfrac{1}{\pi(E)}}\lesssim\pi(E)\,.
\]
Hence,
\[
\frac{\pi(\de E)}{\pi(E)\sqrt{\log\frac{1}{\pi(E)}}}\gtrsim\frac{1}{\sqrt{D^{2}\wedge\norm{\Sigma}\log n}}\,,
\]
and the claim immediately follows.

When $\pi(E)\leq1/16$, we use the SL process $(\pi_{t})_{t\in[0,T]}$
described above, obtaining that 
\[
\pi(\de E)\gtrsim\sqrt{T}\,\E\Bbrack{\pi_{T}(E)\sqrt{\log\tfrac{1}{\pi_{T}(E)}}\,\ind_{[\pi_{T}(E)\leq1/2]}}\,.
\]
By Lemma~\ref{lem:random-measure-control}, if
\[
T\asymp\max\Bbrace{\frac{1}{D^{2}}\,\log\frac{1}{\pi(E)},\Bpar{\norm{\Sigma}\,\bpar{\log\frac{1}{\pi(E)}\vee\log^{2}n}}^{-1}}\,,
\]
then $\pi(\de E)\gtrsim\sqrt{T}\,\pi(E)\log^{1/2}\frac{1}{\pi(E)}$.
When $g_{0}=\log\frac{1}{\pi(E)}(\geq\log16)$ is less than $\log^{2}n$,
we have
\[
T\gtrsim\frac{g_{0}}{D^{2}}+\frac{1}{\norm{\Sigma}\log^{2}n}\geq\frac{1}{D^{2}}+\frac{1}{\norm{\Sigma}\log^{2}n}\geq\frac{1}{D^{2}\wedge\norm{\Sigma}\log^{2}n}\,.
\]
When $g_{0}$ is larger than $\log^{2}n$, it follows from the AM-GM
inequality that
\[
T\gtrsim\frac{g_{0}}{D^{2}}+\frac{1}{\norm{\Sigma}\,g_{0}}\geq\frac{2}{D\,\norm{\Sigma}^{1/2}}\,.
\]
Therefore, 
\[
T\gtrsim\max\bbrace{D\,\norm{\Sigma}^{1/2},D^{2}\wedge\norm{\Sigma}\log^{2}n}^{-1}\,.
\]
Since $\clsi(\pi)\asymp\clch^{-2}(\pi)\lesssim T^{-1}$, we obtain
that 
\[
\clsi(\pi)\lesssim\max\bbrace{D\,\norm{\Sigma}^{1/2},D^{2}\wedge\norm{\Sigma}\log^{2}n}\,.
\]
Since $a\wedge b\leq\sqrt{ab}$, we can further bound the second term
by $D\,\norm{\Sigma}^{1/2}\log n$.
\end{proof}

\end{document}